\documentclass[a4paper, 11pt]{article} \pdfoutput=1

\usepackage{amsmath,amssymb,amsthm,bbm,bm, amscd, mathtools, ascmac}
\usepackage{longtable, geometry}
\usepackage{enumitem}
\usepackage{latexsym}
\usepackage{braket}
\usepackage{mathrsfs}
\usepackage{authblk}
\usepackage{graphicx}
\usepackage{physics}
\usepackage{mathtools}
\usepackage{empheq}
\usepackage{stmaryrd}
\usepackage[utf8]{inputenc}
\usepackage[english]{babel}
\usepackage{esint}

\usepackage[subrefformat=parens]{subcaption}
\usepackage[svgnames,dvipsnames]{xcolor}
\usepackage[colorlinks,citecolor=DarkGreen,linkcolor=FireBrick,linktocpage,unicode]{hyperref}
\usepackage{xparse}
\usepackage{tikz}
\usetikzlibrary{arrows.meta,positioning,calc}

\usepackage{tikz-cd}
\usetikzlibrary{intersections,calc,arrows}
\usetikzlibrary{decorations.pathreplacing}
\usetikzlibrary{positioning}
\usepackage{url}
\hypersetup{urlcolor=blue}
\usepackage{comment}

\newtheorem{thm}{Theorem}[section]
\newtheorem{prop}[thm]{Proposition}
\newtheorem{lem}[thm]{Lemma}
\newtheorem{cor}[thm]{Corollary}
\theoremstyle{definition}

\newtheorem{rem}[thm]{Remark}

\newcommand{\be}{\begin{equation}}
\newcommand{\ee}{\end{equation}}

\newcommand{\im}{\mathrm{i}\,}

\numberwithin{equation}{section}

\newcommand{\cH}{\mathcal{H}}

\newcommand{\e}{e}

\newcommand{\ad}{\mathrm{ad}}

\author[a]{Tadahiro Miyao}

\title{
\textbf{
Hubbard--Heisenberg Thermodynamic Comparison at Half Filling in a Fixed Staggered Field
}
}

\affil[a]{Department of Mathematics,  Hokkaido University\\

E-mail: \text{miyao@math.sci.hokudai.ac.jp}}

\theoremstyle{plain}

\begin{document}
\maketitle

\begin{abstract}
We study the repulsive Hubbard model at half filling
in the canonical ensemble and compare its thermodynamics
with the spin-$\frac12$ antiferromagnetic Heisenberg model
at superexchange coupling $J_0(U)=4t^2/U$.
For fixed hopping $t\ne0$, bounded staggered fields
$|h|\le h_0$, and inverse temperatures
$\beta J_0(U)\ge\ell_0>0$, we establish a volume-uniform
pressure comparison with error $O(U^{-3}+|h|U^{-2})$.
On the Heisenberg field scale $h=J_0(U)x$, the error
improves to $O(U^{-3})$ for bounded $x$.

Using the quantum Lee--Yang theorem of Suzuki and Fisher,
we obtain a volume-uniform $O(U^{-1})$ comparison of
staggered magnetizations on compact positive intervals
of the scaled field $x$.
By combining a recent fermionic Lee--Yang theorem
with analytic interpolation, we improve this estimate
to $O(U^{-2}\log U)$ in the thermodynamic limit.
The latter analysis also establishes analyticity of
the limiting Hubbard pressure at nonzero real fields.
The improved comparison extends to vanishing positive
fields satisfying $U^{-3}\ll h(U)\ll U^{-1}$,
with error
$O((1+\log(U^3h(U)))/(U^3h(U)))$.
Consequently, sufficiently low-temperature Heisenberg
N\'eel order yields a quantitative positive lower bound
on the Hubbard staggered magnetization in this regime.

Furthermore, we prove $O(U^{-2})$ bounds on the
double-occupancy and charge-defect densities and the
normalized charge-density-wave second moment.
At zero staggered field, this decay rate is optimal
for both densities, even in the thermodynamic limit.
Our approach uses finite-order Schrieffer--Wolff
transformations and does not establish spontaneous
N\'eel order at fixed finite $U$.
\end{abstract}

\tableofcontents

\section{Introduction and main results}
\label{sec:intro-main-results}

\subsection{Background and purpose}
\label{subsec:background-purpose}

The Hubbard model, introduced to describe electron
correlations in narrow bands \cite{Hubbard1963},
is a fundamental model of strongly correlated
electrons, Mott physics, and quantum magnetism.
Despite the exact one-dimensional solution of
Lieb and Wu \cite{LiebWu1968} and significant
rigorous advances in the study of its magnetic
properties
\cite{Nagaoka1966,Lieb1989,MielkeTasaki1993,Tasaki1998},
the interplay of electron itinerancy, spin degrees
of freedom, and Coulomb repulsion continues to pose
substantial mathematical challenges.

In the repulsive Hubbard model at half filling,
strong Coulomb repulsion suppresses charge
fluctuations, while virtual hopping processes
generate an antiferromagnetic spin interaction
at the superexchange scale
\[
 J_0(U):=\frac{4t^2}{U}.
\]
This picture is suggested by formal strong-coupling
expansions
\cite{HarrisLange1967,MacDonaldGirvinYoshioka1988}.
Further evidence comes from the
Marshall--Lieb--Mattis theorem for the
antiferromagnetic Heisenberg model and Lieb's
theorem for the half-filled Hubbard model,
which establish the same ground-state total spin
on bipartite lattices
\cite{LiebMattis1962,Lieb1989};
see also \cite{Miyao2019Ferromagnetism}
for the related positivity structures.
These results suggest a close connection
between the two models, but do not provide
quantitative thermodynamic comparisons.
The purpose of this work is to establish such
comparisons with rigorous error bounds, uniformly
in the volume, in a low-temperature regime
that resolves the spin-exchange scale,
and for bounded staggered magnetic fields.

Our analysis is carried out in the canonical
half-filled ensemble
\[
 \cH_\Lambda^{\rm hf}
 =\ker(N_\Lambda-|\Lambda|).
\]
This is a natural setting for comparison with
the effective Heisenberg model, which describes
the spin degrees of freedom associated with
singly occupied sites.

Effective-Hamiltonian methods have played a central
role in rigorous studies of quantum lattice systems,
notably in the quantum Pirogov--Sinai programme
\cite{DattaFernandezFrohlich1996,
DattaFernandezFrohlich1999}
and in investigations of charge-ordered phases
of extended Hubbard-type models
\cite{BorgsJedrzejewskiKotecky1996,
BorgsKoteckyUeltschi1996,BorgsKotecky2000,
Jedrzejewski1994}.
These works provide important precedents for the
rigorous use of effective Hamiltonians.
However, they are primarily concerned with phase
structure and are typically formulated in the
grand-canonical ensemble.

In contrast, our objective is a quantitative
comparison of Gibbs pressures under the exact
canonical half-filling constraint.
This constraint introduces an additional difficulty:
charge defects must occur in balanced doublon--hole
configurations, and their contributions must be
controlled at fixed total particle number.
Consequently, the pressure comparison established
here does not follow directly from the existing
quantum Pirogov--Sinai framework.

The separation between the charge scale $U$ and
the spin scale $J_0(U)$ is essential.
For fixed $t\ne0$ and inverse temperatures satisfying
\[
 \beta J_0(U)\ge\ell_0>0,
\]
we prove a volume-uniform pressure comparison
with error $O(U^{-3}+|h|U^{-2})$ for bounded
staggered fields $h$.
In particular, for $h=J_0(U)x$ with bounded $x$,
the error is $O(U^{-3})$.
This accuracy allows us to obtain quantitative
information about magnetic response on the
Heisenberg exchange scale.
We also establish $O(U^{-2})$ bounds on the
double-occupancy and charge-defect densities
and on the normalized charge-density-wave
second moment.
At zero staggered field, this decay rate is optimal
for the former two densities, even in the thermodynamic
limit.

A further objective is to transfer the pressure
comparison to the staggered magnetization.
The Lee--Yang theorem has long provided important
information about magnetic response and analyticity
\cite{LeeYang1952,Baker1968,FrohlichRodriguez2012}.
Its quantum Heisenberg version, established by
Suzuki and Fisher \cite{SuzukiFisher1971},
plays a key role in our finite-volume analysis.
After a sublattice rotation, this theorem applies
to the antiferromagnetic Heisenberg model.
The resulting zero-location property yields,
through Hadamard factorization, a susceptibility
bound uniform in the volume and inverse temperature.
Combined with our pressure comparison and convexity,
this gives a volume-uniform $O(U^{-1})$ comparison
of staggered magnetizations on compact positive
intervals of the scaled field $x=h/J_0(U)$.

A sharper estimate is obtained in the thermodynamic
limit using the recent fermionic Lee--Yang theorem
of Rayudu et al.~\cite{RMZT26}.
After a suitable particle--hole transformation,
this theorem applies to the grand-canonical
Hubbard model in a staggered magnetic field.
An $\eta$-spin symmetry argument identifies its
limiting pressure with the canonical half-filled
Hubbard pressure and establishes analyticity
at nonzero real fields.
Combining the Lee--Yang analyticity of both models
with a Bernstein--Chebyshev interpolation estimate,
we improve the magnetization comparison to
$O(U^{-2}\log U)$ on compact positive intervals
of the scaled field.
The improved estimate also applies to vanishing
positive fields satisfying
\[
 U^{-3}\ll h(U)\ll U^{-1},
\]
with an error of order
\[
 O\left(
 \frac{1+\log(U^3h(U))}{U^3h(U)}
 \right).
\]
Unlike the finite-volume comparison, these sharper
estimates are obtained after taking the
thermodynamic limit.

These results connect with rigorous studies of
N\'eel order in antiferromagnetic Heisenberg systems
\cite{DysonLiebSimon1978,KuboKishi1988,
BjornbergUeltschi2022,KomaTasaki1993}.
When the reference Heisenberg model has a uniformly
positive spontaneous staggered magnetization $m_0$
at sufficiently low temperatures, our comparison
gives the quantitative Hubbard lower bound
\[
 m_{\beta,U}^{\rm Hub}(h(U))
 \ge m_0-
 O\left(
 \frac{1+\log(U^3h(U))}{U^3h(U)}
 \right).
\]
In particular, the Hubbard magnetization remains
positive in the joint strong-coupling and
vanishing-field regime.
This does not, however, establish spontaneous
N\'eel order at any fixed finite $U$.

Our proof combines finite-order Schrieffer--Wolff
transformations with volume-uniform estimates
of the contributions from charge-defect sectors
to the canonical partition function.

\subsection{Hubbard Hamiltonian at half filling}\label{subsec:model}

\paragraph{Fermions and the half-filled sector.}
Fix \(d\ge1\) and let
$
  \Lambda_L:=\bigl(\mathbb Z/L\mathbb Z\bigr)^d
  $
  with 
  $
   L\in2\mathbb N.
$
Set \(\mathfrak h_{\Lambda_L}:=\ell^2(\Lambda_L)\otimes\mathbb C^2\) and
\[
  \cH_{\Lambda_L}
  :=
  \bigoplus_{n=0}^{2|\Lambda_L|}
  \bigwedge\nolimits^{n}\mathfrak h_{\Lambda_L} .
\]
For \(x\in\Lambda_L\) and \(\sigma\in\{\uparrow,\downarrow\}\), let
\(c_{x\sigma},c_{x\sigma}^\ast\) be the fermionic annihilation and
creation operators on \(\cH_{\Lambda_L}\), satisfying the CAR:
\[
  \{c_{x\sigma},c_{y\tau}^\ast\}
  =
  \delta_{xy}\delta_{\sigma\tau},
  \qquad
  \{c_{x\sigma},c_{y\tau}\}=0,
  \qquad
  \{c_{x\sigma}^\ast,c_{y\tau}^\ast\}=0.
\]
Define
\[
  n_{x\sigma}:=c_{x\sigma}^\ast c_{x\sigma},
  \qquad
  n_x:=n_{x\uparrow}+n_{x\downarrow},
  \qquad
  N_{\Lambda_L}:=\sum_{x\in\Lambda_L} n_x .
\]
We work in the half-filled sector
\[
  \cH_{\Lambda_L}^{\rm hf}
  :=
  \ker\bigl(N_{\Lambda_L}-|\Lambda_L|\bigr)
  =
  \bigwedge\nolimits^{|\Lambda_L|}\mathfrak h_{\Lambda_L} .
\]

\paragraph{Hamiltonian.}
The Hubbard Hamiltonian at half filling is given by
\[
  H^{\rm Hub}_{\Lambda_L}
  :=
  UD_{\Lambda_L}+T_{\Lambda_L},
\]
where \(U>0\) is the on-site Coulomb repulsion,
\[
  D_{\Lambda_L}
  :=
  \sum_{x\in\Lambda_L} n_{x\uparrow}n_{x\downarrow},
\]
and \(T_{\Lambda_L}\) is the nearest-neighbour hopping term.  Let
\(\mathscr B_{\Lambda_L}\) denote the set of unordered nearest-neighbour
bonds in the torus \(\Lambda_L\), each bond being counted once.  Then
\[
  T_{\Lambda_L}
  :=
  -t
  \sum_{\{x,y\}\in\mathscr B_{\Lambda_L}}
  \sum_{\sigma=\uparrow,\downarrow}
  \bigl(
    c_{x\sigma}^\ast c_{y\sigma}
    +
    c_{y\sigma}^\ast c_{x\sigma}
  \bigr),
  \qquad
  t\in\mathbb R\setminus \{0\}.
\]
Define the staggering
\[
  \eta_x:=(-1)^{x_1+\cdots+x_d},
\]
which is well-defined on the torus \((\mathbb Z/L\mathbb Z)^d\) since
\(L\) is even, and set
\[
  S_x^{(3)}
  :=
  \frac12\bigl(n_{x\uparrow}-n_{x\downarrow}\bigr),
  \qquad
  M_{\Lambda_L}
  :=
  \sum_{x\in\Lambda_L}\eta_x S_x^{(3)} .
\]
For \(h\in\mathbb R\), we consider the staggered-field Hamiltonian on
\(\cH_{\Lambda_L}^{\rm hf}\),
\[
  H^{\rm Hub}_{\Lambda_L}(h)
  :=
  H^{\rm Hub}_{\Lambda_L}-hM_{\Lambda_L}.
\]

Since its introduction, the Hubbard model has been
an important subject of rigorous mathematical analysis.
More broadly, interacting lattice fermion systems
continue to pose challenging mathematical problems,
with ongoing developments in equilibrium properties,
response theory, and topological phenomena;
see, e.g.,
\cite{Lieb1989,GiulianiMastropietro2010,
Teufel2020,WesleMarcelliMiyaoMonacoTeufel2025,
BachmannGollerPorta2025}.

\subsection{Reference Heisenberg model and pressures}
\label{subsec:ref-heis-pressure}

Throughout this subsection, for notational simplicity, we write
$
  \Lambda=\Lambda_L.
$

\paragraph{The reference Heisenberg model.}
Let
\[
  \cH^{\rm spin}_\Lambda
  :=
  \bigotimes_{x\in\Lambda}\mathbb C^2
\]
with the tensor factors ordered in any fixed way. Let
\(\sigma^{(1)},\sigma^{(2)},\sigma^{(3)}\) be the Pauli matrices:
\[
  \sigma^{(1)}
  =
  \begin{pmatrix}
    0&1\\
    1&0
  \end{pmatrix},
  \qquad
  \sigma^{(2)}
  =
  \begin{pmatrix}
    0&-\im\\
    \im&0
  \end{pmatrix},
  \qquad
  \sigma^{(3)}
  =
  \begin{pmatrix}
    1&0\\
    0&-1
  \end{pmatrix}.
\]
For each \(x\in\Lambda\) and \(a\in\{1,2,3\}\), define the local spin
operator \(S_x^{(a)}\) on \(\cH_\Lambda^{\rm spin}\), acting as
\(\frac12\sigma^{(a)}\) on the tensor factor at \(x\) and as the identity
on all other tensor factors.  Set
\[
  \boldsymbol S_x
  :=
  \bigl(S_x^{(1)},S_x^{(2)},S_x^{(3)}\bigr).
\]
This notation is chosen to match the corresponding fermionic spin
operators introduced above.  When both the fermionic Hilbert space and
the reference spin Hilbert space appear in the same formula, the ambient
space or the identifying map will be stated explicitly.
Let
\[
  M_\Lambda^{\rm spin}
  :=
  \sum_{x\in\Lambda}\eta_x S_x^{(3)}.
\]
For \(J>0\), define
\[
  H^{\rm Heis}_\Lambda(J)
  :=
  J\sum_{\{x,y\}\in\mathscr B_\Lambda}
  \left(
    \boldsymbol S_x\cdot\boldsymbol S_y-\frac14
  \right),
  \qquad
  \boldsymbol S_x\cdot\boldsymbol S_y
  :=
  \sum_{a=1}^3 S_x^{(a)}S_y^{(a)}.
\]
With staggered field \(h\in\mathbb R\), set
\[
  H^{\rm Heis}_\Lambda(J,h)
  :=
  H^{\rm Heis}_\Lambda(J)-hM_\Lambda^{\rm spin}.
\]

Despite its long history and central role in quantum
magnetism, the Heisenberg model still poses many
mathematical challenges.
For some notable rigorous results, see, e.g.,
\cite{LiebMattis1962,DysonLiebSimon1978,
KuboKishi1988,KomaTasaki1993,
CorreggiGiulianiSeiringer2015}.

\paragraph{Hubbard pressure and Gibbs state.}
For \(\beta>0\), \(U>0\), and \(h\in\mathbb R\), define the finite-volume
Hubbard partition function by
\[
  Z_{\Lambda,\beta,U}^{\rm Hub}(h)
  :=
  \Tr_{\cH_\Lambda^{\rm hf}}
  e^{-\beta H_\Lambda^{\rm Hub}(h)}.
\]
The corresponding pressure is
\[
  p_{\Lambda,\beta,U}^{\rm Hub}(h)
  :=
  \frac1{\beta|\Lambda|}
  \log Z_{\Lambda,\beta,U}^{\rm Hub}(h).
\]
For a Hilbert space $\mathcal H$, let
$\mathcal B(\mathcal H)$ denote the algebra of
bounded linear operators on $\mathcal H$.
For \(O\in\mathcal B(\cH_\Lambda^{\rm hf})\), the Hubbard Gibbs state is
\[
  \omega_{\Lambda,\beta,U,h}^{\rm Hub}(O)
  :=
  \frac{
    \Tr_{\cH_\Lambda^{\rm hf}}
    \left(
      Oe^{-\beta H_\Lambda^{\rm Hub}(h)}
    \right)
  }{
    Z_{\Lambda,\beta,U}^{\rm Hub}(h)
  }.
\]
The finite-volume staggered magnetisation density is
\[
  m_{\Lambda,\beta,U}^{\rm Hub}(h)
  :=
  \frac1{|\Lambda|}
  \omega_{\Lambda,\beta,U,h}^{\rm Hub}(M_\Lambda).
\]
Since the volume is finite, \(h\mapsto p_{\Lambda,\beta,U}^{\rm Hub}(h)\)
is differentiable, and
\[
  \partial_h p_{\Lambda,\beta,U}^{\rm Hub}(h)
  =
  m_{\Lambda,\beta,U}^{\rm Hub}(h).
\]

\paragraph{Heisenberg pressure and Gibbs state.}
For \(J>0\), \(\beta>0\), and \(h\in\mathbb R\), define
\[
  Z_{\Lambda,\beta}^{\rm Heis}(J,h)
  :=
  \Tr_{\cH_\Lambda^{\rm spin}}
  e^{-\beta H_\Lambda^{\rm Heis}(J,h)}.
\]
The corresponding pressure is
\[
  p_{\Lambda,\beta}^{\rm Heis}(J,h)
  :=
  \frac1{\beta|\Lambda|}
  \log Z_{\Lambda,\beta}^{\rm Heis}(J,h).
\]
For \(O\in\mathcal B(\cH_\Lambda^{\rm spin})\), the Heisenberg Gibbs
state is
\[
  \omega_{\Lambda,\beta,J,h}^{\rm Heis}(O)
  :=
  \frac{
    \Tr_{\cH_\Lambda^{\rm spin}}
    \left(
      Oe^{-\beta H_\Lambda^{\rm Heis}(J,h)}
    \right)
  }{
    Z_{\Lambda,\beta}^{\rm Heis}(J,h)
  }.
\]
The finite-volume staggered magnetisation density is
\[
  m_{\Lambda,\beta}^{\rm Heis}(J,h)
  :=
  \frac1{|\Lambda|}
  \omega_{\Lambda,\beta,J,h}^{\rm Heis}(M_\Lambda^{\rm spin}).
\]
Again,
\[
  \partial_h p_{\Lambda,\beta}^{\rm Heis}(J,h)
  =
  m_{\Lambda,\beta}^{\rm Heis}(J,h).
\]

\paragraph{The effective Heisenberg reference model.}
For the Hubbard model with coupling \(U\), set
\[
  J_0(U):=\frac{4t^2}{U}.
\]
The Heisenberg model used as the reference model in the main comparison
theorems is the nearest-neighbour antiferromagnetic Heisenberg model with
coupling \(J_0(U)\) and the same staggered field \(h\):
$
  H_\Lambda^{\rm Heis}(J_0(U),h).
$
The corresponding pressure is denoted by
\[
  p_{\Lambda,\beta,U}^{\rm Heis}(h)
  :=
  p_{\Lambda,\beta}^{\rm Heis}(J_0(U),h).
\]
We also write
\[
  \omega_{\Lambda,\beta,U,h}^{\rm Heis}
  :=
  \omega_{\Lambda,\beta,J_0(U),h}^{\rm Heis}
\]
and
\[
  m_{\Lambda,\beta,U}^{\rm Heis}(h)
  :=
  m_{\Lambda,\beta}^{\rm Heis}(J_0(U),h)
  =
  \partial_h p_{\Lambda,\beta,U}^{\rm Heis}(h).
\]

\subsection{Main theorems}\label{subsec:main-theorems}

We keep $d\ge1$ and $t\in\mathbb R\setminus\{0\}$ fixed,
and let $\Lambda=\Lambda_L$ run through even tori.

In the strong-coupling regime $U\to\infty$, the effective
antiferromagnetic exchange interaction arises from
virtual hopping processes and has coupling strength
\[
 J_0(U)=\frac{4t^2}{U}.
\]
This is the characteristic energy scale of the
effective Heisenberg model.
Since $J_0(U)\to0$, identifying the effective spin
interaction requires a pressure comparison whose
error is smaller than $J_0(U)$.

We consider inverse temperatures satisfying
$\beta J_0(U)\ge\ell_0>0$, so that the thermal energy
is controlled on the Heisenberg exchange scale.
To study the magnetic response on the same scale,
we write $h=J_0(U)x$, where $x$ is a dimensionless
magnetic field. We also establish pressure estimates
on fixed field intervals $|h|\le h_0$.

All constants below are independent of $\Lambda$,
$\beta$, and the magnetic field within the stated
window.

\begin{thm}[Third-order Hubbard--Heisenberg pressure comparison]
\label{thm:pressure-comparison-U3-face}
Fix $d\ge1$, $t\ne0$, $h_0>0$, and $\ell_0>0$.
There exist constants $U_0>0$ and $C_1,C_2>0$,
depending only on $d,t,h_0,\ell_0$, such that
for every $U\ge U_0$, $|h|\le h_0$, and $\beta>0$
satisfying $\beta J_0(U)\ge\ell_0$,
\begin{equation}\label{eq:main-pressure-U3}
 \sup_\Lambda
 \left|
 p_{\Lambda,\beta,U}^{\rm Hub}(h)
 -p_{\Lambda,\beta}^{\rm Heis}(J_0(U),h)
 \right|
 \le\frac{C_1}{U^3}+\frac{C_2|h|}{U^2}.
\end{equation}
Here and below, the supremum over $\Lambda$ runs over
all finite even tori.

In particular, consider magnetic fields on the
Heisenberg exchange scale, parametrized by the
dimensionless variable $x=h/J_0(U)$.
For each $R>0$, set
\[
 U_R=\max\left\{U_0,\frac{4t^2R}{h_0}\right\},
 \qquad
 C_R=C_1+4t^2RC_2.
\]
Then, for every $U\ge U_R$ and $\beta>0$ satisfying
$\beta J_0(U)\ge\ell_0$,
\begin{equation}\label{eq:main-pressure-Heis-scale}
 \sup_\Lambda\sup_{|x|\le R}
 \left|
 p_{\Lambda,\beta,U}^{\rm Hub}(J_0(U)x)
 -p_{\Lambda,\beta}^{\rm Heis}
       (J_0(U),J_0(U)x)
 \right|
 \le\frac{C_R}{U^3}.
\end{equation}
\end{thm}

The proof is given in Proposition~\ref{prop:pressure-comparison-U3}.
Since $J_0(U)=4t^2/U$, the pressure difference is
$o(J_0(U))$ uniformly for $|h|\le h_0$.
The sharper $O(U^{-3})$ estimate on the Heisenberg
exchange scale will be used below to establish
a pointwise comparison of the magnetizations.

\begin{thm}[Quantitative comparison of staggered magnetizations]
\label{thm:magnetization-Heis-scale-face}
Fix $\ell_0>0$ and $0<a<c<\infty$, and set
$I=[a,c]$. Write $J=J_0(U)=4t^2/U$.

\smallskip
\noindent
\emph{(i) Finite-volume comparison.}
There exist constants $U_I>0$ and $C_I<\infty$,
depending only on $d,t,\ell_0,a,c$, such that
for every $U\ge U_I$ and $\beta>0$ satisfying
$\beta J\ge\ell_0$,
\begin{equation}\label{eq:main-magnetization-U1}
 \sup_\Lambda\sup_{x\in I}
 \left|
 m_{\Lambda,\beta,U}^{\rm Hub}(Jx)
 -m_{\Lambda,\beta}^{\rm Heis}(J,Jx)
 \right|
 \le\frac{C_I}{U}.
\end{equation}

\smallskip
\noindent
\emph{(ii) Improved thermodynamic-limit comparison.}
For every fixed finite $\beta,U>0$, the canonical
Hubbard pressure has a thermodynamic limit
$p_{\beta,U}^{\rm Hub}(h)$.
Likewise, the Heisenberg pressure has a
thermodynamic limit $p_\beta^{\rm Heis}(J,h)$
for every fixed finite $\beta,J>0$.

Both limiting pressures are real analytic
at every nonzero real field.
Moreover, the corresponding finite-volume
staggered magnetizations converge, and
\[
\begin{aligned}
 m_{\beta,U}^{\rm Hub}(h)
 &=\lim_{L\to\infty}
   m_{\Lambda_L,\beta,U}^{\rm Hub}(h)
   =\partial_h p_{\beta,U}^{\rm Hub}(h),\\
 m_\beta^{\rm Heis}(J,h)
 &=\lim_{L\to\infty}
   m_{\Lambda_L,\beta}^{\rm Heis}(J,h)
   =\partial_h p_\beta^{\rm Heis}(J,h),
 \qquad h\ne0.
\end{aligned}
\]

Furthermore, upon increasing $U_I$ if necessary,
there exists $\widetilde C_I<\infty$,
depending only on $d,t,\ell_0,a,c$, such that
for all $U\ge U_I$,
\begin{equation}\label{eq:main-magnetization-U2log}
 \sup_{\substack{b\ge\ell_0\\x\in I}}
 \left|
 m_{b/J,U}^{\rm Hub}(Jx)
 -m_{b/J}^{\rm Heis}(J,Jx)
 \right|
 \le\widetilde C_I
       \frac{1+\log U}{U^2}.
\end{equation}

\smallskip
\noindent
\emph{(iii) Comparison in shrinking positive fields.}
Fix $R>0$.
There exist constants $A_R,C_R',U_R'>0$,
depending only on $d,t,\ell_0,R$, such that,
for every $U\ge U_R'$, with
\[
 \varepsilon_U:=\frac{A_R}{U^2},
\]
and every $\varepsilon_U\le x\le R/2$,
\begin{equation}\label{eq:main-magnetization-shrinking}
 \sup_{b\ge\ell_0}
 \left|
 m_{b/J,U}^{\rm Hub}(Jx)
 -m_{b/J}^{\rm Heis}(J,Jx)
 \right|
 \le C_R'\frac{\varepsilon_U}{x}
 \left(1+\log\frac{x}{\varepsilon_U}\right).
\end{equation}

In particular, if $h(U)>0$ satisfies
\[
 U^3h(U)\longrightarrow\infty,
 \qquad
 Uh(U)\longrightarrow0,
\]
then
\begin{equation}\label{eq:main-magnetization-shrinking-physical}
 \sup_{b\ge\ell_0}
 \left|
 m_{b/J,U}^{\rm Hub}(h(U))
 -m_{b/J}^{\rm Heis}(J,h(U))
 \right|
 =
 O\left(
 \frac{1+\log(U^3h(U))}{U^3h(U)}
 \right).
\end{equation}
\end{thm}

The finite-volume estimate \eqref{eq:main-magnetization-U1}
is proved in
Subsection~\ref{subsec:pointwise-magnetization-Lee-Yang},
using the third-order pressure comparison and the
Lee--Yang theorem of Suzuki and Fisher
\cite{SuzukiFisher1971}.

The improved infinite-volume estimates
\eqref{eq:main-magnetization-U2log} and
\eqref{eq:main-magnetization-shrinking}
are proved in Subsection~\ref{subsec:FLY-bounds}.
The proof combines the fermionic Lee--Yang theorem
of Rayudu et al.~\cite{RMZT26} with the
canonical--grand-canonical comparison developed in
Subsection~\ref{subsec:FLY-ensemble}
and the analytic interpolation estimate established in
Subsection~\ref{subsec:FLY-interpolation}.
This yields a logarithmic loss instead of the
square-root loss arising from convexity alone.

The improved magnetization comparison yields a quantitative
lower bound on the Hubbard staggered magnetization in
vanishing positive fields.
In particular, the Hubbard magnetization approaches
the uniform lower bound for the spontaneous staggered
magnetization of the Heisenberg model as $U\to\infty$.

\begin{cor}[Positive Hubbard magnetization in vanishing fields]
\label{cor:positive-Hubbard-magnetization-face}
Fix $\ell_0>0$.
For $\widetilde\beta>0$, let
\[
 m_*^{\rm Heis}(\widetilde\beta)
 :=
 \lim_{\varepsilon\downarrow0}
 \lim_{L\to\infty}
 m_{\Lambda_L,\widetilde\beta}^{\rm Heis}
 (1,\varepsilon)
\]
denote the spontaneous staggered magnetization
of the Heisenberg model with unit exchange coupling
at inverse temperature $\widetilde\beta$.
Suppose that
\[
 m_0:=
 \inf_{\widetilde\beta\ge\ell_0}
 m_*^{\rm Heis}(\widetilde\beta)>0.
\]

Let $h(U)>0$ satisfy
\[
 U^3h(U)\longrightarrow\infty,
 \qquad
 Uh(U)\longrightarrow0.
\]
Then there exists $C<\infty$, depending only on
$d,t,\ell_0$, such that, for all sufficiently large
$U$ and every $\beta>0$ satisfying
$\beta J_0(U)\ge\ell_0$,
\begin{equation}\label{eq:positive-Hubbard-magnetization}
 m_{\beta,U}^{\rm Hub}(h(U))
 \ge
 m_0-
 C\frac{1+\log(U^3h(U))}{U^3h(U)}.
\end{equation}
In particular,
\[
 m_{\beta,U}^{\rm Hub}(h(U))
 \ge\frac{m_0}{2}>0
\]
for all sufficiently large $U$, uniformly in
$\beta J_0(U)\ge\ell_0$.
\end{cor}

The proof of Corollary~\ref{cor:positive-Hubbard-magnetization-face}
is given in Subsection~\ref{subsec:FLY-bounds}.

The assumption $m_0>0$ holds for the spin-$1/2$
nearest-neighbour Heisenberg antiferromagnet in three
dimensions when $\ell_0$ is sufficiently large.
The pioneering work of Dyson, Lieb, and Simon
established N\'eel long-range order in suitable
antiferromagnetic Heisenberg systems
\cite{DysonLiebSimon1978}, and Kubo and Kishi
extended the result to spin $1/2$ in three dimensions
\cite{KuboKishi1988}.
More precisely, the quantitative infrared bounds
\cite[Theorem~3.2 and condition~(3.9)]{BjornbergUeltschi2022}
yield a strictly positive lower bound on the normalized
zero-field staggered second moment, uniformly for all
sufficiently large inverse temperatures.
In matching the Hamiltonian conventions, the reference
counts each nearest-neighbour pair twice, whereas our
Hamiltonian counts each unordered bond once; thus its
inverse temperature corresponds to $\widetilde\beta/2$
for our unit-exchange normalization, up to an irrelevant
additive energy constant.
The symmetry-breaking inequality of Koma and Tasaki
\cite[Corollary~2.2]{KomaTasaki1993} then converts this
uniform long-range-order bound into
\[
 \inf_{\widetilde\beta\ge\ell_0}
 m_*^{\rm Heis}(\widetilde\beta)>0
\]
for all sufficiently large $\ell_0$.

\begin{rem}[Vanishing fields and spontaneous magnetization]
\label{rem:vanishing-fields-order}
Corollary~\ref{cor:positive-Hubbard-magnetization-face}
concerns positive magnetic fields $h(U)$ that vanish
as $U\to\infty$. It does not establish spontaneous
staggered magnetization at any fixed finite $U$.
Indeed, the present magnetization comparison is
not uniform as $h\downarrow0$ with $U$ fixed.
Establishing spontaneous magnetization would require
additional control of the zero-field limit.
\end{rem}

\paragraph{Charge-sector consequences.}
In the strong-coupling regime at half filling,
charge fluctuations are expected to be suppressed,
leaving spin degrees of freedom at low energies
\cite{HarrisLange1967,MacDonaldGirvinYoshioka1988}.
The following theorem establishes an $O(U^{-2})$
bound on the charge-defect density and shows that
macroscopic staggered charge fluctuations vanish
in the strong-coupling limit.

Define
\[
 q_x:=(n_x-1)^2,\qquad
 C_\Lambda^{\rm ch}
 :=\sum_{x\in\Lambda}\eta_x(n_x-1).
\]

\begin{thm}[Suppression of charge defects and CDW fluctuations]
\label{thm:charge-sector-suppression-face}
Fix $h_0>0$ and $\ell_0>0$.
There exist constants $U_1>0$ and $C<\infty$,
depending only on $d,t,h_0,\ell_0$, such that
for every $U\ge U_1$, $|h|\le h_0$, and
$\beta>0$ satisfying $\beta J_0(U)\ge\ell_0$,
\begin{align}\label{eq:charge-consequences}
 \sup_\Lambda\frac1{|\Lambda|}
 \sum_{x\in\Lambda}
 \omega_{\Lambda,\beta,U,h}^{\rm Hub}(q_x)
 &\le\frac{C}{U^2},\\
 \sup_\Lambda\frac1{|\Lambda|}
 \sum_{x\in\Lambda}
 \omega_{\Lambda,\beta,U,h}^{\rm Hub}
 (n_{x\uparrow}n_{x\downarrow})
 &\le\frac{C}{U^2},\\
 \sup_\Lambda
 \frac{\omega_{\Lambda,\beta,U,h}^{\rm Hub}
 ((C_\Lambda^{\rm ch})^2)}{|\Lambda|^2}
 &\le\frac{C}{U^2}.
\end{align}
Moreover, the first CDW moment vanishes identically:
\[
 \omega_{\Lambda,\beta,U,h}^{\rm Hub}
 (C_\Lambda^{\rm ch})=0
\]
for every finite even torus, $U>0$, $\beta>0$,
and $h\in\mathbb R$.

Independently of the strong-coupling and temperature
conditions above, the following bound holds for
all $U>0$, $\beta>0$, and $h\in\mathbb R$:
\begin{equation}\label{eq:charge-direct-bound}
 \sup_\Lambda\frac1{|\Lambda|}
 \sum_{x\in\Lambda}
 \omega_{\Lambda,\beta,U,h}^{\rm Hub}(q_x)
 \le\frac{2}{U+|h|}
 \left(4d|t|+\frac{\log4}{\beta}\right).
\end{equation}
\end{thm}

The $O(U^{-2})$ estimates follow from the
third-order pressure comparison and convexity
of the Hubbard pressure in $U$; see
Proposition~\ref{prop:charge-U2-from-pressure}.
At zero staggered field, this decay rate is optimal
for the doublon and charge-defect densities, even
in the thermodynamic limit, as shown in
Remark~\ref{rem:charge-U2-optimality}.
The general bound \eqref{eq:charge-direct-bound}
follows from the Gibbs variational principle,
while particle--hole symmetry at half filling
implies the exact vanishing of the first CDW moment.

The proof of Theorem~\ref{thm:charge-sector-suppression-face}
is given in Subsection~\ref{subsec:proof-charge-sector-face}.

\subsection{Zero-temperature consequences}
\label{subsec:thermodynamic-zero-temperature}

The existence of the thermodynamic pressure limits
and their analyticity at nonzero fields are addressed
in Theorem~\ref{thm:magnetization-Heis-scale-face}.
Here we discuss the zero-temperature consequences
of our volume-uniform estimates.

Let $e_{\Lambda,U}^{\rm Hub}(h)$ and
$e_\Lambda^{\rm Heis}(J,h)$ denote the ground-state
energy densities of the respective finite-volume
Hamiltonians.
Since the dimensions of their Hilbert spaces
are at most $4^{|\Lambda|}$ and $2^{|\Lambda|}$,
respectively, we have
\[
 \begin{aligned}
 0
 &\le p_{\Lambda,\beta,U}^{\rm Hub}(h)
       +e_{\Lambda,U}^{\rm Hub}(h)
       \le\frac{\log4}{\beta},\\
 0
 &\le p_{\Lambda,\beta}^{\rm Heis}(J,h)
       +e_\Lambda^{\rm Heis}(J,h)
       \le\frac{\log2}{\beta}.
 \end{aligned}
\]
These bounds are uniform in the volume.
Letting $\beta\to\infty$ in
\eqref{eq:main-pressure-U3}, we obtain
\[
 \sup_\Lambda
 \left|
 e_{\Lambda,U}^{\rm Hub}(h)
 -e_\Lambda^{\rm Heis}(J_0(U),h)
 \right|
 \le\frac{C_1}{U^3}+\frac{C_2|h|}{U^2},
 \qquad |h|\le h_0,
\]
for all sufficiently large $U$.
In particular, the comparison is of order
$O(U^{-3})$ on the Heisenberg field scale.

The uniform entropy bounds also imply the
existence of the thermodynamic limits of the
ground-state energy densities and the
interchangeability of the thermodynamic and
zero-temperature limits.
Indeed, convergence of the finite-temperature
pressures, together with the preceding bounds,
gives
\[
 \begin{aligned}
 e_U^{\rm Hub}(h)
 &:=
 \lim_{L\to\infty}e_{\Lambda_L,U}^{\rm Hub}(h)
 =-\lim_{\beta\to\infty}p_{\beta,U}^{\rm Hub}(h),\\
 e^{\rm Heis}(J,h)
 &:=
 \lim_{L\to\infty}e_{\Lambda_L}^{\rm Heis}(J,h)
 =-\lim_{\beta\to\infty}p_\beta^{\rm Heis}(J,h).
 \end{aligned}
\]
Consequently, the same third-order comparison
holds for the infinite-volume ground-state
energy densities.

At each fixed finite volume, the Gibbs states
converge as $\beta\to\infty$ to the normalized
projections onto the respective ground-state
spaces.
Thus the finite-volume magnetization comparison
\eqref{eq:main-magnetization-U1}
and the charge-defect estimates
\eqref{eq:charge-consequences}
remain valid for these ground-state averages,
without requiring uniqueness of the ground state.

Finally, the positive-field analyticity established
at finite temperature does not, by itself,
imply differentiability of the infinite-volume
ground-state energy at every nonzero field.
At points where the limiting ground-state energies
are differentiable, their negative derivatives
give the corresponding zero-temperature
magnetization densities.

\subsection{Possible extensions}
\label{subsec:possible-extensions}

The restriction to even tori with uniform nearest-neighbour
hopping is not essential to the local structure of the
argument. The finite-order Schrieffer--Wolff construction
and the defect-sector estimates are expected to extend
to other boundary conditions and to families of bipartite
graphs with uniformly bounded degree and bounded hopping
amplitudes connecting the two sublattices.
For non-uniform hopping amplitudes $t_{xy}$, the leading
effective spin Hamiltonian has antiferromagnetic exchange
couplings
\[
 J_{xy}^{(0)}=\frac{4|t_{xy}|^2}{U}.
\]
The bipartite structure remains important, since it
ensures the cancellation of the third-order contribution
to the effective spin Hamiltonian.

The Lee--Yang analysis is also compatible with
non-uniform antiferromagnetic exchange couplings.
After the sublattice rotation, the corresponding
Heisenberg Hamiltonian satisfies the conditions of
the Suzuki--Fisher theorem \cite{SuzukiFisher1971}.
Moreover, for real spin-independent hopping amplitudes
on bipartite graphs, the particle--hole transformation
used in Section~\ref{sec:fermionic-Lee-Yang}
still produces a purely pairing quadratic Hamiltonian,
to which the fermionic Lee--Yang theorem
\cite{RMZT26} applies.
The $\eta$-spin symmetry underlying the
canonical--grand-canonical comparison is also
preserved in this setting.

These observations suggest that the improved
infinite-volume magnetization comparison may extend
to more general bipartite systems.
Such an extension would require corresponding
uniform local estimates, control of the effective
field corrections, and suitable thermodynamic limits.
For complex hopping amplitudes, the preservation
of the $\eta$-spin symmetry requires additional
consideration.
Positivity of the magnetization in vanishing fields
would also require long-range order for the
corresponding Heisenberg model.

Finally, the volume-uniform pressure comparisons
yield zero-temperature energy-density bounds,
as discussed in
Subsection~\ref{subsec:thermodynamic-zero-temperature}.

\subsection{Strategy and outline}
\label{subsec:strategy}

The proof proceeds in six main steps.
First, a finite sequence of local Schrieffer--Wolff
transformations removes off-diagonal terms to the
required order, with operator-norm errors proportional
to $|\Lambda|$.
The method uses finite-order integral Taylor remainders
and does not require a convergent infinite-step
block diagonalization.

Second, the truncated interaction preserves particle
number and doublon number on each local support.
A defect-sector trace comparison then controls the
contributions of sectors containing holons and doublons.

Third, a two-site calculation identifies the singly
occupied block with the Heisenberg Hamiltonian,
including the finite-volume coordination number
and the renormalized staggered field.

Fourth, these estimates are combined to establish
the third-order pressure comparison.
Convexity in $U$ then yields the $O(U^{-2})$
charge-density bounds.
A direct Gibbs variational argument provides a
complementary general $O(U^{-1})$ bound without
strong-coupling assumptions.

Fifth, the quantum Lee--Yang theorem of Suzuki
and Fisher \cite{SuzukiFisher1971}, applied after a
sublattice rotation, gives a volume-uniform
susceptibility bound for the Heisenberg model.
Together with convexity, this converts the pressure
comparison into a finite-volume $O(U^{-1})$
magnetization comparison.

Finally, the fermionic Lee--Yang theorem of
Rayudu et al.~\cite{RMZT26} yields a sharper
comparison in the thermodynamic limit.
A particle--hole transformation allows us to apply
this theorem to the grand-canonical Hubbard model,
while the $\eta$-spin symmetry provides a direct
comparison with the canonical half-filled ensemble.
The resulting analyticity, combined with a
Bernstein--Chebyshev interpolation estimate,
improves the magnetization error to
$O(U^{-2}\log U)$ on compact positive intervals
of the scaled field.
The method also gives quantitative bounds in
vanishing positive fields.

Section~\ref{sec:grading-locality} develops the
local algebraic estimates, and
Section~\ref{sec:finite-order-SW} constructs the
finite-order normal forms.
Section~\ref{sec:defects-revision} controls the
defect sectors, while
Section~\ref{sec:H2-Heis} identifies the effective
Heisenberg Hamiltonian.
The pressure comparison and charge-density estimates
are established in
Section~\ref{sec:pressure-assembly-revision}.
Section~\ref{sec:future-consequences} proves the
finite-volume magnetization comparison, and
Section~\ref{sec:fermionic-Lee-Yang} develops the
fermionic Lee--Yang analysis and the improved
infinite-volume results.
The appendices contain the two-site calculation
and the spin compression in fixed-defect sectors.

\section{Local interactions, \texorpdfstring{$D$}{D}-grading,
and homological inverses}
\label{sec:grading-locality}

In this section, we develop the local operator estimates
needed for the finite-order Schrieffer--Wolff construction.
We first describe the charge-defect and singly occupied
sectors, and introduce local CAR interactions and their
norms. We then use the grading by double occupancy to
construct a local inverse of the adjoint action of the
unperturbed Hamiltonian, with bounds uniform in the volume
for operators of bounded support.

Related techniques involving inverses of Liouvillians
have been developed in the study of gapped lattice
systems and linear response
\cite{Teufel2020,WesleMarcelliMiyaoMonacoTeufel2025}.
Here we construct the required inverse directly
from the local $D$-grading.

\subsection{Defect projections and the spin sector}
\label{subsec:grading-norms}

For each $x\in\Lambda$ define the charge-defect projection
\[
   q_x=(n_x-1)^2.
\]
Since
\[
   q_x=1-n_x+2n_{x\uparrow}n_{x\downarrow},
\]
the half-filling condition implies
\begin{equation}\label{eq:spec-qsum-2D}
   \sum_{x\in\Lambda}q_x=2D_\Lambda.
\end{equation}
Define
\begin{equation}\label{eq:def-PQ}
  P=\mathbbm 1_{\{D_\Lambda=0\}}\restriction_{\cH_\Lambda^{\rm hf}},
  \qquad Q=\mathbbm 1-P.
\end{equation}
The range of $P$ is the subspace with exactly one electron
per site.

\subsection{Even local interactions and their norms}

For $X\subset\Lambda$, let $\mathfrak A_X$ denote the
local CAR algebra generated by
$c_{x\sigma}$ and $c_{x\sigma}^\ast$,
$x\in X$ and $\sigma\in\{\uparrow,\downarrow\}$,
acting on the fermionic Fock space.
The fermion-parity automorphism $\Theta_X$ of
$\mathfrak A_X$ is defined by
\[
 \Theta_X(c_{x\sigma})=-c_{x\sigma},
 \qquad
 \Theta_X(c_{x\sigma}^\ast)=-c_{x\sigma}^\ast.
\]
Its fixed-point subalgebra
\[
 \mathfrak A_X^{\rm even}
 :=
 \{A\in\mathfrak A_X:\Theta_X(A)=A\}
\]
is called the even subalgebra.
It consists of linear combinations of monomials
containing an even number of creation and annihilation
operators. By the CAR, even local operators commute
with all operators supported on disjoint sets.
In particular,
\[
 [A_X,B_Y]=0
 \quad\text{if }X\cap Y=\varnothing,\quad
 A_X\in\mathfrak A_X^{\rm even},\
 B_Y\in\mathfrak A_Y.
\]
For the algebraic structure and locality properties
of CAR algebras, see
\cite{BratteliRobinson2,ArakiMoriya2003}.

Interactions are families
$A=(A_X)_{\varnothing\ne X\subset\Lambda}$
with $A_X\in\mathfrak A_X^{\rm even}$.
Self-adjointness of the local terms $A_X$
is not required in this definition.
We identify interactions with their finite-volume sums
$A_\Lambda=\sum_X A_X$ whenever appropriate.

Here and below, $\sum_X$ runs over all nonempty
subsets $X\subset\Lambda$, while $\sum_{X\ni x}$
is restricted to those containing the site $x$.

For $\kappa\ge0$, define
\begin{equation}\label{eq:spec-kappa-norm}
  \|A\|_\kappa
  :=\sup_{x\in\Lambda}\sum_{X\ni x}
       e^{\kappa|X|}\|A_X\|.
\end{equation}
The interaction commutator is defined termwise by
\begin{equation}\label{eq:spec-ad-interaction-def}
  ([A,B])_Z
  :=\sum_{X\cup Y=Z}[A_X,B_Y].
\end{equation}
Because even local operators on disjoint supports commute,
the sum may be restricted to $X\cap Y\ne\varnothing$.

\begin{lem}[Operator norm from the interaction norm]
\label{lem:interaction-to-operator}
For any interaction $A$,
\[
  \|A_\Lambda\|\le |\Lambda|\|A\|_0.
\]
\end{lem}

\begin{proof}
Indeed,
\[
 \left\|\sum_XA_X\right\|
 \le\sum_X\|A_X\|
 =\sum_{x\in\Lambda}\sum_{X\ni x}
       \frac{\|A_X\|}{|X|}
 \le |\Lambda|\|A\|_0.
\]
\end{proof}

\subsection{Local double-occupancy grading}

For $X\subset\Lambda$, set
\[
 D_X=\sum_{x\in X}n_{x\uparrow}n_{x\downarrow},
 \qquad
 N_X=\sum_{x\in X}n_x,
 \qquad
 M_X=\sum_{x\in X}\eta_xS_x^{(3)}.
\]
For an operator $B$, we write
\[
 \operatorname{ad}_B(A):=[B,A]
\]
for its adjoint action.

A local operator $A_X\in\mathfrak A_X^{\rm even}$
is said to have $D$-grade $k\in\mathbb Z$ if
\[
 \operatorname{ad}_{D_X}(A_X)=kA_X.
\]
Thus, an operator of grade $k$ changes the number
of doubly occupied sites by $k$.

Since $D_X$ has integer-valued spectrum,
every local operator admits a finite decomposition
into $D$-grades.
Let $P_{X,n}:=\mathbbm 1_{\{n\}}(D_X)$
denote the spectral projection of $D_X$
corresponding to $n\in\mathbb Z$.
For $A_X\in\mathfrak A_X^{\rm even}$, define
\[
 (A_X)^{(k)}
 :=\sum_{n\in\mathbb Z}
 P_{X,n+k}A_XP_{X,n}.
\]
Then
\[
 A_X=\sum_{k\in\mathbb Z}(A_X)^{(k)},
 \qquad
 \operatorname{ad}_{D_X}\bigl((A_X)^{(k)}\bigr)
 =k(A_X)^{(k)}.
\]
In particular,
\[
 e^{\im\theta\operatorname{ad}_{D_X}}(A_X)
 =
 \sum_{k\in\mathbb Z}
 e^{\im k\theta}(A_X)^{(k)}.
\]
Taking the Fourier coefficients, we obtain
\begin{equation}\label{eq:local-grade-Fourier}
 (A_X)^{(k)}
 =
 \frac{1}{2\pi}\int_0^{2\pi}
 e^{-\im k\theta}
 e^{\im\theta\operatorname{ad}_{D_X}}(A_X)
 \,d\theta.
\end{equation}

We define the diagonal and off-diagonal parts by
\[
 A_X^{\rm diag}:=(A_X)^{(0)},
 \qquad
 A_X^{\rm off}:=\sum_{k\ne0}(A_X)^{(k)}.
\]
For an interaction $A=(A_X)_X$, the decomposition
$A=A^{\rm diag}+A^{\rm off}$ is performed term by term.
The grade projections preserve evenness and do not
enlarge the support of the local terms.

Since
\[
 e^{\im\theta\operatorname{ad}_{D_X}}(A_X)
 =e^{\im\theta D_X}A_Xe^{-\im\theta D_X},
\]
the grade-zero projection is an average of
unitary conjugations and is therefore contractive.
Consequently,
\begin{equation}\label{eq:grade-projection-norm}
 \|A^{\rm diag}\|_\kappa\le\|A\|_\kappa,
 \qquad
 \|A^{\rm off}\|_\kappa\le2\|A\|_\kappa.
\end{equation}

For every interaction $A$ and every nonzero grade $k$,
\[
 \operatorname{ad}_{D_\Lambda}(A_\Lambda^{(k)})
 =kA_\Lambda^{(k)},
 \qquad
 D_\Lambda P=PD_\Lambda=0.
\]
Multiplying the commutator identity by $P$
on both sides gives
\[
 \begin{aligned}
 kPA_\Lambda^{(k)}P
 &=P[D_\Lambda,A_\Lambda^{(k)}]P
 =0.
 \end{aligned}
\]
Consequently,
\begin{equation}\label{eq:P-compression-grade-zero}
 PA_\Lambda^{\rm diag}P=PA_\Lambda P.
\end{equation}

Let $T=(T_e)_{e\in\mathscr B_\Lambda}$ denote the
nearest-neighbour hopping interaction, where $T_e$
is the hopping operator associated with the bond $e$.
Thus, $T_\Lambda=\sum_{e\in\mathscr B_\Lambda}T_e$.

We decompose each bond term into its zero and nonzero
$D$-grade components:
\[
 T_e=T_e^{(0)}+T_e^{\rm off},
 \qquad
 T_e^{\rm off}:=\sum_{k\ne0}T_e^{(k)}.
\]
This defines the interactions
\[
 T^{(0)}=(T_e^{(0)})_{e\in\mathscr B_\Lambda},
 \qquad
 T^{\rm off}=(T_e^{\rm off})_{e\in\mathscr B_\Lambda},
\]
so that $T=T^{(0)}+T^{\rm off}$.
The terms of $T^{(0)}$ preserve the number of
doubly occupied sites, whereas those of $T^{\rm off}$
change this number by one.
Since the grade projections do not enlarge supports,
both interactions remain supported on individual bonds.

A nonzero hopping process starting from the singly
occupied subspace creates a doublon--holon pair.
Hence
\begin{equation}\label{eq:T0-vanishes-on-P}
  T_\Lambda^{(0)}P=PT_\Lambda^{(0)}=0.
\end{equation}

\subsection{Local inverse in a bounded staggered field}
\label{subsec:local-homological-inverse}

For $X\subset\Lambda$, let
\[
 \mathfrak A_X^{\rm off}
 :=
 \bigl\{A_X^{\rm off}:
 A_X\in\mathfrak A_X^{\rm even}\bigr\}
\]
denote the linear subspace of even local operators
with vanishing zero-$D_X$-grade component.

Let
\[
 H_{0,\Lambda}(h):=UD_\Lambda-hM_\Lambda.
\]
The adjoint action $\operatorname{ad}_{H_{0,\Lambda}(h)}$
leaves $\mathfrak A_X^{\rm off}$ invariant.
Whenever its restriction to this subspace is
invertible, we denote the inverse by
\[
 \mathcal I_{h,X}
 :=
 \left(
 \operatorname{ad}_{H_{0,\Lambda}(h)}
 \restriction_{\mathfrak A_X^{\rm off}}
 \right)^{-1}.
\]
The following lemma establishes its existence
and a uniform bound.

\begin{lem}[Uniform inverse of the local adjoint action]
\label{lem:bounded-support-inverse}
Fix an integer $r\ge1$ and $h_0\ge0$.
Suppose that $U>0$ satisfies $U\ge\pi r h_0$.
Then, for every $X\subset\Lambda$ with $|X|\le r$
and every $|h|\le h_0$, the restriction of
$\operatorname{ad}_{H_{0,\Lambda}(h)}$
to $\mathfrak A_X^{\rm off}$ is invertible.

Moreover, for every $A_X\in\mathfrak A_X^{\rm off}$,
\begin{equation}\label{eq:Ih-defining-equation}
 \begin{aligned}
 &\operatorname{ad}_{H_{0,\Lambda}(h)}
   \bigl(\mathcal I_{h,X}(A_X)\bigr)=A_X,\\
 &\|\mathcal I_{h,X}(A_X)\|
   \le\frac{\pi}{U}\|A_X\|.
 \end{aligned}
\end{equation}
The inverse does not enlarge the support and
commutes with $\operatorname{ad}_{N_X}$.
If $A_X$ is self-adjoint, then
$\mathcal I_{h,X}(A_X)$ is anti-self-adjoint.
\end{lem}

\begin{proof}

For $A\in\mathfrak A_X^{\rm even}$, set
\[
 \alpha_\theta(A)
 :=e^{\im\theta\operatorname{ad}_{D_X}}(A)
 =e^{\im\theta D_X}Ae^{-\im\theta D_X},
 \qquad \theta\in\mathbb R.
\]
We define the linear map
$\mathcal C_X:\mathfrak A_X^{\rm even}
\to\mathfrak A_X^{\rm off}$ by
\[
 \mathcal C_X(A)
 :=\frac{\im}{2\pi}\int_0^{2\pi}
       (\theta-\pi)\alpha_\theta(A)\,d\theta.
\]
Indeed, using the grade decomposition
\[
 \alpha_\theta(A)
 =\sum_{k\in\mathbb Z}e^{\im k\theta}A^{(k)}
\]
and the identity
\[
 \frac{\im}{2\pi}\int_0^{2\pi}
 (\theta-\pi)e^{\im k\theta}\,d\theta
 =
 \begin{cases}
  1/k,& k\ne0,\\
  0,& k=0,
 \end{cases}
\]
we obtain
\[
 \mathcal C_X(A)
 =\sum_{k\ne0}\frac1k A^{(k)}.
\]
Consequently,
\[
 \operatorname{ad}_{D_X}\mathcal C_X(A)
 =A^{\rm off}.
\]
Moreover, since $\alpha_\theta$ is a unitary
conjugation,
\[
 \begin{aligned}
 \|\mathcal C_X(A)\|
 \le\frac1{2\pi}
 \int_0^{2\pi}|\theta-\pi|
 \|\alpha_\theta(A)\|\,d\theta=\frac{\pi}{2}\|A\|.
 \end{aligned}
\]
In particular, this bound is independent of
the support size $|X|$.

Since $[D_X,M_X]=0$,
$\operatorname{ad}_{M_X}$ preserves all
$D_X$-grades. Moreover,
\[
 \|\operatorname{ad}_{M_X}\|
 \le2\|M_X\|\le|X|.
\]
On the off-diagonal subspace, put
\[
 \mathcal I_{h,X}
 :=
 \left(
  \mathrm{id}-\frac hU\mathcal C_X
                \operatorname{ad}_{M_X}
 \right)^{-1}
 \frac{\mathcal C_X}{U}.
\]
The inverse exists by a Neumann series, since
\[
 \left\|
  \frac hU\mathcal C_X
                \operatorname{ad}_{M_X}
 \right\|
 \le\frac{\pi|h||X|}{2U}
 \le\frac12.
\]
Its norm is at most $2$, and hence
\[
 \|\mathcal I_{h,X}(A_X^{\rm off})\|
 \le\frac{\pi}{U}\|A_X^{\rm off}\|.
\]

On the off-diagonal subspace,
$\operatorname{ad}_{D_X}\mathcal C_X=1$.
Since $\operatorname{ad}_{D_X}$ and
$\operatorname{ad}_{M_X}$ commute,
\[
 \begin{aligned}
 &U\operatorname{ad}_{D_X}
 \left(
 \mathrm{id}-\frac hU\mathcal C_X
                 \operatorname{ad}_{M_X}
 \right)
 =U\operatorname{ad}_{D_X}
     -h\operatorname{ad}_{M_X}.
 \end{aligned}
\]
This proves
\[
\operatorname{ad}_{H_{0, \Lambda}(h)}
 \mathcal I_{h,X}(A_X^{\rm off})
 =A_X^{\rm off}.
\]

All maps involved act within $\mathfrak A_X$
and commute with $\operatorname{ad}_{N_X}$.
Moreover,
\[
 \mathcal C_X(A)^*
 =-\mathcal C_X(A^*),
\]
whereas
\[
 \bigl(
  \mathcal C_X\operatorname{ad}_{M_X}(A)
 \bigr)^*
 =
 \mathcal C_X\operatorname{ad}_{M_X}(A^*).
\]
The Neumann series therefore gives
\[
 \mathcal I_{h,X}(A)^*
 =-\mathcal I_{h,X}(A^*).
\]
In particular, a self-adjoint input has an
anti-self-adjoint image.
\end{proof}

For an interaction $A$ whose terms have support size
at most $r$, we define
$\mathcal I_h(A^{\rm off})$ term by term.
In particular,
\begin{equation}\label{eq:local-Ih-interaction-bound}
 \|\mathcal I_h(A^{\rm off})\|_\kappa
 \le\frac{\pi}{U}\|A^{\rm off}\|_\kappa
 \qquad (U\ge\pi r h_0,\ U>0).
\end{equation}
The required interaction version of
\eqref{eq:Ih-defining-equation} follows by
summing over $X$.

\section{Finite-order Schrieffer--Wolff normal forms}
\label{sec:finite-order-SW}

We construct a finite-order Schrieffer--Wolff normal form
by successive unitary conjugations of the Hubbard
Hamiltonian. The first transformation eliminates the
off-diagonal hopping terms to leading order, while
the second removes the off-diagonal part of the
resulting quadratic interaction.
The remaining diagonal interaction yields the leading
effective spin Hamiltonian on the singly occupied sector.
By performing only finitely many transformations,
we keep the supports of the generators uniformly bounded.

The method used in this section is based on the
Schrieffer--Wolff transformation
\cite{SchriefferWolff1966}.
Related strong-coupling expansions for the Hubbard model
were developed in
\cite{HarrisLange1967,MacDonaldGirvinYoshioka1988}.
For a systematic treatment of Schrieffer--Wolff
transformations in quantum many-body systems, see
\cite{BravyiDiVincenzoLoss2011}.
Related Lie--Schwinger block-diagonalization methods
have also been studied in
\cite{DelVecchioFrohlichPizzoRossi2021}.

Here we use only finitely many local unitary
transformations and control the remainders
by integral Taylor formulas.
This avoids the need to establish convergence
of an infinite sequence of transformations.

\subsection{The first and second transformations}

Recall that
\[
 H_\Lambda^{\rm Hub}(h)
 =H_{0,\Lambda}(h)+T_\Lambda,
 \qquad T=T^{(0)}+T^{\rm off}.
\]

We define the first generator by
\begin{equation}\label{eq:SW-S1}
 S_1:=\mathcal I_h(T^{\rm off}).
\end{equation}
By the defining property of the local inverse,
\[
 \operatorname{ad}_{S_{1,\Lambda}}
       \bigl(H_{0,\Lambda}(h)\bigr)
 =-T_\Lambda^{\rm off}.
\]
We then introduce the transformed Hamiltonian
\[
 H^{[1]}
 :=e^{\operatorname{ad}_{S_{1,\Lambda}}}
      \bigl(H_\Lambda^{\rm Hub}(h)\bigr).
\]

To identify the leading terms, we apply Taylor's
formula to the adjoint action.
Using
\[
 \operatorname{ad}_{S_{1,\Lambda}}^2
       \bigl(H_{0,\Lambda}(h)\bigr)
 =-\operatorname{ad}_{S_{1,\Lambda}}
       (T_\Lambda^{\rm off}),
\]
we obtain
\begin{equation}\label{eq:SW-first-BCH}
 \begin{aligned}
 H^{[1]}
 ={}&H_{0,\Lambda}(h)+T_\Lambda^{(0)}
 +\underbrace{
 T_\Lambda^{\rm off}
 +\operatorname{ad}_{S_{1,\Lambda}}
       (H_{0,\Lambda}(h))}_{=0}\\
 &+\operatorname{ad}_{S_{1,\Lambda}}(T_\Lambda)
 +\frac12\operatorname{ad}_{S_{1,\Lambda}}^2
       (H_{0,\Lambda}(h))
 +R_{3,\Lambda}\\
 ={}&H_{0,\Lambda}(h)+T_\Lambda^{(0)}
 +\operatorname{ad}_{S_{1,\Lambda}}
       (T_\Lambda^{(0)})
 +\frac12\operatorname{ad}_{S_{1,\Lambda}}
       (T_\Lambda^{\rm off})
 +R_{3,\Lambda}.
 \end{aligned}
\end{equation}
Here $R_{3,\Lambda}$ collects the remaining
terms in the Taylor expansion.
Its exact integral representation, which is
convenient for operator-norm estimates,
will be derived in \eqref{eq:SW-R3-integral}.

This leads to the second-order interaction
\begin{equation}\label{eq:SW-F2}
 F_2:=\operatorname{ad}_{S_1}(T^{(0)})
       +\frac12\operatorname{ad}_{S_1}(T^{\rm off}).
\end{equation}
The commutators in \eqref{eq:SW-F2} are understood
as interaction commutators defined in
\eqref{eq:spec-ad-interaction-def}.
Since $(S_1)_e$ and $T_f$ are even local operators
supported on bonds $e$ and $f$, respectively,
their commutator vanishes whenever $e\cap f=\varnothing$.
Consequently, $(F_2)_Z$ can be nonzero only if
$Z$ is the union of two overlapping bonds.
In particular, $F_2$ has support size at most three.

Since $S_1$ has only nonzero $D$-grades and
$T^{(0)}$ has grade zero,
$\operatorname{ad}_{S_1}(T^{(0)})$ has no
grade-zero component. Thus, the diagonal part
of the second-order interaction $F_2$ is
\begin{equation}\label{eq:SW-H2}
 H^{(2)}(h):=F_2^{\rm diag}
 =\frac12\bigl(
   \operatorname{ad}_{S_1}(T^{\rm off})
   \bigr)^{\rm diag}.
\end{equation}
Here the superscript $(2)$ refers to the order
in the hopping interaction, whereas the superscript
$[1]$ in $H^{[1]}$ labels the first transformation.

We next eliminate the off-diagonal part of $F_2$.
Define the second generator by
\begin{equation}\label{eq:SW-S2}
 S_2:=\mathcal I_h(F_2^{\rm off}).
\end{equation}
By the defining property of $\mathcal I_h$,
\[
 \operatorname{ad}_{S_{2,\Lambda}}
       \bigl(H_{0,\Lambda}(h)\bigr)
 =-F_{2,\Lambda}^{\rm off}.
\]
The combined unitary transformation is
\begin{equation}\label{eq:def-SW-W}
 W_\Lambda^{[2]}(h)
 :=e^{S_{2,\Lambda}}e^{S_{1,\Lambda}}.
\end{equation}

\subsection{Volume-uniform finite-order normal form}

In this subsection, we establish a quantitative
normal form for the two-step transformation
constructed above. We show that the transformed
Hubbard Hamiltonian is diagonal with respect to
$D_\Lambda$ up to a remainder of order
$|\Lambda|U^{-2}$ in operator norm.
We also obtain volume-uniform bounds on the
generators $S_1$ and $S_2$.

The proof relies on finite-support commutator
estimates and Taylor expansions with integral
remainders, avoiding estimates on the full
BCH series.

\begin{prop}[Two-step Schrieffer--Wolff normal form]
\label{prop:finite-order-SW}
Fix $d\ge1$, $t\ne0$, and $h_0\ge0$.
There exist $U_0<\infty$ and $C<\infty$, depending only on
$d,t,h_0$, such that for every finite even torus $\Lambda$,
every $U\ge U_0$, and every $|h|\le h_0$, the interactions
$S_1,S_2$ defined by \eqref{eq:SW-S1} and \eqref{eq:SW-S2}
are even and anti-self-adjoint, have maximal support sizes
at most $2$ and $3$, respectively, and satisfy
\begin{equation}\label{eq:SW-generators-bounds}
 \|S_1\|_0\le C\frac{|t|}{U},
 \qquad
 \|S_2\|_0\le C\frac{t^2}{U^2}.
\end{equation}
The unitary $W_\Lambda^{[2]}(h)$ preserves total particle
number, and
\begin{equation}\label{eq:SW-normal-form}
 \begin{aligned}
 W_\Lambda^{[2]}(h)H_\Lambda^{\rm Hub}(h)
       W_\Lambda^{[2]}(h)^*
=UD_\Lambda-hM_\Lambda+T_\Lambda^{(0)}
       +H_\Lambda^{(2)}(h)+E_\Lambda(h),
 \end{aligned}
\end{equation}
where $E_\Lambda(h)$ is self-adjoint and
\begin{equation}\label{eq:SW-remainder-target}
 \|E_\Lambda(h)\|
 \le C|\Lambda|\frac{|t|^3}{U^2}.
\end{equation}
All bounds are uniform in $\Lambda$ and $|h|\le h_0$.
\end{prop}

\begin{proof}

Set
\[
 K_0:=4d|t|.
\]
With one term per unordered nearest-neighbour bond,
we have
\begin{equation}\label{eq:SW-K-bounds}
 \|T\|_0\le K_0,
 \qquad
 \|T^{(0)}\|_0\le K_0,
 \qquad
 \|T^{\rm off}\|_0\le2K_0.
\end{equation}
Indeed, each bond term satisfies
$\|T_e\|\le2|t|$, and each site belongs to
at most $2d$ bonds. Thus, the definition
\eqref{eq:spec-kappa-norm} gives
$\|T\|_0\le4d|t|$.
The remaining bounds follow from
\eqref{eq:grade-projection-norm} with $\kappa=0$.
We choose
\[
 U_0:=\max\{1,\,3\pi h_0,\,|t|\}
\]
and assume $U\ge U_0$ throughout the proof.
In particular, Lemma~\ref{lem:bounded-support-inverse}
applies to all local terms of support size at most three
that enter the construction.

\medskip\noindent
\emph{Step 1: local bounds for the generators and the quadratic term.}
By Lemma~\ref{lem:bounded-support-inverse}
and \eqref{eq:SW-K-bounds}, we obtain
\begin{equation}\label{eq:SW-S1-explicit}
 \begin{aligned}
 a_1:=\|S_1\|_0
 =\|\mathcal I_h(T^{\rm off})\|_0
 \le\frac{\pi}{U}\|T^{\rm off}\|_0
 \le\frac{2\pi K_0}{U}.
 \end{aligned}
\end{equation}
We next record the finite-support commutator estimate
used throughout this proof. Let $A,B$ be even
interactions with maximal support sizes $s(A)$ and
$s(B)$. Since disjoint even local terms commute,
\[
 \begin{aligned}
 \|\operatorname{ad}_A(B)\|_0
 &\le 2\sup_{z\in\Lambda}
   \sum_{\substack{X\cap Y\ne\varnothing\\z\in X\cup Y}}
       \|A_X\|\,\|B_Y\|\\
 &\le2\sup_{z\in\Lambda}
   \sum_{X\ni z}\|A_X\|
       \sum_{x\in X}\sum_{Y\ni x}\|B_Y\|
       +2\sup_{z\in\Lambda}
   \sum_{Y\ni z}\|B_Y\|
       \sum_{y\in Y}\sum_{X\ni y}\|A_X\|\\
 &\le2\bigl(s(A)+s(B)\bigr)
       \|A\|_0\|B\|_0.
 \end{aligned}
\]
Thus
\begin{equation}\label{eq:finite-support-commutator}
 \|\operatorname{ad}_A(B)\|_0
 \le2\bigl(s(A)+s(B)\bigr)
       \|A\|_0\|B\|_0.
\end{equation}
The bound depends only on the maximal support sizes
of the two input interactions, not on the volume.

Since $T$ is supported on nearest-neighbour bonds,
and both the grade projections and the local inverse
preserve supports by Lemma~\ref{lem:bounded-support-inverse},
we have
\[
 s(S_1),s(T^{(0)}),s(T^{\rm off})\le2.
\]
By the definition of $F_2$ in \eqref{eq:SW-F2},
only commutators between overlapping bonds contribute.
The union of two such bonds contains at most three sites,
so $s(F_2)\le3$.
Since $S_2=\mathcal I_h(F_2^{\rm off})$ and neither
the grade projection nor the local inverse enlarges
supports, we also have
\[
 s(F_2),s(S_2)\le3.
\]
Applying \eqref{eq:finite-support-commutator}
together with \eqref{eq:SW-K-bounds} and
\eqref{eq:SW-S1-explicit}, we obtain
\[
 \begin{aligned}
 \|\operatorname{ad}_{S_1}(T^{(0)})\|_0
 \le8a_1K_0,\qquad
 \|\operatorname{ad}_{S_1}(T^{\rm off})\|_0
 \le16a_1K_0.
 \end{aligned}
\]
Consequently, \eqref{eq:SW-F2} implies
\begin{equation}\label{eq:SW-F2-explicit}
 f_2:=\|F_2\|_0\le16a_1K_0
 \le\frac{32\pi K_0^2}{U}.
\end{equation}
The grade projection preserves the support, and
$\|F_2^{\rm off}\|_0\le2f_2$. Hence
\begin{equation}\label{eq:SW-S2-explicit}
 a_2:=\|S_2\|_0\le\frac{2\pi f_2}{U}
 \le\frac{64\pi^2K_0^2}{U^2}.
\end{equation}

Each local hopping term is self-adjoint and
conserves particle number on its support.
These properties are preserved by the diagonal
and off-diagonal grade decompositions.
By Lemma~\ref{lem:bounded-support-inverse},
$S_1=\mathcal I_h(T^{\rm off})$ is
anti-self-adjoint, and each of its local terms
commutes with the corresponding particle-number operator.

Since $F_2$ is constructed from commutators of
$S_1$ with $T^{(0)}$ and $T^{\rm off}$,
its local terms are self-adjoint and conserve
particle number on their supports.
Applying Lemma~\ref{lem:bounded-support-inverse}
once more, we obtain the same properties for $S_2$,
except that $S_2$ is anti-self-adjoint.
Consequently,
\[
 S_{j,\Lambda}^*=-S_{j,\Lambda},
 \qquad
 [N_\Lambda,S_{j,\Lambda}]=0,
 \qquad j=1,2.
\]
Thus, $W_\Lambda^{[2]}(h)$ is unitary and preserves
the canonical half-filled subspace.

\medskip

\noindent
\emph{Step 2: integral representation and bound for
the first remainder.}
We derive an exact integral representation of the
remainder $R_{3,\Lambda}$ introduced in
\eqref{eq:SW-first-BCH}.
For brevity, write
\[
 H_0:=H_{0,\Lambda}(h),\qquad
 A_0:=T_\Lambda^{(0)},\qquad
 A_1:=T_\Lambda^{\rm off}.
\]
By \eqref{eq:SW-S1}, we have
$\operatorname{ad}_{S_{1,\Lambda}}(H_0)=-A_1$.
Hence,
\[
 \begin{aligned}
 e^{\operatorname{ad}_{S_{1,\Lambda}}}
       (H_0+A_1)
 &=H_0
   +e^{\operatorname{ad}_{S_{1,\Lambda}}}(A_1)
   -\int_0^1
       e^{s\operatorname{ad}_{S_{1,\Lambda}}}(A_1)
       \,ds\\
 &=H_0+\int_0^1s\,
       e^{s\operatorname{ad}_{S_{1,\Lambda}}}
       \bigl(\operatorname{ad}_{S_{1,\Lambda}}(A_1)\bigr)
       \,ds,
 \end{aligned}
\]
where the last equality follows by integration
by parts. Consequently,
\[
 \begin{aligned}
 H^{[1]}
 ={}&H_0+
 e^{\operatorname{ad}_{S_{1,\Lambda}}}(A_0)
 +\int_0^1s\,
       e^{s\operatorname{ad}_{S_{1,\Lambda}}}
       \bigl(\operatorname{ad}_{S_{1,\Lambda}}(A_1)\bigr)
       \,ds.
 \end{aligned}
\]

Taylor's formula with integral remainder gives
\[
 \begin{aligned}
 e^{\operatorname{ad}_{S_{1,\Lambda}}}(A_0)
 ={}&A_0+\operatorname{ad}_{S_{1,\Lambda}}(A_0)
 +\int_0^1(1-s)
 e^{s\operatorname{ad}_{S_{1,\Lambda}}}
 \bigl(\operatorname{ad}_{S_{1,\Lambda}}^2(A_0)\bigr)
 \,ds,
 \end{aligned}
\]
and
\[
 \begin{aligned}
 &\int_0^1s\,
 e^{s\operatorname{ad}_{S_{1,\Lambda}}}
 \bigl(\operatorname{ad}_{S_{1,\Lambda}}(A_1)\bigr)
 \,ds
 =
 \frac12\operatorname{ad}_{S_{1,\Lambda}}(A_1)
 +\frac12\int_0^1(1-s^2)
 e^{s\operatorname{ad}_{S_{1,\Lambda}}}
 \bigl(\operatorname{ad}_{S_{1,\Lambda}}^2(A_1)\bigr)
 \,ds.
 \end{aligned}
\]
The terms
\[
 \operatorname{ad}_{S_{1,\Lambda}}(A_0)
 +\frac12\operatorname{ad}_{S_{1,\Lambda}}(A_1)
 =F_{2,\Lambda}
\]
agree with the second-order interaction defined
in \eqref{eq:SW-F2}.
Comparing the resulting expression for $H^{[1]}$
with \eqref{eq:SW-first-BCH}, we obtain
\begin{equation}\label{eq:SW-R3-integral}
 \begin{aligned}
 R_{3,\Lambda}
 =\int_0^1(1-s)
   e^{s\operatorname{ad}_{S_{1,\Lambda}}}
   \bigl(\operatorname{ad}_{S_{1,\Lambda}}^2(A_0)\bigr)
   \,ds
 +\frac12\int_0^1(1-s^2)
   e^{s\operatorname{ad}_{S_{1,\Lambda}}}
   \bigl(\operatorname{ad}_{S_{1,\Lambda}}^2(A_1)\bigr)
   \,ds.
 \end{aligned}
\end{equation}

Since $S_{1,\Lambda}$ is anti-self-adjoint,
the adjoint action preserves the operator norm
for real $s$. Therefore,
\begin{equation}\label{eq:SW-R3-operator-bound}
 \begin{aligned}
 \|R_{3,\Lambda}\|
 \le\frac12
   \|\operatorname{ad}_{S_{1,\Lambda}}^2(A_0)\|
 +\frac13
   \|\operatorname{ad}_{S_{1,\Lambda}}^2(A_1)\|.
 \end{aligned}
\end{equation}

The first interaction commutators have support sizes
at most three. Applying \eqref{eq:finite-support-commutator}
again, we obtain
\begin{equation}\label{eq:SW-S1-double-commutator-bounds}
 \begin{aligned}
 \|\operatorname{ad}_{S_1}^2(T^{(0)})\|_0
 &\le10a_1\|\operatorname{ad}_{S_1}(T^{(0)})\|_0
 \le80a_1^2K_0,\\
 \|\operatorname{ad}_{S_1}^2(T^{\rm off})\|_0
 &\le10a_1\|\operatorname{ad}_{S_1}(T^{\rm off})\|_0
 \le160a_1^2K_0.
 \end{aligned}
\end{equation}
By Lemma~\ref{lem:interaction-to-operator},
\begin{equation}\label{eq:SW-R3-estimate}
 \begin{aligned}
 \frac{\|R_{3,\Lambda}\|}{|\Lambda|}
 \le\frac12(80a_1^2K_0)+\frac13(160a_1^2K_0)
 =\frac{280}{3}a_1^2K_0.
 \end{aligned}
\end{equation}

\medskip
\noindent
\emph{Step 3: integral remainder for the second transformation.}
Put
\[
 F_d:=F_{2,\Lambda}^{\rm diag},\qquad
 F_o:=F_{2,\Lambda}^{\rm off}.
\]
Since $\operatorname{ad}_{S_{2,\Lambda}}(H_0)=-F_o$,
the same integration-by-parts identity yields
\begin{equation}\label{eq:SW-second-integral-identity}
 \begin{aligned}
 e^{\operatorname{ad}_{S_{2,\Lambda}}}(H_0+F_o)
 =H_0+\int_0^1s\,
  e^{s\operatorname{ad}_{S_{2,\Lambda}}}
  \bigl(\operatorname{ad}_{S_{2,\Lambda}}(F_o)\bigr)\,ds.
 \end{aligned}
\end{equation}
Moreover,
\[
 \begin{aligned}
 e^{\operatorname{ad}_{S_{2,\Lambda}}}(A_0+F_d)
 =A_0+F_d+\int_0^1
  e^{s\operatorname{ad}_{S_{2,\Lambda}}}
  \bigl(\operatorname{ad}_{S_{2,\Lambda}}(A_0+F_d)\bigr)\,ds.
 \end{aligned}
\]
Applying the second adjoint action to
\eqref{eq:SW-first-BCH}, we obtain
\eqref{eq:SW-normal-form}, where
\begin{equation}\label{eq:SW-E-integral}
 \begin{aligned}
 E_\Lambda(h)
 ={}&e^{\operatorname{ad}_{S_{2,\Lambda}}}(R_{3,\Lambda})
 +\int_0^1
   e^{s\operatorname{ad}_{S_{2,\Lambda}}}
   \bigl(\operatorname{ad}_{S_{2,\Lambda}}(A_0+F_d)\bigr)\,ds\\
 &+\int_0^1s\,
   e^{s\operatorname{ad}_{S_{2,\Lambda}}}
   \bigl(\operatorname{ad}_{S_{2,\Lambda}}(F_o)\bigr)\,ds.
 \end{aligned}
\end{equation}
By norm preservation under the adjoint action,
\begin{equation}\label{eq:SW-E-operator-bound}
 \begin{aligned}
 \|E_\Lambda(h)\|
 \le \|R_{3,\Lambda}\|
 +\|\operatorname{ad}_{S_{2,\Lambda}}(A_0+F_d)\|
 +\frac12\|\operatorname{ad}_{S_{2,\Lambda}}(F_o)\|.
 \end{aligned}
\end{equation}

The support sizes of $S_2$, $T^{(0)}$, and $F_2$
are at most $3$, $2$, and $3$, respectively. Using
\eqref{eq:finite-support-commutator}, we find
\begin{equation}\label{eq:SW-S2-commutator-bounds}
 \begin{aligned}
 \|\operatorname{ad}_{S_2}(T^{(0)})\|_0
 &\le10a_2K_0,\\
 \|\operatorname{ad}_{S_2}(F_2^{\rm diag})\|_0
 &\le12a_2f_2,\\
 \|\operatorname{ad}_{S_2}(F_2^{\rm off})\|_0
 &\le24a_2f_2.
 \end{aligned}
\end{equation}
Consequently, by \eqref{eq:SW-R3-estimate},
\eqref{eq:SW-E-operator-bound}, and
Lemma~\ref{lem:interaction-to-operator},
\begin{equation}\label{eq:SW-E-final-estimate}
 \begin{aligned}
 \frac{\|E_\Lambda(h)\|}{|\Lambda|}
 &\le\frac{280}{3}a_1^2K_0
       +10a_2K_0+24a_2f_2\\
 &\le\frac{3040\pi^2K_0^3}{3U^2}
       +\frac{49152\pi^3K_0^4}{U^3}.
 \end{aligned}
\end{equation}
Since $K_0=4d|t|$ and $U\ge|t|$, the right-hand
side is bounded by $C|t|^3/U^2$ with $C$ depending
only on $d$ (or, harmlessly, on the fixed parameters
$d,t,h_0$). The remainder $E_\Lambda(h)$ is
self-adjoint because both sides of
\eqref{eq:SW-normal-form} are self-adjoint.
This proves \eqref{eq:SW-remainder-target}.
\end{proof}

\subsection{Local conservation laws for the truncated normal form}

The charge-sector analysis requires conservation of
particle number and doublon number for each local
term of the truncated normal form.
We establish these properties in the following lemma.

\begin{lem}[Local conservation of particle number and doublon number]
\label{lem:local-conservation-truncated}
The interaction
\begin{equation}\label{eq:def-Phi-effective}
 \Phi(h)=-hM+T^{(0)}+H^{(2)}(h)
\end{equation}
can be represented as a sum of self-adjoint even local
terms $\Phi_X(h)$ such that
\begin{equation}\label{eq:Phi-conservation}
 [\Phi_X(h),N_X]=[\Phi_X(h),D_X]=0
 \qquad\text{for every }X.
\end{equation}
The supports in this representation have size at most three.
\end{lem}

\begin{proof}
The one-site field terms commute with their local
particle-number and doublon-number operators.
Each bond hopping term commutes with the total particle
number on its bond, and its grade-zero component also
commutes with the local doublon number.

By Lemma~\ref{lem:bounded-support-inverse}, the local
inverse preserves commutation with local particle number.
Thus $[S_{1,e},N_e]=0$ for every bond $e$.
For even local operators $A_X,B_Y$ satisfying
$[A_X,N_X]=[B_Y,N_Y]=0$, locality and the Jacobi
identity give, with $Z=X\cup Y$,
\[
 \begin{aligned}
 [N_Z,\operatorname{ad}_{A_X}(B_Y)]
 =\operatorname{ad}_{[N_Z,A_X]}(B_Y)
   +\operatorname{ad}_{A_X}([N_Z,B_Y])
   =0.
 \end{aligned}
\]
Thus every local term of $F_2$ commutes with
$N_X$ on its support. The local grade-zero projection
preserves this commutation relation and, by definition,
commutes with the local doublon-number operator.
Consequently, every term of $H^{(2)}(h)=F_2^{\rm diag}$
commutes with both $N_X$ and $D_X$ on its support.
The same properties hold for the field and
$T^{(0)}$ terms. Finally, the commutators defining
$F_2$ involve only overlapping bonds, so all resulting
supports have cardinality at most three.
\end{proof}


\subsection{Third-order correction and bipartite cancellation}
\label{subsec:third-order-SW}

In this subsection, we identify the third-order
terms arising from the first two Schrieffer--Wolff
transformations and show that their contribution
to the singly occupied sector vanishes by
bipartite symmetry.
We then construct the generator of the third
transformation, which removes the remaining
off-diagonal terms.
The detailed expansion of the transformed
Hamiltonian and the corresponding remainder
estimates will be given in the subsequent proposition.

\begin{lem}[Third-order interaction and local bounds]
\label{lem:third-order-interaction}
Define
\begin{equation}\label{eq:SW-F3-C3}
 \begin{aligned}
 F_3:=\frac12\operatorname{ad}_{S_1}^{2}(T^{(0)})
      +\frac13\operatorname{ad}_{S_1}^{2}(T^{\rm off}),\qquad
 C_3:=F_3+\operatorname{ad}_{S_2}(T^{(0)}).
 \end{aligned}
\end{equation}
Both $F_3$ and $C_3$ are self-adjoint even interactions,
each of whose terms conserves particle number on its support.
Their maximal support sizes are at most four. Moreover,
with $K_0=4d|t|$, $a_1=\|S_1\|_0$ and $a_2=\|S_2\|_0$,
\begin{equation}\label{eq:SW-F3-C3-bounds}
 \begin{aligned}
 f_3:=\|F_3\|_0
 &\le\frac{280}{3}a_1^2K_0
 \le\frac{1120\pi^2K_0^3}{3U^2},\\
 c_3:=\|C_3\|_0
 &\le f_3+10a_2K_0
 \le\frac{3040\pi^2K_0^3}{3U^2}.
 \end{aligned}
\end{equation}
\end{lem}
\begin{proof}
The first commutators of $S_1$ with $T^{(0)}$ and
$T^{\rm off}$ are supported on at most three sites.
Commuting once more with $S_1$, or commuting $S_2$
with $T^{(0)}$, produces terms on at most four sites.
Local particle-number conservation follows from the Jacobi
identity, as in Lemma~\ref{lem:local-conservation-truncated}.
Commutation with an anti-self-adjoint operator preserves
self-adjointness of the input, proving the corresponding assertion.

Using \eqref{eq:SW-S1-double-commutator-bounds}
and \eqref{eq:SW-S2-commutator-bounds},
the definitions of $F_3$ and $C_3$ give
\[
 \begin{aligned}
 f_3
 &\le\frac12(80a_1^2K_0)
      +\frac13(160a_1^2K_0)
  =\frac{280}{3}a_1^2K_0,\\
 c_3
 &\le f_3+10a_2K_0.
 \end{aligned}
\]
Substituting \eqref{eq:SW-S1-explicit} and
\eqref{eq:SW-S2-explicit} proves
\eqref{eq:SW-F3-C3-bounds}.
\end{proof}

Let $\Lambda=A\cup B$ be the bipartition with
$A=\{x:\eta_x=1\}$ and $B=\{x:\eta_x=-1\}$.
On the fermionic Fock space set
\[
 N_A:=\sum_{x\in A}n_x,
 \qquad \Gamma:=e^{\im \pi N_A}.
\]

\begin{lem}[Bipartite oddness of the third-order term]
\label{lem:third-order-bipartite}
For every even torus and $|h|\le h_0$ in the regime of
Lemma~\ref{lem:bounded-support-inverse},
\begin{equation}\label{eq:SW-Gamma-parities}
 \begin{aligned}
 \Gamma T_\Lambda\Gamma^*&=-T_\Lambda,&
 \Gamma S_{1,\Lambda}\Gamma^*&=-S_{1,\Lambda},\\
 \Gamma S_{2,\Lambda}\Gamma^*&=S_{2,\Lambda},&
 \Gamma C_{3,\Lambda}\Gamma^*&=-C_{3,\Lambda}.
 \end{aligned}
\end{equation}
In particular,
\begin{equation}\label{eq:SW-PC3P}
 PC_{3,\Lambda}P=PC_{3,\Lambda}^{\rm diag}P=0.
\end{equation}
\end{lem}
\begin{proof}
Every nearest-neighbour bond connects $A$ to $B$, and
\[
 \Gamma c_{x\sigma}\Gamma^*
 =(-1)^{\mathbf 1_{\{x\in A\}}}c_{x\sigma}.
\]
Thus $\Gamma T_\Lambda\Gamma^*=-T_\Lambda$.
For each $X\subset\Lambda$ put
\[
 \Gamma_X:=e^{\im \pi N_{A\cap X}},\qquad
 N_{A\cap X}:=\sum_{x\in A\cap X}n_x.
\]
Since $[\Gamma_X,D_X]=[\Gamma_X,M_X]=0$, conjugation by
$\Gamma_X$ commutes with $e^{\im \theta\operatorname{ad}_{D_X}}$
and with $\operatorname{ad}_{M_X}$. Hence it commutes with the
Fourier grade projections and with the map
\[
 \mathcal C_X(B)=\frac{\im}{2\pi}\int_0^{2\pi}(\theta-\pi)
   e^{\im \theta\operatorname{ad}_{D_X}}(B)\,d\theta.
\]
The Neumann-series formula for $\mathcal I_{h,X}$ therefore gives
\[
 \mathcal I_{h,X}(\Gamma_X B\Gamma_X^*)
 =\Gamma_X\mathcal I_{h,X}(B)\Gamma_X^*.
\]
On local operators the adjoint actions of $\Gamma$ and $\Gamma_X$
coincide. Thus the local grade projections and the local inverse
preserve $\Gamma$-parity.
It follows that $T^{(0)}$, $T^{\rm off}$ and $S_1$
are $\Gamma$-odd, while $F_2$ and $S_2$ are
$\Gamma$-even. Every term of \eqref{eq:SW-F3-C3}
is therefore $\Gamma$-odd, proving
\eqref{eq:SW-Gamma-parities}.

On $P\cH_\Lambda^{\rm hf}$, every site of $A$ is singly
occupied; hence
\[
 \Gamma P=(-1)^{|A|}P=P\Gamma.
\]
Consequently,
\[
 PC_{3,\Lambda}P
 =P\Gamma^*(\Gamma C_{3,\Lambda}\Gamma^*)\Gamma P
 =-PC_{3,\Lambda}P,
\]
so $PC_{3,\Lambda}P=0$. The identity for the diagonal
part follows from \eqref{eq:P-compression-grade-zero}.
\end{proof}

\subsection{Three-step SW normal form}
\label{subsec:three-step-SW}

We now complete the three-step Schrieffer--Wolff
construction. The following proposition defines the
third generator $S_3$ and establishes the resulting
normal form with an explicit volume-uniform bound
on the remainder.

\begin{prop}[Three-step finite-order Schrieffer--Wolff normal form]
\label{prop:third-order-SW}
Fix $d\ge1$, $t\ne0$ and $h_0\ge0$, and suppose
\[
 U\ge U_0^{[3]}:=\max\{1,\,4\pi h_0,\,K_0\},
 \qquad K_0=4d|t|,\quad |h|\le h_0.
\]
Define the anti-self-adjoint third generator and the unitary
\begin{equation}\label{eq:SW-S3}
 \begin{aligned}
 S_3&:=\mathcal I_h(C_3^{\rm off}),
 &\operatorname{ad}_{S_{3,\Lambda}}(H_{0,\Lambda}(h))
 &=-C_{3,\Lambda}^{\rm off},\\
 W_\Lambda^{[3]}(h)
 &:=e^{S_{3,\Lambda}}e^{S_{2,\Lambda}}e^{S_{1,\Lambda}}.
 \end{aligned}
\end{equation}
The interaction $S_3$ is even, preserves local particle number,
and has maximal support size at most four, with
\begin{equation}\label{eq:SW-S3-bound}
 a_3:=\|S_3\|_0\le\frac{2\pi c_3}{U}
 \le\frac{6080\pi^3K_0^3}{3U^3}.
\end{equation}
The unitary $W_\Lambda^{[3]}(h)$ preserves the canonical
half-filled space and satisfies
\begin{equation}\label{eq:SW-third-normal-form}
 \begin{aligned}
 &W_\Lambda^{[3]}(h)H_\Lambda^{\rm Hub}(h)
                W_\Lambda^{[3]}(h)^*
                =H_{0,\Lambda}(h)+T_\Lambda^{(0)}
    +H_\Lambda^{(2)}(h)+C_{3,\Lambda}^{\rm diag}
    +E_\Lambda^{[3]}(h),
 \end{aligned}
\end{equation}
where $E_\Lambda^{[3]}(h)$ is self-adjoint and obeys
\begin{equation}\label{eq:SW-third-remainder-bound}
 \begin{aligned}
 \frac{\|E_\Lambda^{[3]}(h)\|}{|\Lambda|}
 \le{}&400a_1^3K_0+24a_2f_2+14a_2f_3+70a_2^2K_0\\
 &+12a_3K_0+14a_3f_2+32a_3c_3
 \le C_{\rm SW,3}\frac{K_0^4}{U^3},
 \end{aligned}
\end{equation}
with a numerical constant $C_{\rm SW,3}<\infty$.
Every bound is uniform in the finite even torus and in $|h|\le h_0$.
\end{prop}

\begin{proof}

We abbreviate
\[
 H_0:=H_{0,\Lambda}(h),\quad
 A_0:=T_\Lambda^{(0)},\quad V:=T_\Lambda^{\rm off},
 \quad F_d:=F_{2,\Lambda}^{\rm diag},\quad
 F_o:=F_{2,\Lambda}^{\rm off}.
\]
All adjoint actions below are finite-volume operator
adjoint actions.
Since the generators are anti-self-adjoint,
$e^{s\operatorname{ad}_{S_{j,\Lambda}}}$
preserves the operator norm for real $s$.

\medskip
\noindent
\emph{Step 1: the first transformation through third order.}
Recall from \eqref{eq:SW-first-BCH} that
\[
 H^{[1]}=H_0+A_0+F_{2,\Lambda}+R_{3,\Lambda},
\]
where $R_{3,\Lambda}$ has the integral representation
\eqref{eq:SW-R3-integral}.

We expand the adjoint actions in that representation
once more, using
\[
 e^{s\operatorname{ad}_{S_{1,\Lambda}}}(B)
 =B+\int_0^s
 e^{u\operatorname{ad}_{S_{1,\Lambda}}}
 \bigl(\operatorname{ad}_{S_{1,\Lambda}}(B)\bigr)
 \,du.
\]
The terms without the additional commutator are
\[
 \frac12\operatorname{ad}_{S_{1,\Lambda}}^2(A_0)
 +\frac13\operatorname{ad}_{S_{1,\Lambda}}^2(V)
 =F_{3,\Lambda}.
\]
Interchanging the order of integration in the
remaining terms, we obtain
\begin{equation}\label{eq:SW-third-first-expansion}
 H^{[1]}=H_0+A_0+F_{2,\Lambda}+F_{3,\Lambda}
         +R_{4,\Lambda}^{[1]},
\end{equation}
where
\begin{equation}\label{eq:SW-third-first-integral-remainder}
 \begin{aligned}
 R_{4,\Lambda}^{[1]}
 ={}&\frac12\int_0^1(1-s)^2
    e^{s\operatorname{ad}_{S_{1,\Lambda}}}
    \bigl(\operatorname{ad}_{S_{1,\Lambda}}^3(A_0)\bigr)
    \,ds\\
 &+\frac16\int_0^1(1-s)^2(s+2)
    e^{s\operatorname{ad}_{S_{1,\Lambda}}}
    \bigl(\operatorname{ad}_{S_{1,\Lambda}}^3(V)\bigr)
    \,ds.
 \end{aligned}
\end{equation}
Since $\int_0^1(1-s)^2/2\,ds=1/6$, while
$\int_0^1(1-s)^2(s+2)/6\,ds=1/8$, we obtain
\begin{equation}\label{eq:SW-third-R4-first-operator}
 \|R_{4,\Lambda}^{[1]}\|
 \le\frac16\|\operatorname{ad}_{S_{1,\Lambda}}^3(A_0)\|
    +\frac18\|\operatorname{ad}_{S_{1,\Lambda}}^3(V)\|.
\end{equation}
The twofold interaction commutators have support at most four.
Since $S_1$ is supported on individual bonds,
the double commutators
$\operatorname{ad}_{S_1}^2(T^{(0)})$ and
$\operatorname{ad}_{S_1}^2(T^{\rm off})$
have support size at most four.
Thus, \eqref{eq:finite-support-commutator}
and \eqref{eq:SW-S1-double-commutator-bounds}
yield
\begin{equation}\label{eq:SW-S1-triple-commutator-bounds}
 \begin{aligned}
 \|\operatorname{ad}_{S_1}^3(T^{(0)})\|_0
 &\le12a_1
   \|\operatorname{ad}_{S_1}^2(T^{(0)})\|_0
 \le960a_1^3K_0,\\
 \|\operatorname{ad}_{S_1}^3(T^{\rm off})\|_0
 &\le12a_1
   \|\operatorname{ad}_{S_1}^2(T^{\rm off})\|_0
 \le1920a_1^3K_0.
 \end{aligned}
\end{equation}
Lemma~\ref{lem:interaction-to-operator} consequently implies
\begin{equation}\label{eq:SW-third-R4-first-bound}
 \frac{\|R_{4,\Lambda}^{[1]}\|}{|\Lambda|}
 \le\left(\frac{960}{6}+\frac{1920}{8}\right)a_1^3K_0
 =400a_1^3K_0.
\end{equation}
\medskip
\noindent
\emph{Step 2: the second transformation through third order.}
By \eqref{eq:SW-third-first-expansion}, we have
\[
 H^{[1]}
 =H_0+A_0+F_d+F_o+F_{3,\Lambda}
  +R_{4,\Lambda}^{[1]}.
\]
Recall from \eqref{eq:SW-S2} that
\[
 \operatorname{ad}_{S_{2,\Lambda}}(H_0)=-F_o.
\]
As in the first transformation, integration of the
adjoint action gives
\[
 \begin{aligned}
 e^{\operatorname{ad}_{S_{2,\Lambda}}}(H_0+F_o)
 =H_0+\int_0^1s\,
 e^{s\operatorname{ad}_{S_{2,\Lambda}}}
 \bigl(\operatorname{ad}_{S_{2,\Lambda}}(F_o)\bigr)
 \,ds.
 \end{aligned}
\]
Thus, the leading off-diagonal term $F_o$ is
cancelled by the second transformation.

For the remaining terms, we use the Taylor
expansion
\[
 \begin{aligned}
 e^{\operatorname{ad}_{S_{2,\Lambda}}}(A_0)
 =A_0+\operatorname{ad}_{S_{2,\Lambda}}(A_0)
 +(e^{\operatorname{ad}_{S_{2,\Lambda}}}
    -\mathrm{id}
    -\operatorname{ad}_{S_{2,\Lambda}})(A_0).
 \end{aligned}
\]
Keeping the terms
\[
 F_{3,\Lambda}
 +\operatorname{ad}_{S_{2,\Lambda}}(A_0)
 =C_{3,\Lambda}
\]
and collecting the remaining contributions, we obtain
\begin{equation}\label{eq:SW-third-second-expansion}
 \begin{aligned}
 H^{[2]}
 :=e^{\operatorname{ad}_{S_{2,\Lambda}}}(H^{[1]})
 =H_0+A_0+F_d+C_{3,\Lambda}
   +R_{4,\Lambda}^{[2]}.
 \end{aligned}
\end{equation}
Here the exact remainder is
\begin{equation}\label{eq:SW-third-second-remainder}
 \begin{aligned}
 R_{4,\Lambda}^{[2]}
 ={}&e^{\operatorname{ad}_{S_{2,\Lambda}}}
        (R_{4,\Lambda}^{[1]})
 +(e^{\operatorname{ad}_{S_{2,\Lambda}}}
       -\mathrm{id})(F_{3,\Lambda})
 +(e^{\operatorname{ad}_{S_{2,\Lambda}}}
       -\mathrm{id}
       -\operatorname{ad}_{S_{2,\Lambda}})(A_0)\\
 &+(e^{\operatorname{ad}_{S_{2,\Lambda}}}
       -\mathrm{id})(F_d)
 +\int_0^1s\,
       e^{s\operatorname{ad}_{S_{2,\Lambda}}}
       \bigl(\operatorname{ad}_{S_{2,\Lambda}}(F_o)\bigr)
       \,ds.
 \end{aligned}
\end{equation}

The integral remainder formulas and unitary conjugation give
\[
 \begin{aligned}
 \|R_{4,\Lambda}^{[2]}\|
 \le \|R_{4,\Lambda}^{[1]}\|
       +\|\operatorname{ad}_{S_{2,\Lambda}}(F_{3,\Lambda})\|
 +\frac12\|\operatorname{ad}_{S_{2,\Lambda}}^2(A_0)\|
       +\|\operatorname{ad}_{S_{2,\Lambda}}(F_d)\|
 +\frac12\|\operatorname{ad}_{S_{2,\Lambda}}(F_o)\|.
 \end{aligned}
\]
Since $s(S_2)\le3$, $s(F_3)\le4$, and
$s(\operatorname{ad}_{S_2}(T^{(0)}))\le4$,
the finite-support commutator estimate
\eqref{eq:finite-support-commutator}, together with
\eqref{eq:SW-S2-commutator-bounds}, gives
\[
 \begin{aligned}
 \|\operatorname{ad}_{S_2}(F_3)\|_0
 &\le14a_2f_3,\\
 \|\operatorname{ad}_{S_2}^2(T^{(0)})\|_0
 &\le14a_2
 \|\operatorname{ad}_{S_2}(T^{(0)})\|_0
 \le140a_2^2K_0.
 \end{aligned}
\]
The bounds for
$\operatorname{ad}_{S_2}(F_2^{\rm diag})$ and
$\operatorname{ad}_{S_2}(F_2^{\rm off})$
were already established in
\eqref{eq:SW-S2-commutator-bounds}.
Thus
\begin{equation}\label{eq:SW-third-R4-second-bound}
 \frac{\|R_{4,\Lambda}^{[2]}\|}{|\Lambda|}
 \le400a_1^3K_0+14a_2f_3+70a_2^2K_0+24a_2f_2.
\end{equation}

\medskip
\noindent
\emph{Step 3: the third transformation.}
By \eqref{eq:SW-third-second-expansion}, we have
\[
 H^{[2]}
 =H_0+A_0+F_d+C_{3,\Lambda}^{\rm diag}
  +C_{3,\Lambda}^{\rm off}+R_{4,\Lambda}^{[2]}.
\]
The definition \eqref{eq:SW-S3} gives
\[
 \operatorname{ad}_{S_{3,\Lambda}}(H_0)
 =-C_{3,\Lambda}^{\rm off}.
\]
Thus, as in \eqref{eq:SW-second-integral-identity},
the leading off-diagonal term cancels, leaving
\[
 \begin{aligned}
 e^{\operatorname{ad}_{S_{3,\Lambda}}}
   (H_0+C_{3,\Lambda}^{\rm off})
 =H_0+\int_0^1s\,
   e^{s\operatorname{ad}_{S_{3,\Lambda}}}
   \bigl(\operatorname{ad}_{S_{3,\Lambda}}
   (C_{3,\Lambda}^{\rm off})\bigr)\,ds.
 \end{aligned}
\]
For the remaining terms, we use
\[
 e^{\operatorname{ad}_{S_{3,\Lambda}}}(B)
 =B+\int_0^1
 e^{s\operatorname{ad}_{S_{3,\Lambda}}}
 \bigl(\operatorname{ad}_{S_{3,\Lambda}}(B)\bigr)\,ds.
\]
Combining these identities gives the normal form
\eqref{eq:SW-third-normal-form}, with the exact remainder
\begin{equation}\label{eq:SW-third-final-integral-remainder}
 \begin{aligned}
 E_\Lambda^{[3]}
 ={}&e^{\operatorname{ad}_{S_{3,\Lambda}}}
              (R_{4,\Lambda}^{[2]})
 +\int_0^1e^{s\operatorname{ad}_{S_{3,\Lambda}}}
   \bigl(\operatorname{ad}_{S_{3,\Lambda}}
       (A_0+F_d+C_{3,\Lambda}^{\rm diag})\bigr)\,ds\\
 &+\int_0^1s\,e^{s\operatorname{ad}_{S_{3,\Lambda}}}
       \bigl(\operatorname{ad}_{S_{3,\Lambda}}
       (C_{3,\Lambda}^{\rm off})\bigr)\,ds.
 \end{aligned}
\end{equation}
Therefore,
\[
 \begin{aligned}
 \|E_\Lambda^{[3]}\|
 \le{}&\|R_{4,\Lambda}^{[2]}\|
 +\|\operatorname{ad}_{S_{3,\Lambda}}
       (A_0+F_d+C_{3,\Lambda}^{\rm diag})\|+\frac12\|\operatorname{ad}_{S_{3,\Lambda}}
       (C_{3,\Lambda}^{\rm off})\|.
 \end{aligned}
\]
Since $s(S_3)\le4$, $s(F_2)\le3$, and
$s(C_3)\le4$, the finite-support commutator
estimate \eqref{eq:finite-support-commutator},
together with \eqref{eq:SW-K-bounds} and
\eqref{eq:grade-projection-norm}, yields
\[
 \begin{aligned}
 \|\operatorname{ad}_{S_3}(T^{(0)})\|_0
 &\le12a_3K_0,\\
 \|\operatorname{ad}_{S_3}(F_2^{\rm diag})\|_0
 &\le14a_3f_2,\\
 \|\operatorname{ad}_{S_3}(C_3^{\rm diag})\|_0
 &\le16a_3c_3,\\
 \|\operatorname{ad}_{S_3}(C_3^{\rm off})\|_0
 &\le32a_3c_3.
 \end{aligned}
\]
Combining these estimates with
\eqref{eq:SW-third-R4-second-bound} proves the first inequality
in \eqref{eq:SW-third-remainder-bound}.
Finally, \eqref{eq:SW-S1-explicit}, \eqref{eq:SW-F2-explicit},
\eqref{eq:SW-S2-explicit}, \eqref{eq:SW-F3-C3-bounds} and
\eqref{eq:SW-S3-bound} imply
\begin{equation}\label{eq:SW-third-remainder-polynomial}
 \begin{aligned}
 \frac{\|E_\Lambda^{[3]}\|}{|\Lambda|}
 \le{}&\frac{76672\pi^3K_0^4}{U^3}
   +\frac{4587520\pi^4K_0^5}{3U^4}
   +\frac{591462400\pi^5K_0^6}{9U^5}.
 \end{aligned}
\end{equation}
Since $U\ge K_0$, this is at most
\[
 \left(76672\pi^3+\frac{4587520\pi^4}{3}
   +\frac{591462400\pi^5}{9}\right)\frac{K_0^4}{U^3}.
\]
This proves the second inequality of
\eqref{eq:SW-third-remainder-bound}; self-adjointness of
$E_\Lambda^{[3]}$ follows from
\eqref{eq:SW-third-normal-form}.
\end{proof}

\begin{lem}[Local conservation for the third-order truncation]
\label{lem:third-order-local-conservation}
The interaction
\begin{equation}\label{eq:Phi-third-order}
 \Phi^{[3]}(h):=-hM+T^{(0)}+H^{(2)}(h)+C_3^{\rm diag}
\end{equation}
admits a representation by self-adjoint even local terms
$\Phi_X^{[3]}(h)$ with maximal support size four and
\[
 [\Phi_X^{[3]}(h),N_X]=[\Phi_X^{[3]}(h),D_X]=0
 \qquad(\varnothing\ne X\subset\Lambda).
\]
\end{lem}
\begin{proof}
The first three terms satisfy the assertion by
Lemma~\ref{lem:local-conservation-truncated}.
By Lemma~\ref{lem:third-order-interaction}, every local term
of $C_3$ preserves its local particle number. The local
$D$-grade-zero projection preserves this property, because
$[D_X,N_X]=0$, and makes each resulting term commute
with $D_X$. The support bound follows from $s(C_3)\le4$.
\end{proof}

\section{Comparison with the defect-free sector}
\label{sec:defects-revision}

The normal form obtained in the preceding section
still contains contributions from sectors with
charge defects. In this section, we show that
their contribution to the Gibbs pressure can be
controlled uniformly in the volume, allowing us
to reduce the pressure comparison to the
defect-free sector. The argument relies on
local conservation of particle number and
doublon number.

\subsection{Fixed-defect sectors and local spin compression}
\label{subsec:defect-spin-compression}

We first describe the decomposition of the canonical
half-filled space into sectors with fixed charge
configurations. We then show that local operators,
when compressed to singly occupied sites, admit
a spin representation independent of the charge
configuration outside their supports.

Let $V_0,V_2\subset\Lambda$ be disjoint sets with
$|V_0|=|V_2|$, and let $\cH_{V_0,V_2}$ denote the
subspace of $\cH_\Lambda^{\rm hf}$ in which
the sites of $V_0$ are empty, those of $V_2$
are doubly occupied, and all other sites are
singly occupied. Set
\[
 \mathcal D:=V_0\cup V_2,
 \qquad R:=\Lambda\setminus\mathcal D.
\]
Thus, $\mathcal D$ is the set of charge-defect
sites, while $R$ is the set of remaining spin sites.

Let
\[
 \cH_R^{\rm spin}
 :=\bigotimes_{x\in R}\mathbb C^2,
\]
where the tensor factors are ordered according
to a fixed ordering
$\Lambda=\{x_1,\ldots,x_{|\Lambda|}\}$.
We use the convention
$\cH_\varnothing^{\rm spin}=\mathbb C$.

To identify $\cH_{V_0,V_2}$ with
$\cH_R^{\rm spin}$, introduce the local operators
\[
 \begin{aligned}
 b_x(0)=\mathbbm1,\qquad
 b_x(\uparrow)=c_{x\uparrow}^*, \qquad
 b_x(\downarrow)=c_{x\downarrow}^*,\qquad
 b_x(2)=c_{x\uparrow}^*c_{x\downarrow}^*.
 \end{aligned}
\]
For $\boldsymbol\sigma
\in\{\uparrow,\downarrow\}^R$, let
$a_x(\boldsymbol\sigma)$ equal $0$ on $V_0$,
$2$ on $V_2$, and $\sigma_x$ on $R$.
Define
\begin{equation}\label{eq:defect-sector-spin-identification}
 \begin{aligned}
 \mathcal U_{V_0,V_2}:
 \cH_R^{\rm spin}
 &\longrightarrow\cH_{V_0,V_2},\\
 |\boldsymbol\sigma\rangle_R
 &\longmapsto
 \left(
 \prod_{j=1}^{|\Lambda|}
 b_{x_j}(a_{x_j}(\boldsymbol\sigma))
 \right)|0\rangle.
 \end{aligned}
\end{equation}
The product is taken in the displayed site order.
The images form an orthonormal basis of
$\cH_{V_0,V_2}$, so $\mathcal U_{V_0,V_2}$
is unitary.

For $\varnothing\ne X\subset\Lambda$, recall
that $q_x=(n_x-1)^2$ and define
\[
 P_X:=\prod_{x\in X}(1-q_x).
\]
This is the projection onto configurations
in which every site of $X$ is singly occupied.

For $A_X\in\mathfrak A_X^{\rm even}$,
let $\phi_X$ denote the restriction of
$P_XA_XP_X$ to the singly occupied local
space, identified with
$\cH_X^{\rm spin}=(\mathbb C^2)^{\otimes X}$.
If $X\subset R$, we write $\phi_X^{(R)}$
for the corresponding operator on
$\cH_R^{\rm spin}$, acting as the identity
on $R\setminus X$.

The following lemma shows that this local
spin compression is independent of the
charge configuration outside $X$.

\begin{lem}[Local spin compression in fixed-defect sectors]
\label{lem:defect-spin-compression}
Let $\varnothing\ne X\subset\Lambda$ and
$A_X\in\mathfrak A_X^{\rm even}$.
Then
\begin{equation}\label{eq:defect-compression-norm}
 \|\phi_X\|\le\|A_X\|.
\end{equation}
Moreover, $P_XA_XP_X$ preserves every
fixed-defect sector $\cH_{V_0,V_2}$, and
\begin{equation}\label{eq:defect-compression-sector}
 \begin{aligned}
 \mathcal U_{V_0,V_2}^*
 (P_XA_XP_X)\mathcal U_{V_0,V_2}
 =
 \begin{cases}
 \phi_X^{(R)},&X\subset R,\\
 0,&X\cap\mathcal D\ne\varnothing.
 \end{cases}
 \end{aligned}
\end{equation}
Here the operator on the left is understood
as the restriction to $\cH_{V_0,V_2}$.
In particular, $\phi_X^{(R)}$ depends only
on $A_X$ and $X$, and not on the
holon--doublon assignments outside $X$.
\end{lem}

The proof, based on local CAR matrix units,
is given in Appendix~\ref{app:defect-sector-new}.

\subsection{Uniform pressure comparison with the defect-free sector}
\label{subsec:defect-pressure-comparison}

We now estimate the contribution of charge-defect
sectors to the Gibbs pressure. The local spin
compression established in the preceding subsection
allows us to compare these sectors with the
defect-free spin system.

Let $\Phi=(\Phi_X)$ be an even interaction with
self-adjoint local terms satisfying
\begin{equation}\label{eq:defect-local-conservation}
 [\Phi_X,N_X]=[\Phi_X,D_X]=0.
\end{equation}
Consider
\[
 H=UD_\Lambda+\Phi_\Lambda
\]
on the canonical half-filled space, and let
$P=\mathbbm1_{\{D_\Lambda=0\}}$.
For a self-adjoint operator $K$ on a specified
Hilbert space, write
\[
 p_{\Lambda,\beta}(K)
 :=\frac1{\beta|\Lambda|}
   \log\Tr e^{-\beta K},
\]
where the trace is taken over that space.

\begin{lem}[Uniform defect-sector pressure comparison]
\label{lem:defect-pressure}
Set $\nu=\|\Phi\|_0$.
For every $U>0$, $\beta>0$, and every finite even
torus $\Lambda$,
\begin{equation}\label{eq:defect-pressure-target}
 \begin{aligned}
 0
 \le p_{\Lambda,\beta}(H)
       -p_{\Lambda,\beta}
       (PHP\restriction_{P\cH_\Lambda^{\rm hf}})
 \le\frac1\beta
       \log\bigl(1+e^{-\beta(U/2-2\nu)}\bigr).
 \end{aligned}
\end{equation}
In particular, if $U\ge8\nu$, then
\begin{equation}\label{eq:defect-pressure-decay}
 \begin{aligned}
 0
 \le p_{\Lambda,\beta}(H)
       -p_{\Lambda,\beta}
       (PHP\restriction_{P\cH_\Lambda^{\rm hf}})
 \le\beta^{-1}e^{-\beta U/4}.
 \end{aligned}
\end{equation}
\end{lem}

\begin{proof}
We divide the proof into four steps.

\medskip
\noindent
\emph{Step 1: a local operator lower bound.}
Recall from Subsection~\ref{subsec:defect-spin-compression}
that
\[
 P_X=\prod_{x\in X}(1-q_x),
 \qquad Q_X:=1-P_X.
\]
On the local fermionic Fock space,
\begin{equation}\label{eq:PX-joint-projection}
 P_X=\mathbbm1_{\{N_X=|X|,\,D_X=0\}}.
\end{equation}
Indeed, $D_X=0$ excludes doubly occupied sites,
and $N_X=|X|$ then forces single occupation
at every site of $X$.

By \eqref{eq:defect-local-conservation},
$P_X$ commutes with $\Phi_X$. Thus,
\[
 \begin{aligned}
 \Phi_X
 =P_X\Phi_XP_X+Q_X\Phi_XQ_X
 \ge P_X\Phi_XP_X-\|\Phi_X\|Q_X.
 \end{aligned}
\]
Since the $q_x$ are commuting projections,
$Q_X$ projects onto configurations containing
at least one charge defect in $X$. Hence
\[
 Q_X=1-\prod_{x\in X}(1-q_x)
 \le\sum_{x\in X}q_x.
\]
Consequently,
\[
 \begin{aligned}
 \Phi_\Lambda
 &\ge\sum_XP_X\Phi_XP_X
       -\sum_X\|\Phi_X\|Q_X\\
 &\ge\sum_XP_X\Phi_XP_X
       -\sum_{x\in\Lambda}
        \left(\sum_{X\ni x}\|\Phi_X\|\right)q_x\\
 &\ge\sum_XP_X\Phi_XP_X
       -\nu\sum_{x\in\Lambda}q_x.
 \end{aligned}
\]
Using the half-filling identity
$\sum_xq_x=2D_\Lambda$, we obtain
\begin{equation}\label{eq:defect-lower-H-K}
 H\ge K:=
 \sum_XP_X\Phi_XP_X
 +(U/2-\nu)\sum_{x\in\Lambda}q_x.
\end{equation}

\medskip\noindent
\emph{Step 2: trace decomposition into fixed-defect sectors.}
We first show that $K$ preserves the occupation
number at each site. If $x\in X$, then
\[
 n_xP_X=P_Xn_x=P_X,
\]
and therefore
$[n_x,P_X\Phi_XP_X]=0$.
For $x\notin X$, the same relation follows
from locality. Consequently,
\[
 [K,n_x]=[K,q_x]=0,
 \qquad x\in\Lambda.
\]
Thus, $K$ decomposes into the fixed-defect
sectors introduced in
Subsection~\ref{subsec:defect-spin-compression}.

Let $\mathcal D\subset\Lambda$ be a defect set,
and put
\[
 R=\Lambda\setminus\mathcal D,
 \qquad m=|\mathcal D|.
\]
At half filling, $m$ must be even, with
$m/2$ holons and $m/2$ doublons.
There are $\binom m{m/2}$ possible assignments.

For the interaction $\Phi$, let $\phi_X$
be the local spin compression defined in
Subsection~\ref{subsec:defect-spin-compression}.
Set
\[
 H_R^{\rm spin}
 :=\sum_{\varnothing\ne X\subset R}\phi_X^{(R)},
 \qquad
 Z_R^{\rm spin}
 :=\Tr_{\cH_R^{\rm spin}}
       e^{-\beta H_R^{\rm spin}}.
\]
In particular, $H_\Lambda^{\rm spin}$ and
$Z_\Lambda^{\rm spin}$ are defined by taking
$R=\Lambda$.

By Lemma~\ref{lem:defect-spin-compression},
the compressed interaction satisfies
\begin{equation}\label{eq:defect-compressed-interaction}
 \begin{aligned}
 \mathcal U_{V_0,V_2}^*
 \left(\sum_XP_X\Phi_XP_X\right)
 \mathcal U_{V_0,V_2}
 =H_R^{\rm spin},
 \end{aligned}
\end{equation}
where $V_0\cup V_2=\mathcal D$ and
$|V_0|=|V_2|=m/2$.
The right-hand side is independent of the
holon--doublon assignment.

Since $\sum_xq_x=m$ on this sector,
\eqref{eq:defect-lower-H-K} and
\eqref{eq:defect-compressed-interaction} show that
the restriction of $K$ is unitarily equivalent to
$
 (U/2-\nu)m\,\mathbbm1+H_R^{\rm spin}.
$
Summing over all defect configurations gives
\begin{equation}\label{eq:K-sector-trace-exact}
 \begin{aligned}
 \Tr_{\cH_\Lambda^{\rm hf}}e^{-\beta K}
 =
 \sum_{\substack{\mathcal D\subset\Lambda\\
                 |\mathcal D|\;\mathrm{even}}}
 \binom{|\mathcal D|}{|\mathcal D|/2}
 e^{-\beta(U/2-\nu)|\mathcal D|}
 Z_{\Lambda\setminus\mathcal D}^{\rm spin}.
 \end{aligned}
\end{equation}

\medskip\noindent
\emph{Step 3: comparison of the spin partition functions.}
We compare $Z_R^{\rm spin}$ with the partition
function on the full lattice.
Using the natural identification
$
 \cH_\Lambda^{\rm spin}
 \cong\cH_R^{\rm spin}\otimes
       \cH_{\mathcal D}^{\rm spin},
$
we have
\[
 \begin{aligned}
 H_\Lambda^{\rm spin}
 -(H_R^{\rm spin}\otimes\mathbbm1_{\mathcal D})
 =
 \sum_{\substack{\varnothing\ne X\subset\Lambda\\
                 X\cap\mathcal D\ne\varnothing}}
 \phi_X^{(\Lambda)}.
 \end{aligned}
\]
The contraction estimate
\eqref{eq:defect-compression-norm} and the
definition of $\nu$ imply
\[
 \begin{aligned}
 \left\|H_\Lambda^{\rm spin}
 -(H_R^{\rm spin}\otimes\mathbbm1_{\mathcal D})
 \right\|
 \le
 \sum_{\substack{X\subset\Lambda\\
                 X\cap\mathcal D\ne\varnothing}}
 \|\phi_X\|
 \le
 \sum_{x\in\mathcal D}\sum_{X\ni x}
 \|\Phi_X\|
 \le\nu m.
 \end{aligned}
\]
By the min--max principle, the preceding
operator-norm estimate implies
\[
 \begin{aligned}
 Z_\Lambda^{\rm spin}
 \ge e^{-\beta\nu m}
 \Tr_{\cH_\Lambda^{\rm spin}}
 e^{-\beta(H_R^{\rm spin}
                 \otimes\mathbbm1_{\mathcal D})}
 =e^{-\beta\nu m}2^mZ_R^{\rm spin}.
 \end{aligned}
\]
Hence
\begin{equation}\label{eq:removed-spin-trace-bound}
 Z_R^{\rm spin}
 \le2^{-m}e^{\beta\nu m}Z_\Lambda^{\rm spin}.
\end{equation}
The factor $2^{-m}$ accounts for the spin
degrees of freedom at the deleted sites.

\medskip\noindent
\emph{Step 4: summation over defect configurations.}
Let
\[
 Z_P:=
 \Tr_{P\cH_\Lambda^{\rm hf}}
 e^{-\beta(PHP\restriction_{P\cH_\Lambda^{\rm hf}})}.
\]
By \eqref{eq:defect-local-conservation},
$[\Phi_\Lambda,D_\Lambda]=0$, so $P$
reduces $H$. Consequently,
\[
 \Tr_{\cH_\Lambda^{\rm hf}}e^{-\beta H}
 \ge Z_P.
\]
On the defect-free sector, $D_\Lambda=0$,
and Lemma~\ref{lem:defect-spin-compression}
identifies the $P$-block of $H$ with
$H_\Lambda^{\rm spin}$. Thus
$
 Z_P=Z_\Lambda^{\rm spin}.
$

On the other hand, the operator lower bound
\eqref{eq:defect-lower-H-K} and the
finite-dimensional min--max principle imply
\[
 \Tr e^{-\beta H}\le\Tr e^{-\beta K}.
\]
Combining \eqref{eq:K-sector-trace-exact}
with \eqref{eq:removed-spin-trace-bound}
and setting
$
 a:=e^{-\beta(U/2-2\nu)},
$
we obtain
\[
 \begin{aligned}
 \frac{\Tr e^{-\beta H}}{Z_P}
 \le
 \sum_{\substack{\mathcal D\subset\Lambda\\
                 |\mathcal D|\;\mathrm{even}}}
 \binom{|\mathcal D|}{|\mathcal D|/2}
 2^{-|\mathcal D|}a^{|\mathcal D|}
 \le\sum_{\mathcal D\subset\Lambda}
       a^{|\mathcal D|}
 =(1+a)^{|\Lambda|},
 \end{aligned}
\]
where we used $\binom m{m/2}\le2^m$.

Taking logarithms and dividing by
$\beta|\Lambda|$ proves
\eqref{eq:defect-pressure-target}.

Finally, if $U\ge8\nu$, then
$
 U/2-2\nu\ge U/4.
$
The inequality $\log(1+s)\le s$ for $s\ge0$
therefore yields
\[
 \begin{aligned}
 p_{\Lambda,\beta}(H)
 -p_{\Lambda,\beta}
 (PHP\restriction_{P\cH_\Lambda^{\rm hf}})
\le
 \beta^{-1}\log(1+e^{-\beta U/4})
 \le\beta^{-1}e^{-\beta U/4},
 \end{aligned}
\]
which proves \eqref{eq:defect-pressure-decay}.
\end{proof}

\section{The second-order \texorpdfstring{$P$}{P}-block and its Heisenberg form}
\label{sec:H2-Heis}

In this section, we identify the restriction of
the second-order effective interaction to the
singly occupied sector with an antiferromagnetic
Heisenberg Hamiltonian.
We compute the effective exchange coupling and
staggered field, including their dependence on
the original magnetic field.
This identification will be used in the
Hubbard--Heisenberg pressure comparison.

\subsection{Canonical identification with the spin space}
\label{subsec:spin-rep}

Fix an ordering $\Lambda=\{x_1,\ldots,x_{|\Lambda|}\}$.
For
$\boldsymbol\sigma=(\sigma_x)_{x\in\Lambda}
\in\{\uparrow,\downarrow\}^{\Lambda}$, set
\[
 |\boldsymbol\sigma\rangle_{\rm f}
 :=c_{x_1\sigma_{x_1}}^*\cdots
   c_{x_{|\Lambda|}\sigma_{x_{|\Lambda|}}}^*|0\rangle.
\]
These vectors form an orthonormal basis of
$P\cH_\Lambda^{\rm hf}$.
We identify this space with
$\cH_\Lambda^{\rm spin}
=\bigotimes_{x\in\Lambda}\mathbb C^2$
by the unitary
\begin{equation}\label{eq:def-U-spin-identification}
 \begin{aligned}
  \mathcal U_\Lambda:
  P\cH_\Lambda^{\rm hf}
  &\longrightarrow \cH_\Lambda^{\rm spin},\\
  |\boldsymbol\sigma\rangle_{\rm f}
  &\longmapsto
       \bigotimes_{j=1}^{|\Lambda|}
       |\sigma_{x_j}\rangle.
 \end{aligned}
\end{equation}
In particular,
\begin{equation}\label{eq:P-M-spin-identification}
 \mathcal U_\Lambda PM_\Lambda P\mathcal U_\Lambda^*
 =M_\Lambda^{\rm spin}
 =\sum_{x\in\Lambda}\eta_xS_x^{(3)}.
\end{equation}

For distinct sites $x,y$, let $P_{xy}$ be the unitary
exchanging the two spin factors.
The singlet projection on the bond $e=\{x,y\}$ is
\begin{equation}\label{eq:def-Bxy-spin}
 B_{xy}:=\frac12(\mathbbm 1-P_{xy})
       =\frac14-\boldsymbol S_x\cdot\boldsymbol S_y.
\end{equation}
Accordingly,
\begin{equation}\label{eq:Heis-singlet-form}
 H_\Lambda^{\rm Heis}(J,h)
 =-J\sum_{e=\{x,y\}\in\mathscr B_\Lambda}B_{xy}
  -hM_\Lambda^{\rm spin}.
\end{equation}

\subsection{Vanishing of cross-bond processes}
\label{subsec:cross-bond}

Although the second-order interaction may involve
distinct overlapping bonds, only same-bond terms
survive upon compression to the singly occupied sector.

\begin{lem}[Cross-bond cancellation]
\label{lem:cross-bond-P}
For each bond $e=\{u,v\}\in\mathscr B_\Lambda$, let
\[
 S_{1,e}:=\mathcal I_{h,e}((T_e)^{\rm off}).
\]
Then, for distinct bonds $e,f\in\mathscr B_\Lambda$,
\begin{equation}\label{eq:cross-bond-products-vanish}
 P S_{1,e}(T_f)^{\rm off}P
 =P(T_f)^{\rm off}S_{1,e}P=0.
\end{equation}
In particular,
\[
   P\operatorname{ad}_{S_{1, e}}(T_f^{\rm off})P=0.
\]
\end{lem}

\begin{proof}
For any bond $e=\{u,v\}$, let
\[
 N_e=n_u+n_v,
 \qquad
 D_e=n_{u\uparrow}n_{u\downarrow}
     +n_{v\uparrow}n_{v\downarrow}.
\]
Since $P$ projects onto the singly occupied sector,
\begin{equation}\label{eq:cross-bond-initial-sector}
 N_eP=2P,\qquad D_eP=0.
\end{equation}

Both $S_{1,e}$ and $(T_e)^{\rm off}$ commute with $N_e$
and have only nonzero $D_e$-grades.
On the local $N_e=2$ subspace, the only possible
values of $D_e$ are $0$ and $1$.
Therefore, for $A_e=S_{1,e}$ or
$A_e=(T_e)^{\rm off}$,
\begin{equation}\label{eq:cross-bond-image-sector}
 A_eP
 =\mathbbm1_{\{N_e=2,\,D_e=1\}}A_eP.
\end{equation}

On this local sector the occupations are necessarily
$(n_u,n_v)=(2,0)$ or $(0,2)$.
Hence
\[
 \begin{aligned}
 q_u\mathbbm1_{\{N_e=2,\,D_e=1\}}
 &=\mathbbm1_{\{N_e=2,\,D_e=1\}},\\
 q_v\mathbbm1_{\{N_e=2,\,D_e=1\}}
 &=\mathbbm1_{\{N_e=2,\,D_e=1\}},
 \end{aligned}
\]
and \eqref{eq:cross-bond-image-sector} implies
\begin{equation}\label{eq:cross-bond-defect-at-endpoints}
 q_uA_eP=q_vA_eP=A_eP.
\end{equation}

Now let $e\ne f$. As distinct unordered bonds
both contain two sites, there exist
$
 x\in e\setminus f
 $
 and 
 $
 y\in f\setminus e.
$
Since $x\notin f$ and $(T_f)^{\rm off}$ is an even
operator supported on $f$,
$
 [q_x,(T_f)^{\rm off}]=0.
$
By \eqref{eq:cross-bond-defect-at-endpoints},
\[
 \begin{aligned}
 P(T_f)^{\rm off}S_{1,e}P
 &=P(T_f)^{\rm off}q_xS_{1,e}P\\
 &=Pq_x(T_f)^{\rm off}S_{1,e}P\\
 &=0,
 \end{aligned}
\]
because $Pq_x=0$.

Similarly, $y\notin e$ gives
$[q_y,S_{1,e}]=0$, so
$
 PS_{1,e}(T_f)^{\rm off}P=0
 $.
These two identities prove the assertion.
\end{proof}

\subsection{Intermediate effective parameters}

For $|h|\le h_0$ and $U\ge U_0$ as in
Proposition~\ref{prop:finite-order-SW}, define
\begin{equation}\label{eq:def-J-heff}
 \begin{aligned}
 J(h):=\frac{4t^2U}{U^2-h^2},\qquad
 b(h):=\frac{2ht^2}{U^2-h^2},\qquad
 h_{{\rm eff},\Lambda}(h):=h-z_\Lambda b(h),
 \end{aligned}
\end{equation}
where $z_\Lambda$ denotes the number of distinct
nearest neighbours of each site, under the
convention that every unordered bond is counted once.

For $\Lambda=(\mathbb Z/L\mathbb Z)^d$ with even $L$,
\begin{equation}\label{eq:torus-coordination-number}
 z_\Lambda=
 \begin{cases}
  2d,&L\ge4,\\
  d,&L=2.
 \end{cases}
\end{equation}

\begin{prop}[Exact identification of the second-order spin block]
\label{prop:H2-Heis-exact}
Under the above assumptions, one has the operator identity
\begin{equation}\label{eq:H2-to-Heis-prop}
 \mathcal U_\Lambda
 P(-hM_\Lambda+H_\Lambda^{(2)}(h))P
 \mathcal U_\Lambda^*
 =
 H_\Lambda^{\rm Heis}
 \bigl(J(h),h_{{\rm eff},\Lambda}(h)\bigr)
\end{equation}
on $\cH_\Lambda^{\rm spin}$.
\end{prop}

\begin{proof}
Recall from \eqref{eq:SW-H2} that
\[
 H^{(2)}(h)
 =\frac12\bigl(
   \operatorname{ad}_{S_1}(T^{\rm off})
   \bigr)^{\rm diag}.
\]

Using
\[
 S_{1,\Lambda}
 =\sum_{e\in\mathscr B_\Lambda}S_{1,e},
 \qquad
 T_\Lambda^{\rm off}
 =\sum_{f\in\mathscr B_\Lambda}(T_f)^{\rm off},
\]
we obtain
\begin{equation}\label{eq:H2-bond-decomposition}
 \begin{aligned}
 PH_\Lambda^{(2)}(h)P
 &=\frac12P\bigl(
   \operatorname{ad}_{S_1}(T^{\rm off})
   \bigr)^{\rm diag}P\\
 &=\frac12
   \sum_{e,f\in\mathscr B_\Lambda}
       P\bigl(
       \operatorname{ad}_{S_{1, e}}(T_f^{\rm off})
          \bigr)^{\rm diag}
       P\\
 &=\sum_{e\in\mathscr B_\Lambda}X_e(h),
 \end{aligned}
\end{equation}
where Lemma~\ref{lem:cross-bond-P} was used
in the last line and
\[
 X_e(h)
 :=\frac12P
       \operatorname{ad}_{S_{1, e}}(T_e^{\rm off})P.
\]

This is precisely the bond operator of
Lemma~\ref{lem:two-site-computation}:
the diagonal projection in its definition can
be omitted between two $P$'s by
\eqref{eq:P-compression-grade-zero}.

For $e=\{x,y\}$, write
\[
 M_e^{\rm spin}
 :=\eta_xS_x^{(3)}+\eta_yS_y^{(3)}.
\]
To recall the coefficient identification, choose
the labelling of the endpoints so that
$\eta_x=+1$ and $\eta_y=-1$.

By the two-site calculation in
Lemma~\ref{lem:two-site-computation}, the matrix
of $X_e(h)$ on
$
 \operatorname{span}
 \{|\!\uparrow\downarrow\rangle,
   |\!\downarrow\uparrow\rangle\},
$
after spin identification, is
\begin{equation}\label{eq:H2-two-site-matrix}
 \begin{aligned}
 &\begin{pmatrix}
 -\dfrac{2t^2}{U+h}
   &\dfrac{2Ut^2}{U^2-h^2}\\[2mm]
 \dfrac{2Ut^2}{U^2-h^2}
   &-\dfrac{2t^2}{U-h}
 \end{pmatrix}
=
 -\frac{J(h)}2
 \begin{pmatrix}
 1&-1\\
 -1&1
 \end{pmatrix}
 +b(h)
 \begin{pmatrix}
 1&0\\
 0&-1
 \end{pmatrix}.
 \end{aligned}
\end{equation}

Both sides vanish on the parallel-spin vectors
$|\uparrow\uparrow\rangle$ and
$|\downarrow\downarrow\rangle$.
Since $B_{xy}$ and $M_e^{\rm spin}$ have the
indicated $2\times2$ matrices, this is precisely
\eqref{eq:two-site-computation-formula}, namely
\begin{equation}\label{eq:H2-bond-spin-form}
 \mathcal U_\Lambda X_e(h)\mathcal U_\Lambda^*
 =-J(h)B_{xy}+b(h)M_e^{\rm spin}.
\end{equation}

Summing \eqref{eq:H2-bond-spin-form}
over all bonds and using
\eqref{eq:H2-bond-decomposition} yields
\begin{equation}\label{eq:H2-global-spin-form}
 \begin{aligned}
 \mathcal U_\Lambda
   PH_\Lambda^{(2)}(h)P
   \mathcal U_\Lambda^*
   =
 -J(h)\sum_{e=\{x,y\}\in\mathscr B_\Lambda}B_{xy}
 +b(h)\sum_{e\in\mathscr B_\Lambda}M_e^{\rm spin}.
 \end{aligned}
\end{equation}

Each site belongs to exactly $z_\Lambda$ distinct bonds.
Thus
\begin{equation}\label{eq:bond-field-counting}
 \begin{aligned}
 \sum_{e=\{x,y\}\in\mathscr B_\Lambda}
          M_e^{\rm spin}
 &\quad=
 \sum_{e=\{x,y\}\in\mathscr B_\Lambda}
   \bigl(\eta_xS_x^{(3)}
         +\eta_yS_y^{(3)}\bigr)\\
 &\quad=
 \sum_{x\in\Lambda}
 \bigl|\{e\in\mathscr B_\Lambda:x\in e\}\bigr|
       \eta_xS_x^{(3)}\\
 &\quad=z_\Lambda M_\Lambda^{\rm spin}.
 \end{aligned}
\end{equation}

Combining
\eqref{eq:P-M-spin-identification},
\eqref{eq:H2-global-spin-form}, and
\eqref{eq:bond-field-counting}, we arrive at
\[
 \begin{aligned}
 \mathcal U_\Lambda
   P(-hM_\Lambda+H_\Lambda^{(2)}(h))P
   \mathcal U_\Lambda^*
   &=
 -J(h)\sum_{e=\{x,y\}\in\mathscr B_\Lambda}B_{xy}
 -(h-z_\Lambda b(h))M_\Lambda^{\rm spin}\\
 &=
 H_\Lambda^{\rm Heis}
 \bigl(J(h),h_{{\rm eff},\Lambda}(h)\bigr),
 \end{aligned}
\]
which proves \eqref{eq:H2-to-Heis-prop}.
\end{proof}

In particular, the effective field is
\[
 h_{{\rm eff},\Lambda}(h)=
 \begin{cases}
 \displaystyle
 h-\frac{4dht^2}{U^2-h^2},&L\ge4,\\[2mm]
 \displaystyle
 h-\frac{2dht^2}{U^2-h^2},&L=2.
 \end{cases}
\]
The distinction for $L=2$ is necessary
for the exact finite-volume identity under
the unordered-bond convention.

\section{Pressure comparison and proof assembly}
\label{sec:pressure-assembly-revision}

In this section, we combine the three-step
Schrieffer--Wolff normal form, the defect-sector
pressure estimate, and the identification of the
effective spin Hamiltonian obtained above.
These results allow us to compare the Hubbard
and Heisenberg pressures with an error uniform
in the volume, completing the proof of
Theorem~\ref{thm:pressure-comparison-U3-face}.

\subsection{Finite-dimensional pressure perturbation bound}
\label{subsec:pressure-Lipschitz-revision}

We begin with an elementary perturbation bound
that allows us to translate operator-norm estimates
into uniform bounds on the Gibbs pressure.

\begin{lem}[Pressure Lipschitz bound]
\label{lem:pressure-Lipschitz-bound}
Let $H,K$ be self-adjoint operators on the same
finite-dimensional Hilbert space, and let $|\Lambda|$
be the pressure normalization. Then, for every $\beta>0$,
\[
 \left|
 p_{\Lambda,\beta}(H+K)-p_{\Lambda,\beta}(H)
 \right|
 \le\frac{\|K\|}{|\Lambda|}.
\]
In particular, if $K=K_\Lambda$ is the finite-volume
sum of an interaction, then
\[
 \left|
 p_{\Lambda,\beta}(H+K_\Lambda)
 -p_{\Lambda,\beta}(H)
 \right|
 \le\|K\|_0.
\]
\end{lem}

\begin{proof}
Since $K$ is self-adjoint,
$
 -\|K\|\mathbbm1\le K\le\|K\|\mathbbm1.
$
Consequently,
\[
 H-\|K\|\mathbbm1
 \le H+K
 \le H+\|K\|\mathbbm1.
\]
By the finite-dimensional min--max principle,
\[
 \begin{aligned}
 e^{-\beta\|K\|}\Tr e^{-\beta H}
 \le\Tr e^{-\beta(H+K)}
 \le e^{\beta\|K\|}\Tr e^{-\beta H}.
 \end{aligned}
\]
Taking logarithms and dividing by $\beta|\Lambda|$
gives
\[
 \left|
 p_{\Lambda,\beta}(H+K)-p_{\Lambda,\beta}(H)
 \right|
 \le\frac{\|K\|}{|\Lambda|}.
\]
The final assertion follows from
Lemma~\ref{lem:interaction-to-operator}, since
$
 \|K_\Lambda\|\le|\Lambda|\|K\|_0.
$
\end{proof}

\subsection{Mismatch of the Heisenberg parameters}

We next quantify the difference between the effective
Heisenberg parameters obtained above and those of the
reference model, and estimate its effect on the Gibbs pressure.

\begin{lem}[Explicit parameter mismatch]
\label{lem:parameter-mismatch-new}
Let $h_0\ge0$ and $U>h_0$.
Uniformly in $|h|\le h_0$ and in the finite
even torus $\Lambda$, one has
\begin{align}
 |J(h)-J_0(U)|
 &\le
 \frac{4t^2h_0^2}{U(U^2-h_0^2)},\\
 |h_{{\rm eff},\Lambda}(h)-h|
 &\le
 \frac{4dh_0t^2}{U^2-h_0^2}.
\end{align}
Consequently, for every $\beta>0$,
\begin{align}
 \left|
 p_{\Lambda,\beta}^{\rm Heis}
       (J(h),h_{{\rm eff},\Lambda}(h))
 -p_{\Lambda,\beta}^{\rm Heis}(J_0(U),h)
 \right|\notag
 \le
 d\,|J(h)-J_0(U)|
 +\frac12|h_{{\rm eff},\Lambda}(h)-h|.
 \label{eq:Heis-parameter-pressure}
\end{align}
\end{lem}

\begin{proof}
By \eqref{eq:def-J-heff} and
$J_0(U)=4t^2/U$,
\[
 \begin{aligned}
 J(h)-J_0(U)
 =\frac{4t^2U}{U^2-h^2}
   -\frac{4t^2}{U}
 =\frac{4t^2h^2}{U(U^2-h^2)}.
 \end{aligned}
\]
Since $U>h_0$ and $|h|\le h_0$,
\[
 |J(h)-J_0(U)|
 \le\frac{4t^2h_0^2}{U(U^2-h_0^2)}.
\]

Similarly, using $z_\Lambda\le2d$,
\[
 \begin{aligned}
 |h_{{\rm eff},\Lambda}(h)-h|
 =z_\Lambda|b(h)|
 =\frac{2z_\Lambda|h|t^2}{U^2-h^2}
 \le\frac{4dh_0t^2}{U^2-h_0^2}.
 \end{aligned}
\]

Next, by \eqref{eq:Heis-singlet-form},
\[
 \begin{aligned}
 &H_\Lambda^{\rm Heis}
       (J(h),h_{{\rm eff},\Lambda}(h))
 -H_\Lambda^{\rm Heis}(J_0(U),h)\\
 &\quad=
 -(J(h)-J_0(U))
       \sum_{e\in\mathscr B_\Lambda}B_e
 -(h_{{\rm eff},\Lambda}(h)-h)
       M_\Lambda^{\rm spin}.
 \end{aligned}
\]
Since
\[
 \|B_e\|=1,\qquad
 \|M_\Lambda^{\rm spin}\|\le\frac{|\Lambda|}{2},
 \qquad
 |\mathscr B_\Lambda|\le d|\Lambda|,
\]
we obtain
\[
 \begin{aligned}
 &\bigl\|
 H_\Lambda^{\rm Heis}
       (J(h),h_{{\rm eff},\Lambda}(h))
 -H_\Lambda^{\rm Heis}(J_0(U),h)
 \bigr\|\\
 &\quad\le|\Lambda|
 \left(
 d|J(h)-J_0(U)|
 +\frac12|h_{{\rm eff},\Lambda}(h)-h|
 \right).
 \end{aligned}
\]
Applying Lemma~\ref{lem:pressure-Lipschitz-bound},
we obtain the asserted pressure estimate.
\end{proof}

\subsection{Proof of the third-order pressure comparison}
\label{subsec:pressure-comparison-U3}

We now apply the three-step normal form to the Gibbs pressure.
The bipartite cancellation \eqref{eq:SW-PC3P} removes the
third-order contribution on the singly occupied sector,
while Lemma~\ref{lem:defect-pressure} controls the
remaining charge sectors.

\begin{prop}[Third-order Hubbard--Heisenberg pressure comparison]
\label{prop:pressure-comparison-U3}
Fix $d\ge1$, $t\ne0$, $h_0\ge0$ and $\ell_0>0$.
There exist $U_0<\infty$ and $C_1,C_2<\infty$, depending
only on these fixed parameters, such that for every finite
even torus $\Lambda$, every $U\ge U_0$, $|h|\le h_0$,
and $\beta J_0(U)\ge\ell_0$, with $J_0(U)=4t^2/U$,
\begin{equation}\label{eq:pressure-U3-small-field}
 \left|p_{\Lambda,\beta,U}^{\rm Hub}(h)
 -p_{\Lambda,\beta}^{\rm Heis}(J_0(U),h)\right|
 \le\frac{C_1}{U^3}+\frac{C_2|h|}{U^2}.
\end{equation}
Moreover, for each fixed $R<\infty$, there exist
$U_R,C_R<\infty$ such that whenever
$U\ge U_R$, $|x|\le R$ and $\beta J_0(U)\ge\ell_0$,
\begin{equation}\label{eq:pressure-U3-Heis-scale}
 \left|p_{\Lambda,\beta,U}^{\rm Hub}(J_0(U)x)
 -p_{\Lambda,\beta}^{\rm Heis}(J_0(U),J_0(U)x)\right|
 \le\frac{C_R}{U^3}.
\end{equation}
The estimates are uniform in the volume and the stated
temperature regions.
\end{prop}

\begin{proof}
Let
\[
 H_\Lambda^{\rm tr,[3]}(h)
 :=UD_\Lambda+\Phi_\Lambda^{[3]}(h)
 =H_{0,\Lambda}(h)+T_\Lambda^{(0)}
     +H_\Lambda^{(2)}(h)+C_{3,\Lambda}^{\rm diag}.
\]
By Lemma~\ref{lem:third-order-local-conservation},
the hypotheses of Lemma~\ref{lem:defect-pressure}
are satisfied. Moreover, from
\eqref{eq:SW-K-bounds}, \eqref{eq:SW-F2-explicit}
and \eqref{eq:SW-F3-C3-bounds},
\begin{equation}\label{eq:Phi3-nu-bound}
 \begin{aligned}
 \nu_3(h):=\|\Phi^{[3]}(h)\|_0
 &\le\frac{h_0}{2}+K_0+f_2+c_3\\
 &\le\frac{h_0}{2}+K_0
    +\frac{32\pi K_0^2}{U}
    +\frac{3040\pi^2K_0^3}{3U^2}.
 \end{aligned}
\end{equation}
For $U\ge1$, the right-hand side is bounded by the
fixed number
\[
 \nu_{3,*}:=\frac{h_0}{2}+K_0+32\pi K_0^2
                   +\frac{3040\pi^2K_0^3}{3}.
\]
Choose $U_0$ large enough that the hypotheses of
Propositions~\ref{prop:finite-order-SW}
and~\ref{prop:third-order-SW} both hold and
\[
 U_0\ge\max\{1,\,2h_0,\,8\nu_{3,*}\}.
\]
By Proposition~\ref{prop:third-order-SW},
the three-step normal form
\eqref{eq:SW-third-normal-form} gives
\[
 W_\Lambda^{[3]}(h)H_\Lambda^{\rm Hub}(h)
 W_\Lambda^{[3]}(h)^*
 =
 H_\Lambda^{\rm tr,[3]}(h)+E_\Lambda^{[3]}(h).
\]
Since $W_\Lambda^{[3]}(h)$ is unitary and
preserves the canonical half-filled space,
the invariance of the trace yields
\[
 p_{\Lambda,\beta,U}^{\rm Hub}(h)
 =
 p_{\Lambda,\beta}
 \bigl(H_\Lambda^{\rm tr,[3]}(h)
       +E_\Lambda^{[3]}(h)\bigr).
\]
Applying Lemma~\ref{lem:pressure-Lipschitz-bound}
and the remainder estimate
\eqref{eq:SW-third-remainder-bound}, we obtain
\begin{equation}\label{eq:pressure3-SW-part}
 \begin{aligned}
 &\left|
 p_{\Lambda,\beta,U}^{\rm Hub}(h)
 -p_{\Lambda,\beta}(H_\Lambda^{\rm tr,[3]}(h))
 \right|\le
 \frac{\|E_\Lambda^{[3]}(h)\|}{|\Lambda|}
 \le C_{\rm SW,3}\frac{K_0^4}{U^3}.
 \end{aligned}
\end{equation}

The defect-sector comparison yields
\begin{equation}\label{eq:pressure3-defect-part}
 \begin{aligned}
 0 \le p_{\Lambda,\beta}(H_\Lambda^{\rm tr,[3]}(h))
 -p_{\Lambda,\beta}
    (PH_\Lambda^{\rm tr,[3]}(h)P
             \restriction_{P\cH_\Lambda^{\rm hf}})
  \le\beta^{-1}e^{-\beta U/4}.
 \end{aligned}
\end{equation}
By \eqref{eq:SW-PC3P}, $PC_{3,\Lambda}^{\rm diag}P=0$.
Together with $PT_\Lambda^{(0)}P=0$ and
Proposition~\ref{prop:H2-Heis-exact}, this implies
\[
 \mathcal U_\Lambda
 (PH_\Lambda^{\rm tr,[3]}(h)P
             \restriction_{P\cH_\Lambda^{\rm hf}})
 \mathcal U_\Lambda^*
 =H_\Lambda^{\rm Heis}(J(h),h_{{\rm eff},\Lambda}(h)).
\]
Combining \eqref{eq:pressure3-SW-part} and
\eqref{eq:pressure3-defect-part} with
Lemma~\ref{lem:parameter-mismatch-new}, we obtain
\begin{equation}\label{eq:pressure3-quantitative}
 \begin{aligned}
 &\left|p_{\Lambda,\beta,U}^{\rm Hub}(h)
 -p_{\Lambda,\beta}^{\rm Heis}(J_0(U),h)\right|\\
 &\quad\le C_{\rm SW,3}\frac{K_0^4}{U^3}
       +\beta^{-1}e^{-\beta U/4}
+\frac{4dt^2h^2}{U(U^2-h^2)}
       +\frac{z_\Lambda|h|t^2}{U^2-h^2}.
 \end{aligned}
\end{equation}
Because $U\ge2h_0$ and $z_\Lambda\le2d$,
\[
 \begin{aligned}
 \frac{4dt^2h^2}{U(U^2-h^2)}
 \le\frac{16dt^2h_0^2}{3U^3},\qquad
 \frac{z_\Lambda|h|t^2}{U^2-h^2}
 \le\frac{8dt^2|h|}{3U^2}.
 \end{aligned}
\]
The temperature condition gives
$\beta\ge\ell_0U/(4t^2)$ and therefore
\[
 \begin{aligned}
 \beta^{-1}e^{-\beta U/4}
 \le\frac{4t^2}{\ell_0U}
   \exp\left(-\frac{\ell_0U^2}{16t^2}\right)
 \le\frac{64t^4}{e\ell_0^2U^3},
 \end{aligned}
\]
where the last inequality follows from
$y e^{-y}\le e^{-1}$ for $y\ge0$.
Substitution into \eqref{eq:pressure3-quantitative}
proves \eqref{eq:pressure-U3-small-field}, for example with
\[
 C_1=C_{\rm SW,3}K_0^4
       +\frac{16dt^2h_0^2}{3}
       +\frac{64t^4}{e\ell_0^2},
 \qquad C_2=\frac{8dt^2}{3}.
\]
Finally, choose $U_R\ge U_0$ large enough that
$RJ_0(U)\le h_0$ for $U\ge U_R$ when $h_0>0$.
Then $|h|=|J_0(U)x|\le 4Rt^2/U$, so
\eqref{eq:pressure-U3-small-field} gives
\[
 \left|p_{\Lambda,\beta,U}^{\rm Hub}(J_0(U)x)
 -p_{\Lambda,\beta}^{\rm Heis}(J_0(U),J_0(U)x)\right|
 \le\frac{C_1+4Rt^2C_2}{U^3}.
\]
If $h_0=0$, the general fixed-field statement gives
only $h=0$; the assertion for $|x|\le R$ follows instead
by applying the preceding proof with any positive fixed
field window $h_0'>0$. This is permitted because the
model and the first comparison are available for any
fixed $h_0'>0$, with corresponding constants.
\end{proof}

\subsection{Proof of the charge-sector theorem}
\label{subsec:proof-charge-sector-face}

We first derive the improved charge-sector estimate
from the third-order pressure comparison.

\begin{prop}[Quadratic suppression of charge defects]
\label{prop:charge-U2-from-pressure}
Fix $h_0>0$ and $\ell_0>0$, and let
$U_0,C_1,C_2$ be the constants in
Theorem~\ref{thm:pressure-comparison-U3-face}.
Set
\[
 \mathcal E_U(h)
 :=\frac{8dt^2}{U^2}
   +\frac{10C_2|h|}{U^3}
   +\frac{18C_1}{U^4}.
\]
Then, for every $U\ge2U_0$, $|h|\le h_0$,
and $\beta>0$ satisfying
$\beta J_0(U)\ge\ell_0$,
\begin{equation}\label{eq:charge-U2-doublon}
 \sup_\Lambda
 \frac1{|\Lambda|}
 \omega_{\Lambda,\beta,U,h}^{\rm Hub}(D_\Lambda)
 \le\mathcal E_U(h).
\end{equation}
Consequently,
\begin{equation}\label{eq:charge-U2-consequences}
 \begin{aligned}
 &\sup_\Lambda\frac1{|\Lambda|}
 \sum_{x\in\Lambda}
 \omega_{\Lambda,\beta,U,h}^{\rm Hub}(q_x)
 \le2\mathcal E_U(h),\\
 &\sup_\Lambda\frac1{|\Lambda|}
 \sum_{x\in\Lambda}
 \omega_{\Lambda,\beta,U,h}^{\rm Hub}
 (n_{x\uparrow}n_{x\downarrow})
 \le\mathcal E_U(h),\\
 &\sup_\Lambda
 \frac{\omega_{\Lambda,\beta,U,h}^{\rm Hub}
 ((C_\Lambda^{\rm ch})^2)}{|\Lambda|^2}
 \le2\mathcal E_U(h).
 \end{aligned}
\end{equation}
\end{prop}

\begin{proof}
Fix $\Lambda$, $\beta$, and $h$, and define
\[
 f(V):=p_{\Lambda,\beta,V}^{\rm Hub}(h),
 \qquad
 g(V):=p_{\Lambda,\beta}^{\rm Heis}(J_0(V),h),
 \qquad V>0.
\]
The parameters $\Lambda,\beta,h$ are kept fixed
as $V$ varies.

Differentiating the finite-volume Gibbs pressure
with respect to $V$, we obtain
\[
 f'(V)
 =-\frac1{|\Lambda|}
 \omega_{\Lambda,\beta,V,h}^{\rm Hub}(D_\Lambda).
\]
Moreover, $f$ is convex in $V$ by the Gibbs
variational principle. Applying the corresponding
secant inequality between $U/2$ and $U$ gives
\begin{equation}\label{eq:charge-U2-convex-secant}
 \begin{aligned}
 \frac1{|\Lambda|}
 \omega_{\Lambda,\beta,U,h}^{\rm Hub}(D_\Lambda)
 =-f'(U)
 \le\frac{2}{U}
 \bigl(f(U/2)-f(U)\bigr).
 \end{aligned}
\end{equation}

Since
$
 J_0(U/2)=2J_0(U),
$
the condition $\beta J_0(U)\ge\ell_0$
implies $\beta J_0(U/2)\ge\ell_0$.
Thus, for $U\ge2U_0$,
Theorem~\ref{thm:pressure-comparison-U3-face}
applies at both $V=U$ and $V=U/2$,
with the same $\beta$ and $h$.
It follows that
\begin{equation}\label{eq:charge-U2-pressure-split}
 \begin{aligned}
 f(U/2)-f(U)
 &\le |f(U/2)-g(U/2)|\
 +|g(U/2)-g(U)|
       +|g(U)-f(U)|\\
 &\le |g(U/2)-g(U)|
       +\frac{9C_1}{U^3}
       +\frac{5C_2|h|}{U^2}.
 \end{aligned}
\end{equation}

To estimate the remaining difference, recall that
$
 \left\|
 \boldsymbol S_x\cdot\boldsymbol S_y-\frac14
 \right\|=1.
$
Since $|\mathscr B_\Lambda|\le d|\Lambda|$,
the Heisenberg pressure is Lipschitz continuous
in the exchange coupling, with
\[
 \begin{aligned}
 |g(U/2)-g(U)|
 &\le
 \frac{|J_0(U/2)-J_0(U)|}{|\Lambda|}
 \left\|
 \sum_{\{x,y\}\in\mathscr B_\Lambda}
 \left(\boldsymbol S_x\cdot\boldsymbol S_y
       -\frac14\right)
 \right\|\\
 &\le d|J_0(U/2)-J_0(U)|\\
 &=\frac{4dt^2}{U}.
 \end{aligned}
\]
Combining this with
\eqref{eq:charge-U2-convex-secant} and
\eqref{eq:charge-U2-pressure-split}, we obtain
\[
 \begin{aligned}
 \frac1{|\Lambda|}
 \omega_{\Lambda,\beta,U,h}^{\rm Hub}(D_\Lambda)
 \le\frac2U\left(
 \frac{4dt^2}{U}
 +\frac{9C_1}{U^3}
 +\frac{5C_2|h|}{U^2}\right)
 =\mathcal E_U(h).
 \end{aligned}
\]
This proves \eqref{eq:charge-U2-doublon}.

On the canonical half-filled space,
$
 \sum_{x\in\Lambda}q_x=2D_\Lambda.
$
Furthermore, since the occupation-number
operators commute, the Cauchy--Schwarz
inequality gives
\[
 \begin{aligned}
 (C_\Lambda^{\rm ch})^2
 =\left(\sum_x\eta_x(n_x-1)\right)^2
 \le|\Lambda|\sum_x(n_x-1)^2
 =|\Lambda|\sum_xq_x.
 \end{aligned}
\]
Taking Gibbs expectations proves
\eqref{eq:charge-U2-consequences}.
\end{proof}

We now prove the remaining assertions of the
charge-sector theorem, including the bound
valid without the strong-coupling assumption.

\begin{proof}[Proof of Theorem~\ref{thm:charge-sector-suppression-face}]
The $O(U^{-2})$ estimates in
\eqref{eq:charge-consequences} follow from
Proposition~\ref{prop:charge-U2-from-pressure}.
We therefore turn to the general bound
\eqref{eq:charge-direct-bound}.

Let $N:=|\Lambda|$, and let $\rho$ be the
Gibbs density matrix of $H_\Lambda^{\rm Hub}(h)$
on $\cH_\Lambda^{\rm hf}$.
The local occupation-number operators satisfy
$
 |S_x^{(3)}|=\frac12(1-q_x),
$
while half filling implies
$
 \sum_{x\in\Lambda}q_x=2D_\Lambda.
$
Consequently,
\[
 \begin{aligned}
 -hM_\Lambda
 \ge-|h|\sum_x|S_x^{(3)}|
 =-\frac{|h|N}{2}+|h|D_\Lambda.
 \end{aligned}
\]
Together with $\|T_\Lambda\|\le K_0N$,
where $K_0=4d|t|$, this gives
\begin{equation}\label{eq:charge-operator-lower}
 H_\Lambda^{\rm Hub}(h)
 \ge (U+|h|)D_\Lambda
 -N\left(\frac{|h|}{2}+K_0\right).
\end{equation}

Let
\[
 \mathcal F:=
 -\frac1\beta
 \log\Tr_{\cH_\Lambda^{\rm hf}}
 e^{-\beta H_\Lambda^{\rm Hub}(h)}
\]
be the finite-volume free energy.
Choose a singly occupied spin product state
whose staggered magnetization is
$N\operatorname{sgn}(h)/2$ when $h\ne0$,
and any singly occupied spin product state
when $h=0$.
Its expectation of $D_\Lambda$ and $T_\Lambda$
vanishes, while its energy is $-|h|N/2$.
The Gibbs variational principle therefore gives
\[
 \mathcal F\le-\frac{|h|N}{2}.
\]

On the other hand, the Gibbs variational identity,
\eqref{eq:charge-operator-lower}, and the entropy bound
\[
 S(\rho)\le\log\dim\cH_\Lambda^{\rm hf}
 \le N\log4
\]
imply
\[
 \begin{aligned}
 \mathcal F
 &=\Tr\bigl(\rho H_\Lambda^{\rm Hub}(h)\bigr)
   -\frac1\beta S(\rho)\\
 &\ge
 (U+|h|)
 \omega_{\Lambda,\beta,U,h}^{\rm Hub}(D_\Lambda)
 -N\left(\frac{|h|}{2}+K_0
          +\frac{\log4}{\beta}\right).
 \end{aligned}
\]
Combining the upper and lower bounds yields
\begin{equation}\label{eq:charge-D-bound}
 \frac1N
 \omega_{\Lambda,\beta,U,h}^{\rm Hub}(D_\Lambda)
 \le
 \frac{K_0+\beta^{-1}\log4}{U+|h|}.
\end{equation}
Since $\sum_xq_x=2D_\Lambda$, this proves
\eqref{eq:charge-direct-bound}.

Finally, the first CDW moment vanishes by
particle--hole symmetry.
Consider the spin-exchanging particle--hole
automorphism $\Xi$ defined by
\[
 \Xi(c_{x\uparrow})=\eta_xc_{x\downarrow}^\ast,
 \qquad
 \Xi(c_{x\downarrow})=\eta_xc_{x\uparrow}^\ast.
\]
Using the bipartite property
$\eta_x\eta_y=-1$ for nearest-neighbour bonds,
we obtain
\[
 \Xi(H_\Lambda^{\rm Hub}(h))
 =H_\Lambda^{\rm Hub}(h)+U(N-N_\Lambda),
 \qquad
 \Xi(C_\Lambda^{\rm ch})=-C_\Lambda^{\rm ch}.
\]
Since $\Xi$ is unitarily implemented and preserves
the canonical half-filled space, the Hamiltonian
is invariant on this space.
Consequently, invariance of the Gibbs trace gives
\[
 \omega_{\Lambda,\beta,U,h}^{\rm Hub}
 (C_\Lambda^{\rm ch})
 =-\omega_{\Lambda,\beta,U,h}^{\rm Hub}
 (C_\Lambda^{\rm ch})=0.
\]
This completes the proof.
\end{proof}

\begin{rem}[Optimality of the charge-defect bound]
\label{rem:charge-U2-optimality}
At zero staggered field, the order $U^{-2}$ in
Theorem~\ref{thm:charge-sector-suppression-face}
is optimal for the charge-defect density, even
in the thermodynamic limit.

To see this, fix $\ell_0>0$ and consider $h=0$.
For fixed $\Lambda$ and $\beta$, set
\[
 f(V):=p_{\Lambda,\beta,V}^{\rm Hub}(0),
 \qquad
 g(J):=p_{\Lambda,\beta}^{\rm Heis}(J,0).
\]
By convexity of $f$ in $V$,
\[
 \begin{aligned}
 \frac1{|\Lambda|}
 \omega_{\Lambda,\beta,U,0}^{\rm Hub}(D_\Lambda)
 =-f'(U)
 \ge\frac{f(U)-f(2U)}{U}.
 \end{aligned}
\]

Write
\[
 B_\Lambda
 :=\sum_{\{x,y\}\in\mathscr B_\Lambda}
 \left(\frac14-\boldsymbol S_x\cdot\boldsymbol S_y\right),
\]
so that
$H_\Lambda^{\rm Heis}(J,0)=-JB_\Lambda$.
The Heisenberg pressure $g(J)$ is convex in $J$,
and the normalized trace at $J=0$ gives
$
 g'(0)=\frac{|\mathscr B_\Lambda|}{4|\Lambda|}.
$
Consequently,
\[
 \begin{aligned}
 g(J_0(U))-g(J_0(2U))
 \ge
 \frac{|\mathscr B_\Lambda|}{4|\Lambda|}
 \bigl(J_0(U)-J_0(2U)\bigr)
 =
 \frac{|\mathscr B_\Lambda|}{|\Lambda|}
 \frac{t^2}{2U}.
 \end{aligned}
\]

We now apply the third-order pressure comparison
\eqref{eq:pressure-U3-small-field} at $h=0$,
with the temperature parameter $\ell_0/2$ and a fixed
positive field cutoff, for instance $h_0=1$.
The constant $C_*$ and the lower threshold for $U$
are taken from that comparison with these fixed
parameters. Since
\[
 \beta J_0(U)\ge\ell_0
 \quad\Longrightarrow\quad
 \beta J_0(2U)\ge\ell_0/2,
\]
there is a constant $C_*<\infty$, independent
of $\Lambda$ and $\beta$, such that
\[
 \begin{aligned}
 |f(U)-g(J_0(U))|\le\frac{C_*}{U^3},\qquad
 |f(2U)-g(J_0(2U))|\le\frac{C_*}{8U^3}
 \end{aligned}
\]
for all sufficiently large $U$.
Combining these estimates yields
\begin{equation}\label{eq:charge-U2-uniform-lower}
 \frac1{|\Lambda|}
 \omega_{\Lambda,\beta,U,0}^{\rm Hub}(D_\Lambda)
 \ge
 \frac{|\mathscr B_\Lambda|}{|\Lambda|}
 \frac{t^2}{2U^2}
 -\frac{9C_*}{8U^4}.
\end{equation}

For even tori with $L\ge4$,
$|\mathscr B_\Lambda|=d|\Lambda|$.
Thus, for all sufficiently large $U$,
\[
 \inf_{\substack{L\ge4,\ L\ {\rm even}\\
                 \beta J_0(U)\ge\ell_0}}
 \frac1{|\Lambda_L|}
 \omega_{\Lambda_L,\beta,U,0}^{\rm Hub}(D_{\Lambda_L})
 \ge\frac{dt^2}{4U^2}>0.
\]
Together with Proposition~\ref{prop:charge-U2-from-pressure},
this proves that the $U^{-2}$ decay rate
is optimal for the density of doubly occupied sites
and, by $\sum_xq_x=2D_\Lambda$,
for the charge-defect density.

The lower bound is uniform in the volume and
therefore persists in the thermodynamic limit.
It does not, however, establish optimality of
the bound on the normalized CDW second moment.
\end{rem}

\section{Lee--Yang zeros and quantitative magnetization comparison}
\label{sec:future-consequences}

In this section, we derive a volume- and
temperature-uniform bound on the staggered
susceptibility of the Heisenberg model using
the Lee--Yang theorem.
Combining this bound with the pressure comparison
established above and convexity, we obtain
pointwise estimates for the difference between
the Hubbard and Heisenberg staggered magnetizations,
including in suitable vanishing magnetic fields.

\subsection{Pressures at the Heisenberg field scale}
\label{subsec:scaled-magnetization}

Let
\[
 J:=J_0(U)=\frac{4t^2}{U},\qquad b:=\beta J.
\]
For each finite even torus $\Lambda$, define
\begin{equation}\label{eq:scaled-magnetization-pressures}
 \begin{aligned}
 F_{\Lambda,b}(x)
 :=\frac{1}{J}p_{\Lambda,b/J}^{\rm Heis}(J,Jx),\qquad
 G_{\Lambda,b,U}(x)
 :=\frac{1}{J}p_{\Lambda,b/J,U}^{\rm Hub}(Jx).
 \end{aligned}
\end{equation}
Since
\[
 H_\Lambda^{\rm Heis}(J,Jx)
 =JH_\Lambda^{\rm Heis}(1,x),
\]
we have
\[
 F_{\Lambda,b}(x)
 =\frac{1}{b|\Lambda|}
 \log\Tr_{\cH_\Lambda^{\rm spin}}
 e^{-bH_\Lambda^{\rm Heis}(1,x)}.
\]
For finite $b>0$, both pressures are differentiable and
convex in $x$, with
\[
 \begin{aligned}
 F_{\Lambda,b}'(x)
 =m_{\Lambda,b/J}^{\rm Heis}(J,Jx),\qquad
 G_{\Lambda,b,U}'(x)
 =m_{\Lambda,b/J,U}^{\rm Hub}(Jx).
 \end{aligned}
\]

Fix $R>0$ and $\ell_0>0$. By the Heisenberg-scale estimate
\eqref{eq:pressure-U3-Heis-scale} of
Proposition~\ref{prop:pressure-comparison-U3}, there exist
$U_R<\infty$ and $C_R<\infty$ such that
\[
 \sup_{\substack{\Lambda\text{ even},\ b\ge\ell_0\\|x|\le R}}
 \left|p_{\Lambda,b/J,U}^{\rm Hub}(Jx)
 -p_{\Lambda,b/J}^{\rm Heis}(J,Jx)\right|
 \le\frac{C_R}{U^3},\qquad U\ge U_R.
\]
Dividing by $J=4t^2/U$, we obtain
\begin{equation}\label{eq:scaled-pressure-U3-bound}
 \sup_{\substack{\Lambda\text{ even},\ b\ge\ell_0\\|x|\le R}}
 |G_{\Lambda,b,U}(x)-F_{\Lambda,b}(x)|
 \le\varepsilon_U:=\frac{A_R}{U^2},
 \qquad A_R:=\frac{C_R}{4t^2},\quad U\ge U_R.
\end{equation}
In particular, the estimate holds for every finite even
torus and all $b\ge\ell_0$, including arbitrarily low
positive temperatures.

\subsection{Lee--Yang zeros and longitudinal susceptibility}
\label{subsec:Heisenberg-Lee-Yang}

We apply the quantum Lee--Yang theorem of Suzuki and Fisher
\cite[Theorem~2, p.~239, condition~(C)]{SuzukiFisher1971}.
The theorem concerns real anisotropic pair interactions
\[
 -\sum_{\{u,v\}}
 \left(J_{uv}^{(1)}S_u^{(1)}S_v^{(1)}
       +J_{uv}^{(2)}S_u^{(2)}S_v^{(2)}
       +J_{uv}^{(3)}S_u^{(3)}S_v^{(3)}\right)
 -2H^{(3)}\sum_u\mu_uS_u^{(3)},
\]
possibly with transverse fields. Its hypotheses include
positive temperature and
\begin{equation}\label{eq:Heisenberg-Suzuki-Fisher-condition}
 \mu_u\ge0,\qquad
 J_{uv}^{(3)}\ge|J_{uv}^{(1)}|,
 \qquad J_{uv}^{(3)}\ge|J_{uv}^{(2)}|.
\end{equation}
The inequalities in \eqref{eq:Heisenberg-Suzuki-Fisher-condition}
are non-strict. The theorem asserts that all zeros of the
partition function in the complex longitudinal field
$H^{(3)}$ lie on its imaginary axis.

\begin{lem}[Lee--Yang bound on longitudinal susceptibility]
\label{lem:Heisenberg-chi-Lee-Yang}
For every finite even torus $\Lambda$, every $b>0$,
and every $x>0$,
\begin{equation}\label{eq:Heisenberg-chi-Lee-Yang}
 0\le F_{\Lambda,b}''(x)
 \le\frac{F_{\Lambda,b}'(x)}{x}
 \le\frac{1}{2x}.
\end{equation}
In particular, $x\mapsto F_{\Lambda,b}'(x)/x$
is nonincreasing on $(0,\infty)$. All constants are
independent of the volume and the inverse temperature.
\end{lem}

\begin{proof}
Write $\Lambda=A\mathbin{\dot\cup}B$ for its bipartition
and introduce the unitary
\[
 V_B:=\prod_{v\in B}e^{\im \pi S_v^{(2)}}.
\]
For $v\in B$, conjugation by $V_B$ changes the sign of
$S_v^{(1)}$ and $S_v^{(3)}$ and leaves $S_v^{(2)}$
unchanged. Because every nearest-neighbour bond joins
$A$ to $B$, for $z\in\mathbb C$ we have
\begin{equation}\label{eq:Heisenberg-bipartite-Lee-Yang-rotation}
 \begin{aligned}
 &V_BH_\Lambda^{\rm Heis}(1,z)V_B^*\\
 &\quad=-\sum_{\{u,v\}\in\mathscr B_\Lambda}
 \bigl(S_u^{(1)}S_v^{(1)}
       -S_u^{(2)}S_v^{(2)}
       +S_u^{(3)}S_v^{(3)}\bigr)
       -z\sum_{u\in\Lambda}S_u^{(3)}
       -\frac{|\mathscr B_\Lambda|}{4}.
 \end{aligned}
\end{equation}
The exchange coefficients in the notation of
Suzuki--Fisher are, on every bond,
\[
 (J_{uv}^{(1)},J_{uv}^{(2)},J_{uv}^{(3)})
 =(1,-1,1).
\]
On absent bonds all three coefficients vanish.
Thus \eqref{eq:Heisenberg-Suzuki-Fisher-condition} holds,
including its equality case. There are no transverse
fields, and the longitudinal field in the normalization
of Suzuki--Fisher is
\[
 H^{(3)}=\frac{z}{2},\qquad \mu_u=1,
\]
because their longitudinal term is
$-2H^{(3)}\sum_u\mu_uS_u^{(3)}$.
The scalar shift in \eqref{eq:Heisenberg-bipartite-Lee-Yang-rotation}
does not affect the zeros. Therefore Theorem~2 of
\cite{SuzukiFisher1971} applies for every $b>0$ and implies
\begin{equation}\label{eq:Heisenberg-Lee-Yang-zero-free}
 Z_{\Lambda,b}(z):=
 \Tr_{\cH_\Lambda^{\rm spin}}
 e^{-bH_\Lambda^{\rm Heis}(1,z)}\ne0
 \qquad (\operatorname{Re}z\ne0).
\end{equation}

Let $N=|\Lambda|$. The finite-dimensional matrix
exponential is entire in $z$, and
\begin{equation}\label{eq:Heisenberg-Lee-Yang-exponential-type}
 |Z_{\Lambda,b}(z)|
 \le2^N\exp\left(
 b\|H_\Lambda^{\rm Heis}(1,0)\|
 +\frac{bN}{2}|z|\right).
\end{equation}
Hence $Z_{\Lambda,b}$ has order at most one.
Moreover, a global spin rotation by $\pi$ around
the first spin axis leaves the exchange Hamiltonian
invariant and sends $M_\Lambda^{\rm spin}$ to
$-M_\Lambda^{\rm spin}$. Thus
\[
 Z_{\Lambda,b}(-z)=Z_{\Lambda,b}(z),
 \qquad Z_{\Lambda,b}(0)>0.
\]
The zeros can therefore be listed, with multiplicities,
as pairs $\{\im \lambda_j,-\im \lambda_j\}$ with
$\lambda_j>0$.

By the Hadamard factorization theorem for entire
functions of order at most one, the genus-one factors
for each pair multiply to
\[
 \left(1-\frac{z}{\im \lambda_j}\right)
 e^{z/(\im \lambda_j)}
 \left(1-\frac{z}{-\im \lambda_j}\right)
 e^{z/(-\im \lambda_j)}
 =1+\frac{z^2}{\lambda_j^2}.
\]
Since $Z_{\Lambda,b}$ is an entire function of
order at most one, its zeros satisfy
$\sum_j\lambda_j^{-2}<\infty$.
Consequently, the paired product converges
locally uniformly in $z$.
The remaining exponential
factor has the form $e^{\alpha z+\gamma}$: an
$e^{cz^2}$ factor is excluded by the order bound.
Evenness forces $\alpha=0$, and normalization at zero
eliminates $e^{\gamma}$. Consequently,
\begin{equation}\label{eq:Heisenberg-Lee-Yang-product}
 \frac{Z_{\Lambda,b}(z)}{Z_{\Lambda,b}(0)}
 =\prod_j\left(1+\frac{z^2}{\lambda_j^2}\right).
\end{equation}

The locally uniform convergence permits logarithmic
differentiation for real $x>0$. Since
$F_{\Lambda,b}(x)=(bN)^{-1}\log Z_{\Lambda,b}(x)$,
we obtain
\begin{equation}\label{eq:Heisenberg-Lee-Yang-Stieltjes}
 \frac{F_{\Lambda,b}'(x)}{x}
 =\frac{2}{bN}\sum_j\frac{1}{x^2+\lambda_j^2}.
\end{equation}
Differentiating once more gives
\begin{equation}\label{eq:Heisenberg-Lee-Yang-chi-difference}
 \frac{F_{\Lambda,b}'(x)}{x}-F_{\Lambda,b}''(x)
 =\frac{4x^2}{bN}
 \sum_j\frac{1}{(x^2+\lambda_j^2)^2}\ge0.
\end{equation}
On the other hand, convexity of the finite-volume Gibbs
pressure implies $F_{\Lambda,b}''(x)\ge0$.
Finally, by the Gibbs expectation formula,
\[
 |F_{\Lambda,b}'(x)|
 \le\frac{\|M_\Lambda^{\rm spin}\|}{N}
 \le\frac12.
\]
These inequalities prove \eqref{eq:Heisenberg-chi-Lee-Yang}.
\end{proof}

\begin{rem}[Longitudinal and transverse responses]
For the zero-field exchange Hamiltonian, global
$\mathrm{SU}(2)$ invariance implies that the pressure
in a staggered vector field $\boldsymbol x$ depends only
on $|\boldsymbol x|$. Its transverse static susceptibility
at $(0,0,x)$ is therefore
$\chi_\perp(x)=F_{\Lambda,b}'(x)/x$, while the
longitudinal susceptibility is
$\chi_\parallel(x)=F_{\Lambda,b}''(x)$.
In particular, the inequality
$\chi_\parallel(x)\le\chi_\perp(x)$ is a consequence
of the Lee--Yang theorem, not of the rotation symmetry
alone. Furthermore, the representation
\eqref{eq:Heisenberg-Lee-Yang-Stieltjes} shows that
$u\mapsto F_{\Lambda,b}'(\sqrt u)/\sqrt u$ is a
Stieltjes function. Indeed, for each integer $n\ge0$,
\[
 (-1)^n\frac{d^n}{du^n}
 \frac{F_{\Lambda,b}'(\sqrt u)}{\sqrt u}
 =\frac{2n!}{bN}
 \sum_j\frac{1}{(u+\lambda_j^2)^{n+1}}\ge0.
\]
\end{rem}

\subsection{Uniform pointwise comparison of magnetization}
\label{subsec:pointwise-magnetization-Lee-Yang}

We now combine the uniform pressure comparison
with convexity and the preceding Lee--Yang
susceptibility bound to obtain a pointwise
comparison of the staggered magnetizations.

\begin{prop}[Pointwise comparison at the Heisenberg field scale]
\label{prop:pointwise-magnetization-Lee-Yang}
Fix a compact interval $I=[a,c]\Subset(0,\infty)$,
a number $R>c$, and $\ell_0>0$.
Under the hypotheses of
Proposition~\ref{prop:pressure-comparison-U3},
there exists $U_{I,R}<\infty$ such that for every
$U\ge U_{I,R}$,
\begin{equation}\label{eq:pointwise-magnetization-U1}
 \sup_{\substack{\Lambda\text{ even},\ b\ge\ell_0\\x\in I}}
 \left|m_{\Lambda,b/J,U}^{\rm Hub}(Jx)
 -m_{\Lambda,b/J}^{\rm Heis}(J,Jx)\right|
 \le\frac{2}{U}\sqrt{\frac{2A_R}{a}},
 \qquad J=\frac{4t^2}{U}.
\end{equation}
Here $A_R$ is the constant in
\eqref{eq:scaled-pressure-U3-bound}. In particular,
the $O(U^{-1})$ rate is uniform in the volume
and over all $b\ge\ell_0$.
\end{prop}

\begin{proof}
Set $\varepsilon:=A_R/U^2$ and abbreviate
$F:=F_{\Lambda,b}$ and $G:=G_{\Lambda,b,U}$.
For $x\in I$, choose
\[
 0<r<\min\{a/2,R-c\}.
\]
By convexity of $G$, its derivative lies between
the left and right secant slopes.
Using $|G-F|\le\varepsilon$ from
\eqref{eq:scaled-pressure-U3-bound}
and the convexity of $F$, we obtain
\[
 \begin{aligned}
 F'(x-r)-\frac{2\varepsilon}{r}
 &\le\frac{G(x)-G(x-r)}{r}\\
 &\le G'(x)\\
 &\le\frac{G(x+r)-G(x)}{r}\\
 &\le F'(x+r)+\frac{2\varepsilon}{r}.
 \end{aligned}
\]
Since $y\ge a/2>0$ for
$y\in[x-r,x+r]$,
Lemma~\ref{lem:Heisenberg-chi-Lee-Yang} gives
\[
 0\le F''(y)\le\frac1{2y}\le\frac1a.
\]
Applying the mean-value theorem to $F'$ on
$[x-r,x]$ and $[x,x+r]$, respectively, gives
\[
 F'(x)-\frac{r}{a}
 \le F'(x-r)
 \le F'(x+r)
 \le F'(x)+\frac{r}{a}.
\]
Together with the preceding secant bounds, this yields
\[
 |G'(x)-F'(x)|
 \le\frac{r}{a}+\frac{2\varepsilon}{r}.
\]

For sufficiently large $U$, the choice
$
 r=\sqrt{2a\varepsilon}
   =\frac{\sqrt{2aA_R}}{U}
$
is admissible and yields
\[
 |G'(x)-F'(x)|
 \le2\sqrt{\frac{2A_R}{a}}\,U^{-1}.
\]
The identities between pressure derivatives and
magnetizations prove
\eqref{eq:pointwise-magnetization-U1}.
\end{proof}

\begin{rem}[Order of limits]
The preceding estimates hold for all finite inverse
temperatures $b\ge\ell_0$, with constants independent
of $b$ and $\Lambda$. They do not establish spontaneous
staggered magnetization at any fixed finite $U$ after
taking the zero-field limit.
\end{rem}

\section{Infinite-volume magnetization comparison via the fermionic Lee--Yang theorem}
\label{sec:fermionic-Lee-Yang}

We sharpen the finite-volume magnetization comparison of
Section~\ref{sec:future-consequences} after taking the thermodynamic
limit.  The additional ingredients are the fermionic Lee--Yang
theorem for the Hubbard grand-canonical ensemble, an $\eta$-spin
comparison with canonical half filling, and analytic interpolation.
The finite-volume canonical partition function is not assumed to
have the Lee--Yang zero property.

Throughout this section, $N=|\Lambda|$, and the notation for the
Hubbard and Heisenberg Hamiltonians, pressures, and staggered fields
is as in Subsections~\ref{subsec:model} and
\ref{subsec:ref-heis-pressure}.  We retain the scaled pressures
$F_{\Lambda,b}$ and $G_{\Lambda,b,U}$ introduced in
\eqref{eq:scaled-magnetization-pressures}.

\subsection{Canonical--grand-canonical comparison}\label{subsec:FLY-ensemble}
For $h\in\mathbb R$ define, on the \emph{full} Fock space $\cH_{\Lambda}$,
\begin{equation}\label{eq:K}
 K_\Lambda(h)=H_\Lambda^{\mathrm{Hub}}(h)-\frac U2N_\Lambda,
 \qquad Z_k(h)=\Tr_{\{N_\Lambda=k\}}\e^{-\beta K_\Lambda(h)},
 \qquad Z^{\mathrm{gc}}(h)=\sum_{k=0}^{2N} Z_k(h).
\end{equation}
Then
\begin{equation}\label{eq:grand-canonical-shift}
 Z_N(h)=\e^{\beta U N/2}Z_{\Lambda,\beta,U}^{\mathrm{Hub}}(h),\qquad
 \frac1{\beta N}\log Z_N(h)=p_{\Lambda,\beta,U}^{\mathrm{Hub}}(h)+\frac U2.
\end{equation}

\begin{lem}[Finite-volume ensemble comparison]\label{lem:ensemble}
For every finite $\beta,U>0$ and real $h$, and every even torus,
\begin{equation}\label{eq:gc-can}
 1\le\frac{Z^{\mathrm{gc}}(h)}{Z_N(h)}
 \le (N+1)\left[1+2\e^{\beta(U/2+8d|t|)}
                                 \cosh\left(\frac{\beta h}{2}\right)\right].
\end{equation}
In particular,
\begin{equation}\label{eq:pressure-gap-explicit}
\begin{aligned}
0\le\frac1{\beta N}\log Z^{\mathrm{gc}}(h)-\frac U2
                         -p_{\Lambda,\beta,U}^{\mathrm{Hub}}(h)
\le\frac{U/2+8d|t|+|h|/2}{N}
       +\frac{\log(3(N+1))}{\beta N}.
\end{aligned}
\end{equation}
Thus the pressure difference tends to zero as $L\to\infty$ at fixed
finite $U,\beta,h$.
\end{lem}

\begin{proof}
\emph{Step 1: $\eta$-pairing symmetry.}
For each site define
$Q_x^+=c_{x\uparrow}^* c_{x\downarrow}^*$ and
$Q_x^-=(Q_x^+)^*$.
The canonical anticommutation relations give
\begin{equation}\label{eq:eta-CAR}
 [Q_x^+,Q_x^-]=n_x-1,\qquad
 [n_x,Q_x^\pm]=\pm2Q_x^\pm,\qquad
 [n_{x\uparrow}n_{x\downarrow},Q_x^+]=Q_x^+.
\end{equation}
Set
\begin{equation}\label{eq:eta-def}
 \eta^+=\sum_x\eta_xQ_x^+,\qquad
 \eta^-=(\eta^+)^*,\qquad
 \eta^3=\frac12(N_\Lambda-N).
\end{equation}
Here recall $\eta_x=(-1)^{x_1+\cdots+x_d}$.
Since the $Q_x^\pm$ are even operators, pairs from distinct sites
commute.  Therefore
$[\eta^+,\eta^-]=2\eta^3$ and
$[\eta^3,\eta^\pm]=\pm\eta^\pm$.
By \eqref{eq:eta-CAR},
\begin{equation}\label{eq:eta-on-site}
 [UD_x-\tfrac U2n_x,Q_x^+]
     =(U-U)Q_x^+=0.
\end{equation}
Moreover $Q_x^+$ creates a local spin singlet, so
$[n_{x\uparrow}-n_{x\downarrow},Q_x^+]=0$.
Finally, for a hopping bond $\{x,y\}$, the terms in
$[T_{\{x,y\}},\eta^+]$ are proportional to
$\eta_x+\eta_y=0$; the same holds for $\eta^-$ by adjunction.
Consequently
\begin{equation}\label{eq:eta-commutation}
 [K_\Lambda(h),\eta^+]
 =[K_\Lambda(h),\eta^-]
 =[K_\Lambda(h),\eta^3]=0.
\end{equation}
These relations hold for nonzero staggered $h$ as well as $h=0$.

\smallskip\noindent\emph{Step 2: comparison of particle-number weights.}
Since $K_\Lambda(h)$ commutes with the generators of the
$\eta$-spin $\mathrm{SU}(2)$ symmetry, each of its eigenspaces
decomposes into irreducible $\eta$-spin representations.

An irreducible $\eta$-spin representation with quantum number
$j$ has dimension $2j+1\le N+1$ and contains the weights
$m=-j,-j+1,\ldots,j$. By the definition
$
  \eta^3=\frac12(N_\Lambda-N),
$
a state of weight $m$ belongs to the particle-number sector
$N_\Lambda=N+2m$.

Since $N$ is even, a multiplet contributing to an even
particle-number sector has integer $j$ and therefore contains
exactly one state of weight $m=0$, corresponding to the
half-filled sector $N_\Lambda=N$.
Similarly, a multiplet contributing to an odd particle-number
sector has half-integer $j$ and contains exactly one state of
weight $m=1/2$, corresponding to $N_\Lambda=N+1$.

All states within a given $\eta$-spin multiplet
have the same energy, since $K_\Lambda(h)$ commutes
with the $\eta$-spin generators.
A spin-$j$ multiplet contains $2j+1$ states,
with particle numbers $k=N+2m$,
where $m=-j,-j+1,\ldots,j$.

Let $a_j$ denote the sum of the Boltzmann weights
of all spin-$j$ multiplets, counting multiplicities.
Since $N$ is even, integer $j$ corresponds to even
particle numbers, while half-integer $j$ corresponds
to odd particle numbers.
Consequently,
\[
\begin{aligned}
\sum_{k\ \mathrm{even}}Z_k
 &=\sum_{j\in\mathbb Z_{\ge0}}(2j+1)a_j,
& Z_N
 &=\sum_{j\in\mathbb Z_{\ge0}}a_j,\\
\sum_{k\ \mathrm{odd}}Z_k
 &=\sum_{j\in\mathbb Z_{\ge0}+1/2}(2j+1)a_j,
& Z_{N+1}
 &=\sum_{j\in\mathbb Z_{\ge0}+1/2}a_j.
\end{aligned}
\]
Here the sums over $j$ include only multiplets
that occur in the Fock space.
Since $a_j\ge0$ and $2j+1\le N+1$, we obtain
\begin{equation}\label{eq:multiplets}
  \sum_{k\ \mathrm{even}} Z_k
  \le (N+1)Z_N,
  \qquad
  \sum_{k\ \mathrm{odd}} Z_k
  \le (N+1)Z_{N+1}.
\end{equation}

\smallskip\noindent\emph{Step 3: detaching one vertex.}
Choose $x_0\in\Lambda$ and write
$\Lambda'=\Lambda\setminus\{x_0\}$.
Let $K_{\Lambda'}(h)$ denote the restriction of
$K_\Lambda(h)$ to $\Lambda'$, retaining the hopping
bonds internal to $\Lambda'$ and the original
on-site terms and staggered field.

On the one-site Fock space at $x_0$, define
\begin{equation}\label{eq:isolated-site-hamiltonian}
 K_{x_0}(h)
 :=U n_{x_0\uparrow}n_{x_0\downarrow}
   -\frac U2(n_{x_0\uparrow}+n_{x_0\downarrow})
   -\frac{h\eta_{x_0}}2
    (n_{x_0\uparrow}-n_{x_0\downarrow}).
\end{equation}
Set
\[
 K^0:=K_{\Lambda'}(h)+K_{x_0}(h),
 \qquad
 V:=K_\Lambda(h)-K^0.
\]
Thus $V$ consists precisely of the hopping terms
on the bonds incident to $x_0$.

With our unordered-bond convention there are at
most $2d$ such bonds, and each spinful hopping
bond has norm at most $2|t|$. Consequently,
\begin{equation}\label{eq:v}
 \lVert V\rVert\le v:=4d|t|.
\end{equation}
Both $K^0$ and $V$ conserve the total particle
number. Define
\[
 Z_k^0
 :=\Tr_{\{N_\Lambda=k\}}\e^{-\beta K^0}.
\]
The min--max principle in each finite-dimensional
particle-number sector gives
\begin{equation}\label{eq:minmax}
 \e^{-\beta v}Z_k^0
 \le Z_k
 \le \e^{\beta v}Z_k^0.
\end{equation}

We next express $Z_k^0$ through the two
decoupled subsystems.
The one-site Hamiltonian
\eqref{eq:isolated-site-hamiltonian} has energies
\[
 E_0=E_2=0,
 \qquad
 E_\uparrow=-\frac{U+h\eta_{x_0}}2,
 \qquad
 E_\downarrow=-\frac{U-h\eta_{x_0}}2
\]
for the empty, doubly occupied, spin-up and
spin-down states, respectively.
Hence its Boltzmann weights, grouped by local
particle number $r\in\{0,1,2\}$, are
$w_0=w_2=1$ and
\begin{equation}\label{eq:w1}
 \begin{aligned}
 w_1
 =\e^{\beta(U+h\eta_{x_0})/2}
   +\e^{\beta(U-h\eta_{x_0})/2}
 =2\e^{\beta U/2}\cosh(\beta h/2)\ge2.
 \end{aligned}
\end{equation}

For the remaining $N-1$ sites, put
\[
 N_{\Lambda'}:=\sum_{x\in\Lambda'}n_x,
 \qquad
 W_k:=
 \Tr_{\{N_{\Lambda'}=k\}}
 \e^{-\beta K_{\Lambda'}(h)}.
\]
We set $W_k=0$ when $k$ is outside the
admissible particle-number range of $\Lambda'$.

Since $K^0$ is the sum of two decoupled
Hamiltonians, the decomposition of the total
particle number gives the exact convolution identity
\[
 Z_k^0
 =\sum_{r=0}^2w_rW_{k-r}
 =W_k+w_1W_{k-1}+W_{k-2}.
\]

We next compare $W_{N+1}$ and $W_{N-1}$.
The graph induced by $\Lambda'$ is still bipartite,
and $K_{\Lambda'}(h)$ commutes with the
$\eta$-spin generators on $\Lambda'$.
In particular,
\[
 [K_{\Lambda'}(h),\eta_{\Lambda'}^-]=0,
 \qquad
 \eta_{\Lambda'}^3
 =\frac12\bigl(N_{\Lambda'}-(N-1)\bigr).
\]

Let $\psi$ be an eigenvector of $K_{\Lambda'}(h)$
with energy $E$ in the particle-number sector
$N_{\Lambda'}=N+1$.
Then $\eta_{\Lambda'}^3\psi=\psi$.
Using the $\eta$-spin commutation relation,
we obtain
\[
\begin{aligned}
 \|\eta_{\Lambda'}^-\psi\|^2
 &=\langle\psi,
    \eta_{\Lambda'}^+\eta_{\Lambda'}^-\psi\rangle\\
 &=\|\eta_{\Lambda'}^+\psi\|^2
   +2\langle\psi,\eta_{\Lambda'}^3\psi\rangle\\
 &=\|\eta_{\Lambda'}^+\psi\|^2
   +2\|\psi\|^2\\
 &\ge 2\|\psi\|^2.
\end{aligned}
\]
Thus $\eta_{\Lambda'}^-\psi\ne0$ whenever
$\psi\ne0$.

Moreover, $\eta_{\Lambda'}^-$ lowers the
particle number by two and commutes with
$K_{\Lambda'}(h)$.
Consequently, it defines an injective linear map
from the energy-$E$ eigenspace in the
$N_{\Lambda'}=N+1$ sector into the corresponding
energy-$E$ eigenspace in the $N_{\Lambda'}=N-1$
sector.

Hence the multiplicity of every energy $E$
in the former sector is no greater than its
multiplicity in the latter.
Summing the Boltzmann weights over all
eigenvalues gives
\begin{equation}\label{eq:W-weight}
 W_{N+1}\le W_{N-1}.
\end{equation}

Applying the convolution identity at
$k=N+1$ and $k=N$, we obtain
\begin{equation}\label{eq:one-site-convolution}
\begin{aligned}
 Z_{N+1}^0
 &=W_{N+1}+w_1W_N+W_{N-1}
 \le w_1W_N+2W_{N-1},\\
 Z_N^0
 &=W_N+w_1W_{N-1}+W_{N-2}
 \ge W_N+w_1W_{N-1}.
\end{aligned}
\end{equation}
Since $w_1\ge2$, we have $2\le w_1^2$;
hence
\[
 Z_{N+1}^0
 \le w_1W_N+w_1^2W_{N-1}
 \le w_1 Z_N^0.
\]

Together with \eqref{eq:minmax}, this gives
\begin{equation}\label{eq:neighbor-comparison}
 \begin{aligned}
 \frac{Z_{N+1}}{Z_N}
 \le
 \e^{2\beta v}\frac{Z_{N+1}^0}{Z_N^0}
 \le\e^{2\beta v}w_1
 =2\e^{\beta(U/2+8d|t|)}
   \cosh(\beta h/2).
 \end{aligned}
\end{equation}

Finally, \eqref{eq:multiplets} implies
\[
 \begin{aligned}
 Z^{\mathrm{gc}}
 =\sum_{k\ \mathrm{even}}Z_k
   +\sum_{k\ \mathrm{odd}}Z_k
 \le(N+1)(Z_N+Z_{N+1}).
 \end{aligned}
\]
Combining this with
\eqref{eq:neighbor-comparison} and
$Z^{\mathrm{gc}}\ge Z_N$ proves \eqref{eq:gc-can}.

\smallskip\noindent\emph{Step 4: pressure error.}
Put $\kappa:=U/2+2v=U/2+8d|t|$.
For real $h$,
\[
 2\e^{\beta\kappa}\cosh(\beta h/2)
 \le2\e^{\beta(\kappa+|h|/2)}.
\]
Since $\kappa+|h|/2\ge0$,
\[
 1+2\e^{\beta(\kappa+|h|/2)}
 \le3\e^{\beta(\kappa+|h|/2)}.
\]
Taking logarithms in \eqref{eq:gc-can}
therefore gives
\[
 \begin{aligned}
 0
 \le\frac1{\beta N}
       \log\frac{Z^{\mathrm{gc}}(h)}{Z_N(h)}
 \le\frac{\kappa+|h|/2}{N}
       +\frac{\log(3(N+1))}{\beta N}.
 \end{aligned}
\]
By \eqref{eq:grand-canonical-shift},
the expression on the left is exactly
\[
 \frac1{\beta N}\log Z^{\mathrm{gc}}(h)
 -\frac U2-p_{\Lambda,\beta,U}^{\mathrm{Hub}}(h).
\]
This proves \eqref{eq:pressure-gap-explicit}.

\end{proof}

The grand-canonical pressure for the finite-range Hubbard Hamiltonian
at each fixed $\beta,U,h$ has a thermodynamic limit along even tori,
by the standard boundary-surface estimates for finite-range lattice
fermion systems.  The analogous statement holds for the Heisenberg
pressure.  Consequently, Lemma~\ref{lem:ensemble} implies the
existence of the canonical half-filled thermodynamic pressure and
identifies it as
\begin{equation}\label{eq:FLY-ensemble-limit}
 p_{\beta,U}^{\mathrm{Hub}}(h)
 =p_{\beta,U}^{\mathrm{gc}}(U/2,h)-\frac U2,
 \qquad h\in\mathbb R,
\end{equation}
where $p_{\beta,U}^{\mathrm{gc}}(U/2,h)$ is the limiting full-Fock-space
pressure associated with $K_\Lambda(h)$.

\begin{rem}[Order of limits]\label{rem:ensemble-limits}
The right side of \eqref{eq:pressure-gap-explicit} is $o(1)$ in the
volume for \emph{fixed} $\beta,U,h$, but it does not supply a
$U$-uniform convergence rate under $\beta=bU/(4t^2)$.
The proof below first identifies the infinite-volume pressures at each
fixed $U,b$ and only then uses a bound uniform in $b\ge\ell_0$.
\end{rem}

\subsection{Fermionic Lee--Yang theorem and analyticity}
\label{subsec:FLY-analytic}
We use two external theorems: the fermionic Lee--Yang theorem of
Rayudu--Misawa--Zhao--Takahashi~\cite{RMZT26}, and the
Suzuki--Fisher spin Lee--Yang theorem~\cite{SuzukiFisher1971}.

\begin{lem}[Matching the fermionic Lee--Yang theorem]\label{lem:PH}
For every finite bipartite torus, $U>0$ and $\beta>0$, the
\emph{grand-canonical} partition function $Z^{\mathrm{gc}}(z)$ of
\eqref{eq:K}, analytically continued to $z\in\mathbb C$, has all its
zeros on the imaginary $z$-axis.  This statement is not asserted for
$Z_N(z)$.
\end{lem}

\begin{proof}
The centered interaction satisfies
\begin{equation}\label{eq:interaction-centered}
 UD_\Lambda-\frac U2N_\Lambda
 =
 U\sum_x
 (n_{x\uparrow}-\tfrac12)
 (n_{x\downarrow}-\tfrac12)
 -\frac{UN}{4}.
\end{equation}

We perform a particle--hole transformation on the
full fermionic Fock space $\cH_\Lambda$.
Let
\[
 F:=(A\times\{\downarrow\})
   \cup(B\times\{\uparrow\}).
\]
Since $|F|=N$ is even, the operator
\[
 \mathcal U_{\rm ph}
 :=
 \prod_{j\in F}^{\longrightarrow}
 (c_j^*-c_j),
\]
with the product taken in a fixed ordering,
is unitary and satisfies
\[
 \mathcal U_{\rm ph}^*c_{x\sigma}
 \mathcal U_{\rm ph}
 =
 \begin{cases}
 c_{x\sigma}^*,
   &(x,\sigma)\in F,\\
 c_{x\sigma},
   &(x,\sigma)\notin F.
 \end{cases}
\]
Indeed, each factor $c_j^*-c_j$ is unitary,
and the stated transformation follows from
the CAR and the even cardinality of $F$.

The transformation exchanges particles and holes
in the modes belonging to $F$.
In particular, $\mathcal U_{\rm ph}^*|0\rangle_\Lambda$
is, up to a phase, the state in which precisely
these modes are occupied.
However, we express the conjugated Hamiltonian
in the original CAR generators, so the original
Fock representation can be used throughout.

At each site, exactly one spin mode is
particle--hole transformed. Hence
\[
 \begin{aligned}
 \mathcal U_{\rm ph}^*
 (n_{x\uparrow}-\tfrac12)
 (n_{x\downarrow}-\tfrac12)
 \mathcal U_{\rm ph}
 =
 -(n_{x\uparrow}-\tfrac12)
  (n_{x\downarrow}-\tfrac12).
 \end{aligned}
\]
Moreover, since $\eta_x=1$ on $A$
and $\eta_x=-1$ on $B$,
\[
 \mathcal U_{\rm ph}^*
 \left(
 \sum_x\eta_x(n_{x\uparrow}-n_{x\downarrow})
 \right)
 \mathcal U_{\rm ph}
 =N_\Lambda-N.
\]

The nearest-neighbour hopping is transformed
into a purely pairing quadratic operator:
\[
\begin{aligned}
 T_{\rm pair}
 :=-t\sum_{\substack{
       \{x,y\}\in\mathscr B_\Lambda\\
       x\in A,\ y\in B}}
 \Bigl[
 &c_{x\uparrow}^*c_{y\uparrow}^*
 -c_{x\downarrow}^*c_{y\downarrow}^*
 +\mathrm{h.c.}
 \Bigr].
\end{aligned}
\]
Here the relative minus sign in the
spin-down pairing term follows from the CAR.

Combining these identities with
\eqref{eq:interaction-centered}, we obtain
the exact operator identity
\begin{equation}\label{eq:Hubbard-PH-transformed}
\begin{aligned}
 \mathcal U_{\rm ph}^*
 K_\Lambda(z)\mathcal U_{\rm ph}
 =
 T_{\rm pair}
 -U\sum_x
 (n_{x\uparrow}-\tfrac12)
 (n_{x\downarrow}-\tfrac12)
 -\frac z2N_\Lambda
 +\frac{zN}{2}-\frac{UN}{4}.
\end{aligned}
\end{equation}

Thus the conjugated Hamiltonian has an attractive
centered density interaction and a purely pairing
quadratic part.
In the notation of \cite[Definition~1]{RMZT26},
the number-conserving quadratic matrix is $A=0$,
the pairing matrix satisfies $B=-B^T$,
and the density interaction is attractive.

Set
\[
 H_{\rm LY}
 :=
 T_{\rm pair}
 -U\sum_x
 (n_{x\uparrow}-\tfrac12)
 (n_{x\downarrow}-\tfrac12).
\]
Since $\mathcal U_{\rm ph}$ is unitary and
independent of $z$, trace invariance gives
\[
 \begin{aligned}
 Z_{\Lambda,\beta}^{\rm gc}(z)
 =
 \Tr\e^{-\beta\mathcal U_{\rm ph}^*
                  K_\Lambda(z)\mathcal U_{\rm ph}}
 =
 \e^{\beta UN/4-\beta zN/2}
 \Tr\e^{-\beta(H_{\rm LY}-(z/2)N_\Lambda)}.
 \end{aligned}
\]

The magnetic-field variable in
\cite[Theorem~2]{RMZT26} is
$\zeta=-z/2$.
Since the number-conserving quadratic matrix
vanishes, that theorem implies that
the partition function on the right
has no zeros whenever $\Re\zeta\ne0$,
equivalently, whenever $\Re z\ne0$.

The exponential prefactor is nowhere zero.
Consequently,
\begin{equation}\label{eq:Hubbard-gc-Lee-Yang}
 Z_{\Lambda,\beta}^{\rm gc}(z)\ne0,
 \qquad \Re z\ne0.
\end{equation}
\end{proof}

\begin{lem}[Hadamard product and uniform complex response]
\label{lem:analytic}
Let $Z_\Lambda(z)$ denote either the grand-canonical Hubbard
partition function $Z^{\mathrm{gc}}(z)$ at chemical potential $U/2$, or the
Heisenberg partition function defined in Subsection~\ref{subsec:ref-heis-pressure}.
For every fixed finite volume and positive finite inverse temperature:
\begin{enumerate}[label=\textup{(\alph*)},leftmargin=*]
\item $Z_\Lambda$ is even, entire of order at most one, satisfies
 $Z_\Lambda(0)>0$, and has all its zeros on the imaginary axis.
\item Its logarithmic derivative
 $m_\Lambda(z)=(\beta N)^{-1} Z_\Lambda'(z)/Z_\Lambda(z)$ is
 holomorphic on $\{\Re z>0\}$.
\item On the open cone
 \begin{equation}\label{eq:cone}
  \Omega:=\{z\in\mathbb C:\Re z>0,\ |\Im z|<(\Re z)/2\},
 \end{equation}
 one has the volume-, temperature-, and coupling-uniform bound
 \begin{equation}\label{eq:complex-bound}
  |m_\Lambda(s+\im y)|\le\frac{2\sqrt5}{3}m_\Lambda(s)
           \le\frac{\sqrt5}{3},\qquad s>0,\quad |y|\le s/2.
 \end{equation}
 The same estimate holds for derivatives in scaled field of the
 scaled grand-canonical and Heisenberg pressures
 $z\mapsto J^{-1}p_\Lambda(Jz)$, $J>0$.
\end{enumerate}
\end{lem}
\begin{proof}
Finite-dimensional matrix exponential estimates imply
\begin{equation}\label{eq:exponential-type}
 |Z_\Lambda(z)|\le\dim\mathcal H_\Lambda\,
 \exp\{\beta\lVert H_\Lambda(0)\rVert+\beta N|z|/2\}.
\end{equation}
Global spin reversal sends $M_\Lambda\mapsto-M_\Lambda$ and fixes
both zero-field Hamiltonians, so $Z_\Lambda(-z)=Z_\Lambda(z)$.
At the origin the trace of the exponential of a self-adjoint matrix is
strictly positive.  The imaginary-axis zero assertion is
Lemma~\ref{lem:PH} for the Hubbard grand-canonical case and
\cite[Theorem~2]{SuzukiFisher1971} for the bipartitely rotated Heisenberg model:
its transformed pair couplings are $(1,-1,1)$, which satisfy the
non-strict condition $J^{(3)}\ge|J^{(1)}|,|J^{(2)}|$.

The Hadamard product for the Heisenberg partition function was
obtained in \eqref{eq:Heisenberg-Lee-Yang-product}.  For the Hubbard
grand-canonical partition function, list the zeros as
$\{\pm \im \lambda_j\}$, $\lambda_j>0$, with multiplicity.
The order bound \eqref{eq:exponential-type} implies
$\sum_j\lambda_j^{-2}<\infty$ and permits pairing of the
genus-one factors.  Evenness and order at most one exclude the
remaining nonconstant exponential factor.  In either case,
normalization at $z=0$ gives
\begin{equation}\label{eq:product}
 \frac{Z_\Lambda(z)}{Z_\Lambda(0)}
 =\prod_j\left(1+\frac{z^2}{\lambda_j^2}\right),\qquad
 m_\Lambda(z)=\frac{2z}{\beta N}
      \sum_j\frac1{z^2+\lambda_j^2}.
\end{equation}
The paired products and differentiated sums converge locally
uniformly away from their zeros.

For $z=s+\im y$ with $|y|\le s/2$,
\begin{equation}\label{eq:kernel}
 |z^2+\lambda^2|\ge s^2-y^2+\lambda^2
 \ge\tfrac34(s^2+\lambda^2),\qquad
 |z|\le\tfrac{\sqrt5}{2}s.
\end{equation}
The sum in \eqref{eq:product} is positive for real $s>0$,
so \eqref{eq:kernel} gives
$|m_\Lambda(z)|\le(2\sqrt5/3)m_\Lambda(s)$.
The real-field derivative is the Gibbs expectation of the conjugate
magnetization density, hence
$0\le m_\Lambda(s)\le\lVert M_\Lambda\rVert/N\le1/2$.
For $z\mapsto J^{-1}p_\Lambda(Jz)$, differentiation removes the $J$
by the chain rule and $J\Omega\subset\Omega$ for every $J>0$.
\end{proof}

\begin{lem}[Analyticity after taking the thermodynamic limit]
\label{lem:limit}
For every fixed finite $\beta,U>0$, the canonical half-filled
Hubbard pressure has a thermodynamic limit $p_{\beta,U}^{\mathrm{Hub}}(h)$.
It is real analytic at each $h\ne0$.  For $h>0$ it agrees with a
holomorphic limit on $\Omega$ of the logarithmic grand-canonical pressures
at chemical potential $\mu=U/2$, minus $U/2$.
At every nonzero real $h$ the finite-volume canonical magnetizations
converge to $\partial_h p_{\beta,U}^{\mathrm{Hub}}(h)$.
The same existence, analyticity, and derivative convergence hold for
the Heisenberg reference at fixed finite $\beta,J>0$.
\end{lem}
\begin{proof}
Fix $U,\beta\in(0,\infty)$ and write
\begin{equation}\label{eq:log-branch}
 q_\Lambda(z)=\frac1{\beta N}\log Z^{\mathrm{gc}}_\Lambda(z)-\frac U2,
 \qquad z\in\Omega,
\end{equation}
choosing the holomorphic logarithm to be real at $z=1$.
This branch exists because $\Omega$ is simply connected and
$Z^{\mathrm{gc}}_\Lambda$ has no zeros there.
By Lemma~\ref{lem:analytic},
$|q_\Lambda'(z)|\le\sqrt5/3$ on $\Omega$, uniformly in $\Lambda$.
At the real basepoint $z=1$, both $\lVert K_\Lambda(1)\rVert/N$ and
$\log(\dim\cH_\Lambda)/(\beta N)=\log4/\beta$ are bounded
independently of the volume for fixed $U,\beta$.
Consequently $\{q_\Lambda(1)\}$ is bounded.  Since $\Omega$ is convex,
integration of $q_\Lambda'$ from $1$ to any $z\in\Omega$ shows that
$\{q_\Lambda\}$ is locally uniformly bounded on $\Omega$.

By the pressure limit identified in \eqref{eq:FLY-ensemble-limit}, for every $h>0$,
\begin{equation}\label{eq:real-identity}
 \lim_{L\to\infty}q_{\Lambda_L}(h)
 =\lim_{L\to\infty}p_{\Lambda_L,\beta,U}^{\mathrm{Hub}}(h)
 =p_{\beta,U}^{\mathrm{Hub}}(h).
\end{equation}
Since $(0,\infty)$ has accumulation points in $\Omega$,
Vitali's theorem, together with the local boundedness
established above, implies that $q_{\Lambda_L}$
converges locally uniformly on $\Omega$ to a holomorphic
function $q$ satisfying
\[
 q(h)=p_{\beta,U}^{\mathrm{Hub}}(h),\qquad h>0.
\]
Cauchy's integral formula further yields
$q_{\Lambda_L}'\to q'$ locally uniformly on $\Omega$.
In particular, $p_{\beta,U}^{\mathrm{Hub}}$ is real analytic
on $(0,\infty)$.

The same normal-family argument applies to the
Heisenberg partition functions, using the standard
thermodynamic limits of their real pressures and
the bound from Lemma~\ref{lem:analytic}.
Spin-reversal symmetry extends the analyticity
of both limiting pressures to negative fields.

It remains to establish the convergence of the
finite-volume magnetizations.
For either model, write $p_\Lambda(h)$ for the
finite-volume pressure and $p(h)$ for its
thermodynamic limit.
Each $p_\Lambda$ is differentiable and convex
in $h$, and its derivative equals the normalized
magnetization.
Thus, for $r>0$,
\begin{equation}\label{eq:convex-squeeze}
 \frac{p_\Lambda(h)-p_\Lambda(h-r)}r
 \le p_\Lambda'(h)\le
 \frac{p_\Lambda(h+r)-p_\Lambda(h)}r.
\end{equation}
Letting $L\to\infty$ at fixed $r$ and then
$r\downarrow0$, the differentiability of $p$
at every $h\ne0$ yields
\[
 \begin{aligned}
 \lim_{L\to\infty}
 m_{\Lambda_L,\beta,U}^{\mathrm{Hub}}(h)
 &=\partial_h p_{\beta,U}^{\mathrm{Hub}}(h),\\
 \lim_{L\to\infty}
 m_{\Lambda_L,\beta}^{\mathrm{Heis}}(J,h)
 &=\partial_h p_\beta^{\mathrm{Heis}}(J,h),
 \qquad h\ne0.
 \end{aligned}
\]
This completes the proof of the magnetization
limits asserted in the main theorem.

\end{proof}

\begin{rem}[No finite-volume canonical Lee--Yang claim]
\label{rem:canonical-zeros}
The canonical fixed-particle-number partition functions do not
inherit the grand-canonical Lee--Yang property.  For example, in the
two-site atomic limit $t=0$ at half filling,
$Z_{\Lambda,\beta,U}^{\mathrm{Hub}}(h)=2\cosh(\beta h)+2+2\e^{-\beta U}$ has zeros
with nonzero real part when $U,\beta\in(0,\infty)$.
This observation is consistent with the analytic-limit argument:
it is the grand-canonical finite-volume pressure, not the canonical
finite-volume pressure, that admits a uniform Lee--Yang zero-free
holomorphic extension.
\end{rem}

\subsection{Analytic interpolation}
\label{subsec:FLY-interpolation}
Let $J_1=[a_1,c_1]$, with midpoint
$m=(a_1+c_1)/2$ and half-length
$\ell=(c_1-a_1)/2>0$.
For $\rho>1$, let $E_\rho(J_1)$ denote the
closed Bernstein ellipse whose boundary is
parametrized by
\begin{equation}\label{eq:ellipse-def}
 \partial E_\rho(J_1)
 =
 \left\{
 m+\frac{\ell}{2}(w+w^{-1}):
 |w|=\rho
 \right\}.
\end{equation}
Its semiaxes have lengths
$\ell(\rho+\rho^{-1})/2$ and
$\ell(\rho-\rho^{-1})/2$.

\begin{lem}[Interior Bernstein--Chebyshev estimate]
\label{lem:cheb}
Let $I=[a,c]$ be contained in the interior of $J_1$.
Suppose $f$ is holomorphic on an open neighbourhood of
$E_\rho(J_1)$ for a fixed $\rho>1$ and obeys
\begin{equation}\label{eq:interpolation-hypotheses}
 \sup_{s\in J_1}|f(s)|\le\varepsilon,
 \qquad\sup_{z\in E_\rho(J_1)}|f(z)|\le M,
 \qquad0<\varepsilon\le M.
\end{equation}
There is a constant $C_{I,J_1,\rho}$ independent of
$f,\varepsilon,M$ such that
\begin{equation}\label{eq:cheb}
 \sup_{x\in I}|f'(x)|
 \le C_{I,J_1,\rho}\,\varepsilon
       \left(1+\log\frac M\varepsilon\right).
\end{equation}
More explicitly, writing $\ell=|J_1|/2$ and assuming
$\operatorname{dist}(I,\partial J_1)/\ell\ge\delta>0$,
one may take
$C_{I,J_1,\rho}\le C_{\delta,\rho}/\ell$.
\end{lem}
\begin{proof}
Use the affine change of variables $s=m+\ell u$ to map $J_1$
to $[-1,1]$, and write $\widehat f(u)=f(m+\ell u)$.
Its Chebyshev expansion $\widehat f=\sum_{k=0}^\infty c_kT_k$
is locally uniformly convergent in the interior of the ellipse.

By the classical estimates for Chebyshev expansions
of analytic functions
\cite[Theorems~8.1 and~8.2]{Trefethen2019},
the Chebyshev coefficients satisfy
\[
 |c_k|\le 2M\rho^{-k},\qquad k\ge1.
\]
Set $P_q=\sum_{k=0}^q c_kT_k$.
The corresponding truncation estimate gives
\begin{equation}\label{eq:cheb-tail}
 \begin{aligned}
 \lVert \widehat f-P_q\rVert_{L^\infty([-1,1])}
 \le C_\rho M\rho^{-q},\qquad
 \lVert P_q\rVert_{L^\infty([-1,1])}
 \le\varepsilon+C_\rho M\rho^{-q},
 \end{aligned}
\end{equation}
where $C_\rho=2/(\rho-1)$.
The image $\widehat I$ of $I$ lies in $[-1+\delta,1-\delta]$.
The classical Bernstein inequality for algebraic
polynomials (see, e.g., \cite[Eq.~(3)]{Sroka2022})
gives, for $-1<u<1$,
\begin{equation}\label{eq:bernstein}
 |P_q'(u)|
 \le \frac{q}{\sqrt{1-u^2}}
       \lVert P_q\rVert_{L^\infty([-1,1])}.
\end{equation}
For $u\in\widehat I$,
\[
 |T_k'(u)|
 =\left|\frac{k\sin(k\arccos u)}{\sqrt{1-u^2}}\right|
 \le\frac{k}{\sqrt{\delta(2-\delta)}}.
\]
Together with $|c_k|\le2M\rho^{-k}$, this ensures
uniform convergence of the differentiated Chebyshev
series on $\widehat I$. Consequently,
\begin{equation}\label{eq:cheb-tail-derivative}
\begin{aligned}
 \lVert (\widehat f-P_q)'\rVert_{L^\infty(\widehat I)}
 \le\frac{2M}{\sqrt{\delta(2-\delta)}}
       \sum_{k=q+1}^\infty k\rho^{-k}
 \le C_{\delta,\rho}M(q+1)\rho^{-q},
\end{aligned}
\end{equation}
where the last inequality follows from
\[
 \sum_{k=q+1}^\infty k\rho^{-k}
 =\rho^{-q}
 \left(\frac{q+1}{\rho-1}
       +\frac1{(\rho-1)^2}\right).
\]

Combining \eqref{eq:cheb-tail}--\eqref{eq:cheb-tail-derivative}
and restoring the affine derivative factor $1/\ell$ gives
\begin{equation}\label{eq:interpolation-q}
 \sup_I|f'|\le\frac{C_{\delta,\rho}}\ell
       \left[q\varepsilon+M(q+1)\rho^{-q}\right].
\end{equation}
Take
$q=\max\{1,\lceil\log(M/\varepsilon)/\log\rho\rceil\}$.
Then $M\rho^{-q}\le\varepsilon$, and
$q+1\le C_\rho(1+\log(M/\varepsilon))$.
This proves \eqref{eq:cheb} and its explicit scaling statement.
\end{proof}

\begin{rem}[Why one logarithm occurs]
\label{rem:log-sharp}
Using the global Markov polynomial inequality gives a $q^2$ loss,
whereas the interior Bernstein inequality \eqref{eq:bernstein} gives
only $q$.  The logarithm is natural for these hypotheses:
functions $f_q(z)=\varepsilon\sin(q(z-x_0))$ have
$\sup_{\mathbb R}|f_q|\le\varepsilon$, and on a fixed complex
neighbourhood their maximum grows as $\varepsilon\e^{O(q)}$,
while $|f_q'(x_0)|=q\varepsilon$.
This observation does not claim sharpness of the actual Hubbard error.
\end{rem}

\subsection{Improved magnetization bounds}
\label{subsec:FLY-bounds}
\begin{proof}[Proof of Theorem~\ref{thm:magnetization-Heis-scale-face}, parts (ii) and (iii)]
By Lemma~\ref{lem:limit}, the infinite-volume pressure and
magnetization limits exist at all nonzero real fields.  For
$J=J_0(U)=4t^2/U$ and $b=\beta J$, introduce the auxiliary
infinite-volume scaled pressures
\[
 G_{b,U}(z):=J^{-1}p_{b/J,U}^{\rm Hub}(Jz),\qquad
 F_b(z):=J^{-1}p_{b/J}^{\rm Heis}(J,Jz).
\]
By the chain rule,
\begin{equation}\label{eq:pressure-derivative-chain-rule}
 G_{b,U}'(x)=m_{b/J,U}^{\mathrm{Hub}}(Jx),\qquad
 F_b'(x)=m_{b/J}^{\mathrm{Heis}}(J,Jx).
\end{equation}
Thus the physical magnetization comparisons stated in the theorem
are equivalent to the auxiliary derivative estimates for
$G_{b,U}'-F_b'$ obtained below.
Fix $b\ge\ell_0$ and $U\ge U_R$.
By Lemma~\ref{lem:limit}, the infinite-volume scaled pressures
$G_{b,U}$ and $F_b$ are holomorphic on $\Omega$, and their derivatives are
the corresponding infinite-volume staggered magnetizations.  Put
\begin{equation}\label{eq:D}
 D_{b,U}(z)=G_{b,U}(z)-F_b(z).
\end{equation}
Equation \eqref{eq:scaled-pressure-U3-bound}, passed to the \emph{real} thermodynamic limit,
gives
\begin{equation}\label{eq:real-small}
 |D_{b,U}(s)|\le\varepsilon_U=\frac{A_R}{U^2}
 \qquad(0\le s\le R,\ b\ge\ell_0).
\end{equation}
By Lemma~\ref{lem:analytic} and locally uniform convergence of the
analytic grand-canonical and Heisenberg pressures,
\begin{equation}\label{eq:uniform-D-derivative}
 |D_{b,U}'(z)|\le B:=\frac{2\sqrt5}{3}
 \qquad(z\in\Omega),
\end{equation}
uniformly in $b\ge\ell_0$ and $U\ge U_R$.

\smallskip\noindent\emph{Part (ii): fixed positive scaled fields.}
Choose $R>c$ (for instance $R=4c+2a$), then fix a closed interval
$I\Subset J_1\Subset(0,R)$.
There is a $\rho>1$ sufficiently close to $1$, depending only on
$I,J_1$, such that the whole ellipse $E_\rho(J_1)$ lies in $\Omega$.
The cone $\Omega$ is convex.  Pick a real point $s_0\in J_1$ and
integrate \eqref{eq:uniform-D-derivative} along the segment from $s_0$
to any $z\in E_\rho(J_1)$:
\begin{equation}\label{eq:ellipse-bound-fixed}
 |D_{b,U}(z)|\le|D_{b,U}(s_0)|+B|z-s_0|
 \le\varepsilon_U+B\,\operatorname{diam}(E_\rho(J_1)).
\end{equation}
For sufficiently large $U$ we have $\varepsilon_U\le1$.
Thus the right side is bounded by a fixed $M_I\ge\varepsilon_U$,
independent of $b,U$.
Lemma~\ref{lem:cheb} applied to $D_{b,U}$ yields
\begin{equation}\label{eq:fixed-bound-raw}
 \sup_{x\in I}|D_{b,U}'(x)|
 \le C_{I,J_1,\rho}\,\varepsilon_U
      \left(1+\log\frac{M_I}{\varepsilon_U}\right).
\end{equation}
Since $\varepsilon_U=A_R/U^2$, this proves \eqref{eq:main-magnetization-U2log}
with constants independent of $b\ge\ell_0$.

\smallskip\noindent\emph{Part (iii): a uniform shrinking ellipse.}
Fix $R>0$ and consider $0<x\le R/2$.
Take
\begin{equation}\label{eq:shrinking-interval}
 J_x=[x/2,3x/2],\qquad I_x=\{x\},\qquad\rho=\frac32.
\end{equation}
The half-length of $J_x$ is $\ell_x=x/2$.
The Bernstein ellipse $E_{3/2}(J_x)$ has real semiaxis
$13x/24$ and imaginary semiaxis $5x/24$.

For $0<x\le R/2$, every $z\in E_{3/2}(J_x)$ satisfies
\begin{equation}\label{eq:shrinking-geometry}
 \frac{11x}{24}\le\Re z\le\frac{37x}{24}<R,
 \qquad
 |\Im z|\le\frac{5x}{24},
 \qquad
 \frac{|\Im z|}{\Re z}\le\frac5{11}<\frac12.
\end{equation}
Thus $E_{3/2}(J_x)\subset\Omega$, and its projection
onto the real axis lies in $(0,R)$.
In particular, the ellipse parameter $\rho=3/2$
is independent of $x$.

Since the ellipse is convex and contains $x$,
\eqref{eq:real-small} and
\eqref{eq:uniform-D-derivative} imply
\begin{equation}\label{eq:shrinking-ellipse-bound}
 \sup_{z\in E_{3/2}(J_x)}|D_{b,U}(z)|
 \le\varepsilon_U+\frac{13B}{24}x
 =:M_x.
\end{equation}

Under the affine rescaling
$u=2(s-x)/x$, the interval $J_x$ becomes
$[-1,1]$, and its midpoint $x$ corresponds to $u=0$.
Hence Lemma~\ref{lem:cheb}, with the fixed parameter
$\rho=3/2$, yields
\begin{equation}\label{eq:shrinking-raw}
 \begin{aligned}
 |D_{b,U}'(x)|
 \le\frac{C\varepsilon_U}{x}
       \left(1+\log\frac{M_x}{\varepsilon_U}\right)
 \le\frac{C'\varepsilon_U}{x}
       \left(1+\log\left(1+\frac{x}{\varepsilon_U}\right)\right).
 \end{aligned}
\end{equation}
The constants $C,C'$ are independent of $x$, $b$, and $U$.

If $x\ge\varepsilon_U$, then
$\log(1+x/\varepsilon_U)\le\log2+\log(x/\varepsilon_U)$.
This proves \eqref{eq:main-magnetization-shrinking}, with no $x$-, $b$-, or
$U$-dependence hidden in the constant.
Finally,
\begin{equation}\label{eq:physical-parameter-identity}
 x=\frac{Uh}{4t^2},\qquad
 \frac{\varepsilon_U}{x}=\frac{4t^2 A_R}{U^3h},\qquad
 \frac{x}{\varepsilon_U}=\frac{U^3h}{4t^2A_R}.
\end{equation}
If $Uh(U)\to0$ and $U^3h(U)\to\infty$, then
$x_U\to0$ and $x_U/\varepsilon_U\to\infty$;
hence all restrictions eventually hold and
\eqref{eq:main-magnetization-shrinking-physical} follows.
\end{proof}

\begin{proof}[Proof of Corollary~\ref{cor:positive-Hubbard-magnetization-face}]
Fix $R=1$ in the shrinking-field estimate of
Theorem~\ref{thm:magnetization-Heis-scale-face} and set
$J=4t^2/U$, $b=\beta J$, and $x_U=h(U)/J$.
The assumptions $U^3h(U)\to\infty$ and $Uh(U)\to0$
imply $x_U/\varepsilon_U\to\infty$ and $x_U\to0$.
Thus the improved comparison gives, uniformly in $b\ge\ell_0$,
\[
 \left|m_{b/J,U}^{\rm Hub}(h(U))
       -m_b^{\rm Heis}(1,x_U)\right|
 \le C\frac{1+\log(U^3h(U))}{U^3h(U)}
\]
for all sufficiently large $U$.
At each fixed $b\ge\ell_0$, the Heisenberg pressure is convex
in the staggered field, so its magnetization is nondecreasing.
By the definition of $m_*^{\rm Heis}(b)$,
\[
 m_b^{\rm Heis}(1,x_U)\ge m_*^{\rm Heis}(b)\ge m_0.
\]
Consequently,
\[
 m_{\beta,U}^{\rm Hub}(h(U))\ge
 m_0-C\frac{1+\log(U^3h(U))}{U^3h(U)}.
\]
Since the error tends to zero, the lower bound $m_0/2>0$
follows for all sufficiently large $U$, uniformly in $b\ge\ell_0$.
\end{proof}

\appendix

\section{Two-site Schrieffer--Wolff computation}
\label{app:two-site-revision}

The following lemma and its calculation have been transferred from
the two-site Hubbard Hamiltonian. Their algebraic content does not
use the infinite LS/SW iteration.

\begin{lem}[Two-site computation on a bond]
\label{lem:two-site-computation}
Assume $|h|\le h_0$, $U>0$, and
$U\ge 2\pi h_0$.  Fix a nearest-neighbour bond
\(e=\{x,y\}\in\mathscr B_\Lambda\), and label its endpoints so that
$
  \eta_x=+1
 $ 
  and 
$
  \eta_y=-1.
$
On \(P\cH^{\rm hf}_\Lambda\), with
$
  P=\mathbbm 1_{\{D_\Lambda=0\}}\restriction_{\cH_\Lambda^{\rm hf}},
$
define
\[
  X_{xy}(h)
  :=
  P\frac12
  \left(
  \ad_{
    \mathcal I_h((T_{e})^{\rm off})
    }(
    (T_e)^{\rm off}
    )
  \right)^{\rm diag}
  P.
\]
After identifying the singly occupied sector with the spin
space, \(X_{xy}(h)\) acts nontrivially only on the spin factors
at \(x\) and \(y\). More precisely, under the spin
identification
$
  \mathcal U_\Lambda:P\cH^{\rm hf}_\Lambda\to\cH^{\rm spin}_\Lambda,
$
one has
\begin{equation}\label{eq:two-site-computation-formula}
  \mathcal U_\Lambda X_{xy}(h)\mathcal U_\Lambda^\ast
  =
  -J_{xy}(h)B_{xy}
  +
  b_{xy}(h)M_e^{\rm spin},
\end{equation}
where
$
  M_e^{\rm spin}:=\eta_xS_x^{(3)}+\eta_yS_y^{(3)}
  =
  S_x^{(3)}-S_y^{(3)}
$
on \(\cH_\Lambda^{\rm spin}\), and
\begin{equation}\label{eq:two-site-coeffs}
  J_{xy}(h):=\frac{4t^2U}{U^2-h^2},
  \qquad
  b_{xy}(h):=\frac{2ht^2}{U^2-h^2}.
\end{equation}
In particular,
\[
  J_{xy}(0)=\frac{4t^2}{U},
  \qquad
  b_{xy}(0)=0.
\]
\end{lem}
\begin{proof}
By the support-preserving property of $\mathcal I_h$, the operators
\[
 \begin{aligned}
 S_{1,e}:=\mathcal I_h((T_e)^{\rm off}),\qquad
 C_e(h):=\frac12\bigl[S_{1,e},(T_e)^{\rm off}\bigr]^{\rm diag}
 \end{aligned}
\]
are even and supported on $e=\{x,y\}$.  The operator
$X_{xy}(h)=P C_e(h)P$ contains the global projection $P$,
so it is not literally a bond-supported operator on the full
fermionic Fock space.  However, on the singly occupied sector,
the spin compression of the even bond operator $C_e(h)$ acts
only on the spin factors at $x,y$ and as the identity elsewhere.
This can also be seen from the on-site matrix units in
Lemma~\ref{lem:defect-spin-compression}. Thus the following
two-site calculation determines the action of the operator
$
 \mathcal U_\Lambda X_{xy}(h)\mathcal U_\Lambda^*
$
on the full spin space.

Write
$
  H_0(h):=UD_\Lambda-hM_\Lambda.
$
On operators supported in \(\{x,y\}\), the commutator with \(H_0(h)\)
reduces to the commutator with the two-site operator
\[
  H_{0,e}(h):=UD_e-hM_e,
  \qquad
  D_e:=n_{x\uparrow}n_{x\downarrow}
      +n_{y\uparrow}n_{y\downarrow},\qquad
      M_e:=S_x^{(3)}-S_y^{(3)}.
\]
Indeed, the part of \(H_0(h)\) supported outside \(e\) commutes with
\((T_e)^{\rm off}\) and with the local inverse generated from it.  Thus the
calculation may be performed on the two-site \(N=2\) subspace.

\smallskip
\noindent
\emph{Two-site basis and fermionic sign convention.}
We use
\[
  |\sigma\tau\rangle
  :=
  c_{x\sigma}^\ast c_{y\tau}^\ast|0\rangle,
  \qquad
  \sigma,\tau\in\{\uparrow,\downarrow\},
\]
and
\[
  |d_x\rangle
  :=
  c_{x\uparrow}^\ast c_{x\downarrow}^\ast|0\rangle,
  \qquad
  |d_y\rangle
  :=
  c_{y\uparrow}^\ast c_{y\downarrow}^\ast|0\rangle.
\]
The local two-site singly occupied subspace is spanned by
$
  |\!\uparrow\uparrow \rangle,\,
  |\!\downarrow\uparrow\rangle,\, 
  |\!\uparrow\downarrow\rangle,\, 
  |\!\downarrow\downarrow\rangle,
$
and its orthogonal complement in the local
two-site \(N_e=2\) sector is
$
  \operatorname{span}\{|d_x\rangle,|d_y\rangle\}.
$
The signs below are with respect to this convention.

On the local singly occupied subspace, one has \(D_e=0\), hence
\(H_{0,e}(h)=-hM_e\).  On
\(\operatorname{span}\{|d_x\rangle,|d_y\rangle\}\), one has \(D_e=1\)
and \(M_e=0\), hence \(H_{0,e}(h)=U\).

Moreover, \((T_e)^{\rm off}\) maps the local singly occupied subspace into
\(\operatorname{span}\{|d_x\rangle,|d_y\rangle\}\) and annihilates
\(|\!\uparrow\uparrow\rangle\) and \(|\!\downarrow\downarrow\rangle\).  A
direct CAR computation gives
\begin{equation}\label{eq:two-site-hop-action}
  (T_e)^{\rm off}|\!\uparrow\downarrow\rangle
  =
  -t(|d_x\rangle+|d_y\rangle),
  \qquad
  (T_e)^{\rm off}|\!\downarrow\uparrow\rangle
  =
  t(|d_x\rangle+|d_y\rangle),
\end{equation}
and
$
  (T_e)^{\rm off}|\!\uparrow\uparrow\rangle
  =
  (T_e)^{\rm off}|\!\downarrow\downarrow\rangle
  =
  0.
$
The reverse actions are
\begin{equation}\label{eq:two-site-hop-action-reverse}
  (T_e)^{\rm off}|d_x\rangle
  =
  t\bigl(|\!\downarrow\uparrow\rangle-|\!\uparrow\downarrow\rangle\bigr),
  \qquad
  (T_e)^{\rm off}|d_y\rangle
  =
  t\bigl(|\!\downarrow\uparrow\rangle-|\!\uparrow\downarrow\rangle\bigr).
\end{equation}

Since \(H_{0,e}(h)\) is diagonal in the above basis, the defining
relation
$
  \ad_{H_{0,e}(h)}(\mathcal I_h(A))=A
$
gives, for eigenvectors \(|\phi\rangle,|\psi\rangle\) with eigenvalues
\(E_\phi,E_\psi\),
\begin{equation}\label{eq:energy-difference-division}
  \langle\psi|\mathcal I_h(A)|\phi\rangle
  =
  \frac{\langle\psi|A|\phi\rangle}{E_\psi-E_\phi},
  \qquad E_\psi\neq E_\phi.
\end{equation}
Applying this to \(A=(T_e)^{\rm off}\), we use
$
  M_e|\!\uparrow\downarrow\rangle=|\!\uparrow\downarrow\rangle
$
and 
$
  M_e|\!\downarrow\uparrow\rangle=-|\!\downarrow\uparrow\rangle.
$
Thus
\[
  H_{0,e}(h)|\!\uparrow\downarrow\rangle
  =
  -h|\!\uparrow\downarrow\rangle,
  \qquad
  H_{0,e}(h)|\!\downarrow\uparrow\rangle
  =
  h|\!\downarrow\uparrow\rangle,
\]
whereas
\[
  H_{0,e}(h)|d_x\rangle=U|d_x\rangle,
  \qquad
  H_{0,e}(h)|d_y\rangle=U|d_y\rangle.
\]
Hence the relevant energy differences are \(U+h\) and \(U-h\), and
\eqref{eq:two-site-hop-action}--\eqref{eq:energy-difference-division}
give
\begin{equation}\label{eq:two-site-Ih-action}
  \mathcal I_h((T_e)^{\rm off})|\!\uparrow\downarrow\rangle
  =
  -\frac{t}{U+h}(|d_x\rangle+|d_y\rangle),
  \qquad
  \mathcal I_h((T_e)^{\rm off})|\!\downarrow\uparrow\rangle
  =
  \frac{t}{U-h}(|d_x\rangle+|d_y\rangle).
\end{equation}
Similarly, using \eqref{eq:two-site-hop-action-reverse},
\begin{equation}\label{eq:two-site-Ih-action-reverse}
  \mathcal I_h((T_e)^{\rm off})|d_x\rangle
  =
  \frac{t}{U+h}|\!\uparrow\downarrow\rangle
  -
  \frac{t}{U-h}|\!\downarrow\uparrow\rangle,
\end{equation}
and the same formula holds with \(|d_x\rangle\) replaced by
\(|d_y\rangle\).  Also,
$
  \mathcal I_h((T_e)^{\rm off})|\!\uparrow\uparrow\rangle
  =
  \mathcal I_h((T_e)^{\rm off})|\!\downarrow\downarrow\rangle
  =
  0,
$
because \((T_e)^{\rm off}\) has no off-diagonal matrix elements from these
two states.

Using \eqref{eq:two-site-hop-action}--\eqref{eq:two-site-Ih-action} and
\(PX^{\rm diag}P=PXP\), we compute the restriction of
\[
  X_{xy}(h)
  =
  P\frac12\ad_{\mathcal I_h((T_e)^{\rm off})}((T_e)^{\rm off})P
\]
to the local singly occupied subspace.  In the ordered basis
$
  (|\!\uparrow\uparrow\rangle,
   |\!\downarrow\uparrow\rangle,
   |\!\uparrow\downarrow\rangle,
   |\!\downarrow\downarrow\rangle),
$
this yields
\[
  \mathcal U_\Lambda X_{xy}(h)\mathcal U_\Lambda^\ast
  =
  \begin{pmatrix}
    0 & 0 & 0 & 0\\
    0 & -\frac{2t^2}{U-h} &
        \frac{2Ut^2}{(U-h)(U+h)} & 0\\
    0 & \frac{2Ut^2}{(U-h)(U+h)} &
        -\frac{2t^2}{U+h} & 0\\
    0 & 0 & 0 & 0
  \end{pmatrix}.
\]
On the other hand, in the same basis,
\[
  B_{xy}
  =
  \begin{pmatrix}
    0 & 0 & 0 & 0\\
    0 & \frac12 & -\frac12 & 0\\
    0 & -\frac12 & \frac12 & 0\\
    0 & 0 & 0 & 0
  \end{pmatrix},
\qquad
  M_e^{\rm spin}
  =
  \begin{pmatrix}
    0 & 0 & 0 & 0\\
    0 & -1 & 0 & 0\\
    0 & 0 & 1 & 0\\
    0 & 0 & 0 & 0
  \end{pmatrix}.
\]
Comparing the off-diagonal entries gives
$
  J_{xy}(h)
  =
  \frac{4t^2U}{U^2-h^2}.
$
Comparing, for instance, the \((2,2)\)-entries gives
$
  -\frac{J_{xy}(h)}{2}-b_{xy}(h)
  =
  -\frac{2t^2}{U-h},
$
and hence
$
  b_{xy}(h)
  =
  \frac{2ht^2}{U^2-h^2}.
$
Finally, since both \(B_{xy}\) and \(M_e^{\rm spin}\) have vanishing
\((1,1)\)- and \((4,4)\)-entries, no multiple of the identity appears.
This proves \eqref{eq:two-site-computation-formula}.
\end{proof}

\section{Proof of the local spin-compression lemma}
\label{app:defect-sector-new}

We prove Lemma~\ref{lem:defect-spin-compression},
which identifies the compression of an even local
operator with a spin operator on the singly occupied
sites. The main point is to verify that this
identification is independent of the holon--doublon
configuration outside the support of the operator.

\begin{proof}[Proof of Lemma~\ref{lem:defect-spin-compression}]
Fix a nonempty subset $X\subset\Lambda$.
For each $x\in X$ and
$\sigma,\tau\in\{\uparrow,\downarrow\}$, define
\[
 E_{x,\sigma\tau}
 :=
 (1-q_x)c_{x\sigma}^*c_{x\tau}(1-q_x).
\]

These operators vanish on the empty and doubly
occupied states at $x$.
Let $|0\rangle_x$ denote the local vacuum and set
$|\sigma\rangle_x:=c_{x\sigma}^*|0\rangle_x$
for $\sigma\in\{\uparrow,\downarrow\}$.
On the singly occupied subspace, the operators
$E_{x,\sigma\tau}$ act as matrix units:
\[
 E_{x,\sigma\tau}|\rho\rangle_x
 =\delta_{\tau\rho}|\sigma\rangle_x,
 \qquad \rho\in\{\uparrow,\downarrow\}.
\]
In particular,
\begin{equation}\label{eq:defect-matrix-unit-relations}
 E_{x,\sigma\tau}E_{x,\rho\lambda}
 =\delta_{\tau\rho}E_{x,\sigma\lambda},
 \qquad
 E_{x,\sigma\tau}^*=E_{x,\tau\sigma}.
\end{equation}

Each $E_{x,\sigma\tau}$ is even and supported
at the single site $x$. Hence the CAR imply
\begin{equation}\label{eq:defect-matrix-units-commute}
 [E_{x,\sigma\tau},E_{y,\rho\lambda}]=0,
 \qquad x\ne y.
\end{equation}

For
$\boldsymbol\sigma,\boldsymbol\tau
\in\{\uparrow,\downarrow\}^X$, set
\[
 E_{X,\boldsymbol\sigma\boldsymbol\tau}
 :=
 \prod_{x\in X}E_{x,\sigma_x\tau_x}.
\]

The product is independent of the ordering by
\eqref{eq:defect-matrix-units-commute}.

Let $\cH_X$ denote the local fermionic Fock
space over $X$, with vacuum $|0\rangle_X$.
For $\boldsymbol\sigma=(\sigma_x)_{x\in X}$, define
\[
 |\boldsymbol\sigma\rangle_X^{\rm F}
 :=
 \left(\prod_{x\in X}^{\longrightarrow}
 c_{x\sigma_x}^*\right)|0\rangle_X,
\]
where the product follows the fixed site ordering.
These vectors form an orthonormal basis of the
singly occupied subspace $P_X\mathcal F_X$.

We identify this subspace with
$\cH_X^{\rm spin}=(\mathbb C^2)^{\otimes X}$
by mapping $|\boldsymbol\sigma\rangle_X^{\rm F}$
to the corresponding spin basis vector
$\bigotimes_{x\in X}|\sigma_x\rangle$.
On $\cH_X$, we then have
\[
 E_{X,\boldsymbol\sigma\boldsymbol\tau}
 =
 |\boldsymbol\sigma\rangle_X^{\rm F}
 {}_X^{\rm F}\langle\boldsymbol\tau|.
\]
Thus, under this identification,
$E_{X,\boldsymbol\sigma\boldsymbol\tau}$
corresponds to the matrix unit
$|\boldsymbol\sigma\rangle_X
{}_X\langle\boldsymbol\tau|$
on $\cH_X^{\rm spin}$.

Let $A_X\in\mathfrak A_X^{\rm even}$ and
write
\[
 a_{\boldsymbol\sigma\boldsymbol\tau}
 :=
 {}_X\langle\boldsymbol\sigma|
 \phi_X|\boldsymbol\tau\rangle_X,
\]
where $\phi_X$ is the local spin compression
defined in Subsection~\ref{subsec:defect-spin-compression}.
Since $P_XA_XP_X$ vanishes outside the singly
occupied local subspace, its matrix-unit
expansion is
\begin{equation}\label{eq:defect-compression-expansion}
 P_XA_XP_X
 =
 \sum_{\boldsymbol\sigma,\boldsymbol\tau
       \in\{\uparrow,\downarrow\}^X}
 a_{\boldsymbol\sigma\boldsymbol\tau}
 E_{X,\boldsymbol\sigma\boldsymbol\tau}.
\end{equation}
This is an identity in the local CAR algebra
$\mathfrak A_X$, and hence also holds when
the operators are embedded into the full
fermionic algebra.

We now restrict this identity to a fixed-defect
sector $\cH_{V_0,V_2}$.
Recall that
$\mathcal D=V_0\cup V_2$ and
$R=\Lambda\setminus\mathcal D$.

If $X\cap\mathcal D\ne\varnothing$, at least
one site of $X$ is empty or doubly occupied.
Since $P_X$ projects onto single occupation
at every site of $X$, we have
$
 P_X\restriction{\cH_{V_0,V_2}}=0,
$
and therefore
\begin{equation}\label{eq:defect-compression-zero}
 (P_XA_XP_X)\restriction_{\cH_{V_0,V_2}}=0.
\end{equation}

Suppose next that $X\subset R$.
Under the unitary identification
$\mathcal U_{V_0,V_2}$ defined in
\eqref{eq:defect-sector-spin-identification},
the single-site matrix units act as the
corresponding spin matrix units.

To see that no sign depending on the exterior
charge configuration appears, consider an
ordered fermionic basis vector.
For a singly occupied site $x$, the annihilation
operator $c_{x\tau}$ produces the fermionic sign
determined by the occupations preceding $x$
in the fixed ordering. The subsequent creation
operator $c_{x\sigma}^*$ produces the same sign,
so the two signs cancel.
Consequently, $E_{x,\sigma\tau}$ acts as the
ordinary on-site matrix unit, independently
of the occupations at the other sites.

Since the matrix units at distinct sites commute
by \eqref{eq:defect-matrix-units-commute}, it follows
that
\[
 \begin{aligned}
 \mathcal U_{V_0,V_2}^*
 E_{X,\boldsymbol\sigma\boldsymbol\tau}
 \mathcal U_{V_0,V_2}
 =
 \left(
 |\boldsymbol\sigma\rangle_X
 {}_X\langle\boldsymbol\tau|
 \right)^{(R)}.
 \end{aligned}
\]
Substituting this into
\eqref{eq:defect-compression-expansion}, we obtain
\[
 \begin{aligned}
 \mathcal U_{V_0,V_2}^*
 (P_XA_XP_X)
 \mathcal U_{V_0,V_2}
 =
 \sum_{\boldsymbol\sigma,\boldsymbol\tau}
 a_{\boldsymbol\sigma\boldsymbol\tau}
 \left(
 |\boldsymbol\sigma\rangle_X
 {}_X\langle\boldsymbol\tau|
 \right)^{(R)}
 =\phi_X^{(R)}.
 \end{aligned}
\]
Together with \eqref{eq:defect-compression-zero},
this proves \eqref{eq:defect-compression-sector}.
In particular, the compressed operator is
independent of the holon--doublon assignments
outside $X$.

Finally, we prove \eqref{eq:defect-compression-norm}.
Let
$
 \mathcal V_X:\cH_X^{\rm spin}
 \longrightarrow P_X\cH_X
$
be the unitary identification defined by the
ordered fermionic basis vectors introduced above.
By the definition of the local spin compression,
\[
 \phi_X
 =
 \mathcal V_X^*
 \bigl(P_XA_XP_X\restriction_{P_X\cH_X}\bigr)
 \mathcal V_X.
\]
Since $\mathcal V_X$ is unitary, we have
\[
 \|\phi_X\|
 =
 \bigl\|P_XA_XP_X
       \restriction_{P_X\cH_X}\bigr\|.
\]
Moreover, $P_XA_XP_X$ vanishes on
$(P_X\cH_X)^\perp$. Hence
\[
 \bigl\|P_XA_XP_X
       \restriction_{P_X\cH_X}\bigr\|
 =
 \|P_XA_XP_X\|_{\mathcal B(\cH_X)}.
\]
As $P_X$ is an orthogonal projection,
\[
 \|\phi_X\|
 =\|P_XA_XP_X\|_{\mathcal B(\cH_X)}
 \le\|A_X\|.
\]
This proves \eqref{eq:defect-compression-norm}
and completes the proof.

\end{proof}

\subsection*{Acknowledgements}
This work was supported by JSPS KAKENHI Grant Number 23H01086.
The author thanks Hal Tasaki and Stefan Teufel for
their valuable comments during the One World IAMP
Mathematical Physics Seminar.
The author is also grateful to Hal Tasaki for
bringing the recent work of Rayudu et al.\
\cite{RMZT26} to his attention.
During the revision of an earlier version of this manuscript, the author used OpenAI's ChatGPT to explore and examine mathematical arguments and to improve the exposition. The author takes full responsibility for the mathematical content of the paper.

\subsection*{Declarations}
The author declares no conflicts of interest relevant to this article.
Data sharing is not applicable, as no datasets were generated or analyzed.

\bibliographystyle{abbrvurl}
\bibliography{Reference}

@article{ArakiMoriya2003,
  author  = {Araki, Huzihiro and Moriya, Hajime},
  title   = {Equilibrium Statistical Mechanics
             of Fermion Lattice Systems},
  journal = {Reviews in Mathematical Physics},
  volume  = {15},
  number  = {2},
  pages   = {93--198},
  year    = {2003},
  doi     = {10.1142/S0129055X03001606}
}

@misc{BachmannGollerPorta2025,
  author        = {Bachmann, Sven and Goller, Leonardo
                   and Porta, Marcello},
  title         = {Anyons in the {$\pi$}-Flux Phase
                   of Fermionic Matter Coupled to a
                   {$\mathbb{Z}_2$}-Gauge Field},
  year          = {2025},
  eprint        = {2508.21502},
  archivePrefix = {arXiv},
  primaryClass  = {math-ph},
  url           = {https://arxiv.org/abs/2508.21502}
}

@article{Baker1968,
  author  = {Baker, George A., Jr.},
  title   = {Some Rigorous Inequalities Satisfied
             by the Ferromagnetic Ising Model
             in a Magnetic Field},
  journal = {Phys. Rev. Lett.},
  volume  = {20},
  pages   = {990--992},
  year    = {1968},
  doi     = {10.1103/PhysRevLett.20.990}
}

@incollection{BjornbergUeltschi2022,
  author    = {Bj{\"o}rnberg, Jakob E. and Ueltschi, Daniel},
  title     = {Reflection Positivity and Infrared Bounds for
               Quantum Spin Systems},
  booktitle = {The Physics and Mathematics of Elliott Lieb},
  editor    = {Frank, Rupert L. and Laptev, Ari and
               Lewin, Mathieu and Seiringer, Robert},
  volume    = {1},
  pages     = {77--108},
  publisher = {EMS Press},
  year      = {2022},
  doi       = {10.4171/90-1/5}
}

@article{BorgsJedrzejewskiKotecky1996,
  author  = {Borgs, Christian and Jedrzejewski, Janusz and Koteck{\'y}, Roman},
  title   = {{The staggered charge-order phase of the extended {Hubbard} model in the atomic limit}},
  journal = {Journal of Physics A: Mathematical and General},
  volume  = {29},
  number  = {4},
  pages   = {733--747},
  year    = {1996},
  doi     = {10.1088/0305-4470/29/4/005}
}

@article{BorgsKotecky2000,
  author  = {Borgs, Christian and Koteck{\'y}, Roman},
  title   = {{Low temperature phase diagrams of fermionic lattice systems}},
  journal = {Communications in Mathematical Physics},
  volume  = {208},
  pages   = {575--604},
  year    = {2000},
  doi     = {10.1007/s002200050002}
}

@article{BorgsKoteckyUeltschi1996,
  author  = {Borgs, Christian and Koteck{\'y}, Roman and Ueltschi, Daniel},
  title   = {{Low temperature phase diagrams for quantum perturbations of classical spin systems}},
  journal = {Communications in Mathematical Physics},
  volume  = {181},
  pages   = {409--446},
  year    = {1996},
  doi     = {10.1007/BF02101010}
}

@book{BratteliRobinson2,
  author    = {Bratteli, Ola and Robinson, Derek W.},
  title     = {Operator Algebras and Quantum Statistical Mechanics 2:
               Equilibrium States. Models in Quantum Statistical Mechanics},
  publisher = {Springer},
  address   = {Berlin},
  edition   = {2},
  year      = {1997},
  doi       = {10.1007/978-3-662-03444-6}
}

@article{BravyiDiVincenzoLoss2011,
  author  = {Bravyi, Sergey and DiVincenzo, David P. and Loss, Daniel},
  title   = {{Schrieffer--Wolff transformation for quantum many-body systems}},
  journal = {Annals of Physics},
  volume  = {326},
  number  = {10},
  pages   = {2793--2826},
  year    = {2011},
  doi     = {10.1016/j.aop.2011.06.004}
}

@article{CorreggiGiulianiSeiringer2015,
  author  = {Correggi, Michele and Giuliani, Alessandro and Seiringer, Robert},
  title   = {{Validity of the spin-wave approximation for the free energy of the {Heisenberg} ferromagnet}},
  journal = {Communications in Mathematical Physics},
  volume  = {339},
  number  = {1},
  pages   = {279--307},
  year    = {2015},
  doi     = {10.1007/s00220-015-2402-0}
}

@article{DattaFernandezFrohlich1996,
  author  = {Datta, Nilanjana and Fern{\'a}ndez, Roberto and Fr{\"o}hlich, J{\"u}rg},
  title   = {{Low-temperature phase diagrams of quantum lattice systems. {I}. Stability for quantum perturbations of classical systems with finitely many ground states}},
  journal = {Journal of Statistical Physics},
  volume  = {84},
  pages   = {455--534},
  year    = {1996},
  doi     = {10.1007/BF02179651}
}

@article{DattaFernandezFrohlich1999,
  author  = {Datta, Nilanjana and Fern{\'a}ndez, Roberto and Fr{\"o}hlich, J{\"u}rg},
  title   = {{Effective Hamiltonians and phase diagrams for tight-binding models}},
  journal = {Journal of Statistical Physics},
  volume  = {96},
  pages   = {545--611},
  year    = {1999},
  doi     = {10.1023/A:1004594122474}
}

@article{DelVecchioFrohlichPizzoRossi2021,
  author  = {Del Vecchio, Simone and Fr{\"o}hlich, J{\"u}rg and Pizzo, Alessandro and Rossi, Stefano},
  title   = {{Lie--Schwinger Block-Diagonalization and Gapped Quantum Chains with Unbounded Interactions}},
  journal = {Communications in Mathematical Physics},
  volume  = {381},
  pages   = {1115--1152},
  year    = {2021},
  doi     = {10.1007/s00220-020-03878-y}
}

@article{DysonLiebSimon1978,
  author  = {Dyson, Freeman J. and Lieb, Elliott H. and Simon, Barry},
  title   = {{Phase transitions in quantum spin systems with isotropic and nonisotropic interactions}},
  journal = {Journal of Statistical Physics},
  volume  = {18},
  pages   = {335--383},
  year    = {1978},
  doi     = {10.1007/BF01106729}
}

@article{FrohlichRodriguez2012,
  author  = {Fr{\"o}hlich, J{\"u}rg and Rodriguez, Pierre-Fran{\c{c}}ois},
  title   = {{Some Applications of the Lee--Yang Theorem}},
  journal = {Journal of Mathematical Physics},
  volume  = {53},
  number  = {9},
  pages   = {095218},
  year    = {2012},
  doi     = {10.1063/1.4749391},
  url     = {https://doi.org/10.1063/1.4749391}
}

@article{GiulianiMastropietro2010,
  author  = {Giuliani, Alessandro and Mastropietro, Vieri},
  title   = {{The two-dimensional {Hubbard} model on the honeycomb lattice}},
  journal = {Communications in Mathematical Physics},
  volume  = {293},
  pages   = {301--346},
  year    = {2010},
  doi     = {10.1007/s00220-009-0910-5}
}

@article{HarrisLange1967,
  author  = {Harris, A. Brooks and Lange, Robert V.},
  title   = {Single-particle excitations in narrow energy bands},
  journal = {Physical Review},
  volume  = {157},
  number  = {2},
  pages   = {295--314},
  year    = {1967},
  doi     = {10.1103/PhysRev.157.295}
}

@article{Hubbard1963,
  author  = {Hubbard, J.},
  title   = {{Electron correlations in narrow energy bands}},
  journal = {Proceedings of the Royal Society of London. Series A. Mathematical and Physical Sciences},
  volume  = {276},
  number  = {1365},
  pages   = {238--257},
  year    = {1963},
  doi     = {10.1098/rspa.1963.0204}
}

@article{Jedrzejewski1994,
  author  = {Jedrzejewski, Janusz},
  title   = {Phase diagrams of extended {Hubbard} models in the atomic limit},
  journal = {Physica A: Statistical Mechanics and its Applications},
  volume  = {205},
  number  = {4},
  pages   = {702--717},
  year    = {1994},
  doi     = {10.1016/0378-4371(94)90231-3}
}

@article{KuboKishi1988,
  author  = {Kubo, Kenn and Kishi, Tatsuya},
  title   = {Existence of Long-Range Order in the {XXZ} Model},
  journal = {Phys. Rev. Lett.},
  volume  = {61},
  pages   = {2585--2587},
  year    = {1988},
  doi     = {10.1103/PhysRevLett.61.2585}
}

@article{KomaTasaki1993,
  author  = {Koma, Tohru and Tasaki, Hal},
  title   = {{Symmetry Breaking in Heisenberg Antiferromagnets}},
  journal = {Commun. Math. Phys.},
  volume  = {158},
  pages   = {191--214},
  year    = {1993},
  doi     = {10.1007/BF02097237}
}

@article{Miyao2019Ferromagnetism,
  author  = {Miyao, Tadahiro},
  title   = {Stability of Ferromagnetism in Many-Electron Systems},
  journal = {Journal of Statistical Physics},
  volume  = {176},
  number  = {5},
  pages   = {1211--1271},
  year    = {2019},
  doi     = {10.1007/s10955-019-02341-0},
  url     = {https://doi.org/10.1007/s10955-019-02341-0}
}

@article{LeeYang1952,
  author  = {Lee, T. D. and Yang, C. N.},
  title   = {{Statistical Theory of Equations of State
             and Phase Transitions. II. Lattice Gas
             and Ising Model}},
  journal = {Phys. Rev.},
  volume  = {87},
  pages   = {410--419},
  year    = {1952},
  doi     = {10.1103/PhysRev.87.410}
}

@article{LiebMattis1962,
  author  = {Lieb, Elliott and Mattis, Daniel},
  title   = {Ordering Energy Levels of Interacting Spin Systems},
  journal = {J. Math. Phys.},
  volume  = {3},
  pages   = {749--751},
  year    = {1962},
  doi     = {10.1063/1.1724276}
}

@article{LiebWu1968,
  author  = {Lieb, Elliott H. and Wu, F. Y.},
  title   = {Absence of {Mott} transition in an exact solution of the short-range, one-band model in one dimension},
  journal = {Physical Review Letters},
  volume  = {20},
  number  = {25},
  pages   = {1445--1448},
  year    = {1968},
  doi     = {10.1103/PhysRevLett.20.1445}
}

@article{Lieb1989,
  author  = {Lieb, Elliott H.},
  title   = {Two theorems on the {Hubbard} model},
  journal = {Physical Review Letters},
  volume  = {62},
  number  = {10},
  pages   = {1201--1204},
  year    = {1989},
  doi     = {10.1103/PhysRevLett.62.1201}
}

@article{MacDonaldGirvinYoshioka1988,
  author  = {MacDonald, A. H. and Girvin, S. M. and Yoshioka, D.},
  title   = {{$t/U$} expansion for the {Hubbard} model},
  journal = {Physical Review B},
  volume  = {37},
  number  = {16},
  pages   = {9753--9756},
  year    = {1988},
  doi     = {10.1103/PhysRevB.37.9753}
}

@article{MielkeTasaki1993,
  author  = {Mielke, Andreas and Tasaki, Hal},
  title   = {{Ferromagnetism in the {Hubbard} model. Examples from models with degenerate single-electron ground states}},
  journal = {Communications in Mathematical Physics},
  volume  = {158},
  number  = {2},
  pages   = {341--371},
  year    = {1993},
  doi     = {10.1007/BF02108079}
}

@article{Nagaoka1966,
  author  = {Nagaoka, Yosuke},
  title   = {Ferromagnetism in a narrow, almost half-filled \(s\) band},
  journal = {Physical Review},
  volume  = {147},
  number  = {1},
  pages   = {392--405},
  year    = {1966},
  doi     = {10.1103/PhysRev.147.392}
}

@misc{RMZT26,
  author        = {Rayudu, C. and Misawa, T. and Zhao, A.
                   and Takahashi, J.},
  title         = {{Lee--Yang theorem for fermions}},
  year          = {2026},
  eprint        = {2609.23942},
  archivePrefix = {arXiv},
  note          = {arXiv:2609.23942v1},
  url           = {https://arxiv.org/abs/2609.23942}
}

@article{SchriefferWolff1966,
  author  = {Schrieffer, J. R. and Wolff, P. A.},
  title   = {{Relation between the {Anderson} and {Kondo} Hamiltonians}},
  journal = {Physical Review},
  volume  = {149},
  number  = {2},
  pages   = {491--492},
  year    = {1966},
  doi     = {10.1103/PhysRev.149.491}
}

@article{Sroka2022,
  author  = {Sroka, Grzegorz},
  title   = {{Constants in {Markov's} and {Bernstein}
             Inequality on a Finite Interval in
             {$\mathbb{R}$}}},
  journal = {Anal. Math. Phys.},
  volume  = {12},
  pages   = {128},
  year    = {2022},
  doi     = {10.1007/s13324-022-00711-8}
}

@article{SuzukiFisher1971,
  author  = {Suzuki, Masuo and Fisher, Michael E.},
  title   = {{Zeros of the Partition Function for the Heisenberg,
             Ferroelectric, and General Ising Models}},
  journal = {Journal of Mathematical Physics},
  volume  = {12},
  number  = {2},
  pages   = {235--246},
  year    = {1971},
  doi     = {10.1063/1.1665583},
  url     = {https://doi.org/10.1063/1.1665583}
}

@article{Tasaki1998,
  author  = {Tasaki, Hal},
  title   = {{From {Nagaoka}'s ferromagnetism to flat-band ferromagnetism and beyond: An introduction to ferromagnetism in the {Hubbard} model}},
  journal = {Progress of Theoretical Physics},
  volume  = {99},
  number  = {4},
  pages   = {489--548},
  year    = {1998},
  doi     = {10.1143/PTP.99.489}
}

@article{Teufel2020,
  author  = {Teufel, Stefan},
  title   = {{Non-equilibrium almost-stationary states and linear response for gapped quantum systems}},
  journal = {Communications in Mathematical Physics},
  volume  = {373},
  pages   = {621--653},
  year    = {2020},
  doi     = {10.1007/s00220-019-03407-6}
}

@book{Trefethen2019,
  author    = {Trefethen, Lloyd N.},
  title     = {Approximation Theory and Approximation Practice},
  edition   = {Extended Edition},
  publisher = {Society for Industrial and Applied Mathematics},
  address   = {Philadelphia},
  year      = {2019},
  series    = {Other Titles in Applied Mathematics},
  volume    = {164},
  doi       = {10.1137/1.9781611975949},
  isbn      = {978-1-61197-593-2}
}

@article{WesleMarcelliMiyaoMonacoTeufel2025,
  author  = {Wesle, Marius and Marcelli, Giovanna and Miyao, Tadahiro and Monaco, Domenico and Teufel, Stefan},
  title   = {{Near linearity of the macroscopic Hall current response in infinitely extended gapped fermion systems}},
  journal = {Communications in Mathematical Physics},
  volume  = {406},
  pages   = {199},
  year    = {2025},
  doi     = {10.1007/s00220-025-05361-y}
}

\end{document}